\documentclass[english,a4paper,12pt]{article}
\usepackage[utf8]{inputenc}
\usepackage[T1]{fontenc}
\usepackage{titling}

\usepackage{amsmath, amssymb}
\usepackage{pgf}
\usepackage{amsfonts}
\usepackage{dsfont}
\usepackage{mathrsfs}
\usepackage{mathabx}
\usepackage{stmaryrd}
\usepackage{makeidx}
\usepackage{amsbsy}
\usepackage{amsthm}
\usepackage{color}
\usepackage{fullpage}
\usepackage{enumitem}
\usepackage[vcentermath]{youngtab}
\usepackage{marginnote} 
\usepackage{mdframed}	
\newmdenv[topline=false, bottomline=false, skipabove=\topsep, skipbelow=\topsep]{siderules}
\usepackage{tikz}
\usetikzlibrary{decorations.text,calc,arrows.meta}
\usepackage[vcentermath]{youngtab}
\usepackage{tkz-euclide}
\usepackage{pgfplots}
\usetikzlibrary{decorations.shapes}
\usetikzlibrary{positioning}
\usetikzlibrary{patterns}
\usetikzlibrary{shapes.geometric}
\usetikzlibrary{decorations.pathmorphing}
\usetikzlibrary{arrows.meta}
\tikzset{snake it/.style={decorate, decoration=snake}}
\usetikzlibrary{arrows,shapes,positioning}
\usetikzlibrary{decorations.markings}
\tikzstyle arrowstyle=[scale=1]
\tikzstyle directed=[postaction={decorate,decoration={markings,mark=at position .65 with {\arrow[arrowstyle]{stealth}}}}]
\tikzstyle reverse directed=[postaction={decorate,decoration={markings,mark=at position .65 with {\arrowreversed[arrowstyle]{stealth};}}}]

\newtheorem{theorem}{Theorem}
\newtheorem{corollary}{Corollary}

\newtheorem{proposition}{Proposition}
\newtheorem{lemma}{Lemma}

\newtheorem{definition}{Definition}

\def\bC{{\mathbb C}}

\def\bN{{\mathbb N}}

\def\bZ{{\mathbb Z}}

\def\sH{\mathscr{H}}

\def\sK{\mathscr{K}}

\def\sV{\mathscr{V}}

\def\cA{{\mathfrak A}}
\def\cB{{\mathfrak B}}

\def\cM{{\mathfrak M}}
\def\cN{{\mathfrak N}}

\def\cA{{\ca A}}
\def\cB{{\ca B}}
\def\cC{{\ca C}}
\def\cD{{\ca D}}

\def\cL{{\ca L}}
\def\cM{{\ca M}}
\def\cN{{\ca N}}

\def\cT{{\ca T}}

\def\cV{{\ca V}}

\def\cZ{{\ca Z}}

\newcommand{\ca}[1]{{\cal #1}}
\newcommand{\ben}{\begin{equation}}
\newcommand{\een}{\end{equation}}
\def\bena{\begin{eqnarray}}
\def\eena{\end{eqnarray}}
\def\non{\nonumber}

\def\1{{\mathds{1}}}

\def\dim{{\mathrm{dim}}}
\def\Hom{{\mathrm{Hom}}}
\def\End{{\mathrm{End}}}

\renewcommand{\epsilon}{\varepsilon}

\renewcommand{\i}{\imath}
\newcommand{\RR}{\mathbb{R}}
\newcommand{\CC}{\mathbb{C}}

\begin{document}
\title{
Anyonic Chains --  $\alpha$-Induction -- CFT -- Defects -- Subfactors 
}

	\author{Stefan Hollands$^{1}$\thanks{\tt stefan.hollands@uni-leipzig.de}\\
	{\it $^1$ Universit\" at Leipzig, ITP and MPI-MiS Leipzig}
	}

\date{\today}
	
\maketitle

\begin{abstract}
Given a unitary fusion category, one can define the Hilbert space of 
a so-called ``anyonic spin-chain'' and nearest neighbor Hamiltonians providing a real-time evolution. There is considerable evidence 
that suitable scaling limits of such systems can lead to $1+1$-dimensional conformal field theories (CFTs), and in fact, can be used potentially to construct novel classes of 
CFTs. Besides the Hamiltonians and their densities, the spin chain is known to carry an algebra of symmetry operators commuting with the Hamiltonian, 
and these operators have an interesting representation as matrix-product-operators (MPOs). On the other hand, fusion categories are well-known 
to arise from a von Neumann algebra-subfactor pair. In this work, we investigate some interesting consequences of such structures for the 
corresponding anyonic spin-chain model. One of our main results is the construction of a novel algebra of MPOs acting on a bi-partite anyonic chain. 
We show that this algebra is precisely isomorphic to the defect algebra of $1+1$ CFTs as constructed by Fr\" ohlich et al. and Bischoff et al., even though
the model is defined on a finite lattice. We thus conjecture that its central projections are associated with the irreducible vertical (transparent) defects in the scaling limit of the model. 
Our results partly rely on the observation that MPOs are closely related to the so-called ``double triangle algebra'' arising in subfactor theory. 
In our subsequent constructions, we use insights into the structure of the double triangle algebra by B\" ockenhauer et al.
based on the braided structure of the categories and on $\alpha$-induction. The introductory section of this paper to subfactors and fusion categories has the character of a review.
\end{abstract}

\tableofcontents

\section{Introduction}

It has been known for a long time that low dimensional quantum systems can exhibit certain excitations -- so-called anyons \cite{wilczek1982quantum} -- obeying a generalization of 
the usual Bose-Fermi exchange statistics. A prominent example are unitary, rational $1+1$-dimensional conformal field theories (CFTs) and the mathematical 
structure emerging from the possibility to fuse and exchange such excitations is that of a ``modular tensor category''  \cite{rowell2009classification,bakalov2001lectures,frohlich2006correspondences,bischoff2015tensor}
which is a special type of ``fusion category'' \cite{etingof2005fusion}. 
The non-trivial exchange statistics corresponds to a braiding, which brings about a connection with low dimensional topology and topological quantum field theories. 
There are many different aspects and approaches to these structures \cite{moore1989classical, moore1990lectures, fredenhagen1989superselection, 
fredenhagen1992superselection,witten1989quantum, fuchs2002tft, fuchs2004tft, fuchs2005tft, longo1989index,longo1990index}, which is inevitably an incomplete list. 
Anyonic excitations are, however, not intrinsically tied to 
conformal symmetry or even relativistic kinematics and can be viewed more generally as states of systems exhibiting some sort of topolgicial order, such as gapped systems 
without a local order parameter. Among other things, systems of this type have been discussed as models for universal quantum computing, 
see e.g. \cite{wang2010topological} and references therein.

While modular tensor categories on the one hand can be seen as an {\em output} -- or algebraic skeleton -- of certain quantum field theories, 
they have been used more recently as an {\em input} to construct certain quantum mechanical spin systems called ``anyonic'' spin-chains, which may or 
may not have a quantum field theory as their scaling limit. Such systems were first considered by \cite{kitaev2006anyons,feiguin2007interacting} and subsequently 
studied (see e.g. \cite{trebst2008short,gils2009collective,gils2013anyonic,trebst2008collective,ardonne2011microscopic,vanhove2021critical,huang2021numerical,zini2018conformal,lootens2021matrix}) with the aim to provide relatively simple 
quantum mechanical systems exhibiting topological excitations that are practically built in from the start. In an 
ordinary spin chain, the sites correspond to certain representations of a group (e.g. $SU(2)$ or ${\mathbb Z}_2$), whereas in an anyonic chain, 
the sites are associated with certain objects of a fusion category. The possible ways of fusing $L$ of these objects by subsequent applications
of the fusion rules are in one-to-one correspondence with the basis vectors of the Hilbert space of the anyonic chain of length $L$, see fig. \ref{fig:18}. 

\begin{figure}[h!]
\begin{center}
  \includegraphics[width=0.8\textwidth,]{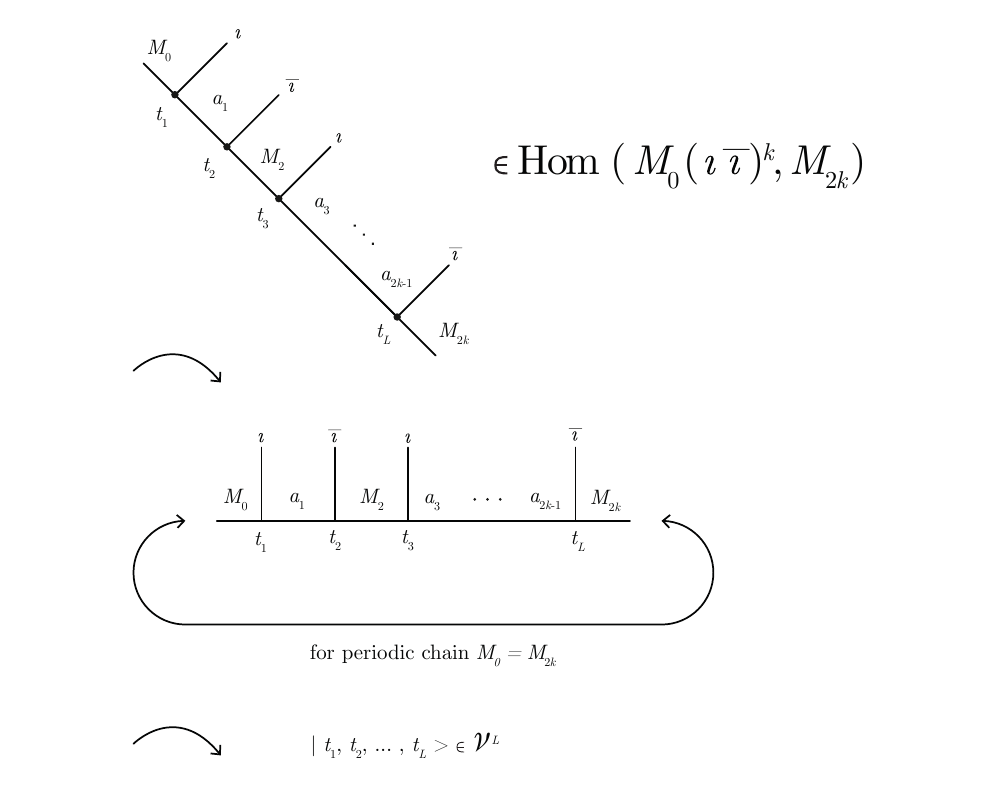}
  \end{center}
  \label{fig:18}
  \caption{Relation between intertwiner spaces and spin-chains.}
\end{figure}

The Hilbert space carries an algebra of ``symmetry operators'', for each simple object of the input fusion category. These symmetry operators 
can be seen \cite{bultinck2017anyons} as so-called ``matrix product operators'' (MPOs) which are built from the $F$-symbols of the category, and they obey an algebra 
that is isomorphic to the input fusion rules. The symmetry operators commute with each term of a very natural class of nearest neighbor interaction 
Hamiltonians and in this sense can broadly speaking be seen as being of a ``topological'' nature.

Unitary fusion categories, the input of anyonic spin chain constructions, appear prominently in the context of finite index inclusions of a factor-subfactor pair $\cN \subset \cM$ 
of von Neumann algebras, the study of which was initiated in the seminal work by Jones \cite{jones1983index} and subsequently elaborated by many authors 
using a variety of different approaches, see e.g. 
\cite{popa1990classification, bisch2000singly, ocneanu1988quantized,ocneanu1991quantum, bisch1997algebras,bisch1997bimodules,evans1998quantum, jones1999planar} as an inevitably 
incomplete list of references. 
It was realized almost from the beginning \cite{ocneanu1988quantized,ocneanu1991quantum} (see also \cite{evans1998quantum}) that the 
invariants of a factor-subfactor pair allow for the construction of certain algebraic structures such as the ``string algebra'' that in retrospect have 
a large amount of mathematical overlap with the basic constructions related to anyonic spin chains. This connection was recently 
investigated by \cite{kawahigashi2021projector,kawahigashi2020remark} who in particular eludicated the 
mathematical connection to the tube algebra \cite{izumi2000structure} (closely related also to the so-called Longo-Rehren inclusion 
\cite{longo1995nets,muger2003subfactors,muger2003subfactors}) whose role in describing anyonic excitations had been described by \cite{bultinck2017anyons}.
Since anyonic spin chains admit in certain cases continuum limits which are CFTs, we may call such constructions {\em bottom up}, i.e. from subfactors and fusion categories to QFTs. 

\subsection{Inclusions of von Neumann algebras and anyonic chains}
\label{subsec:intro1}

The passage from a factor-subfactor pair $\cN \subset \cM$ to the physical objects associated with an anyonic chain can roughly be described as follows. 
Since $\cN$ is contained in $\cM$, we have the inclusion (identity) map $\i: \cN \to \cM$. $\i$ is of course not in general invertible, but under the hypothesis of ``finite index'', 
there is a kind of inverse $\bar \i$; however $\bar \i \i$ is not the identity map of $\cN$ but instead ``unitarily equivalent'' to a ``sum'' of maps $\bar \i \i \cong id_\cN \oplus \rho_1 \oplus \dots$. 
We think of these maps $\rho_1, \dots$ as a kind of irreducible ``representation'' of $\cN$ different from the defining representation $id_\cN$ on the given Hilbert space. In a similar way, we also 
have $\i \bar \i \cong id_\cM \oplus A_1 \oplus \dots$, with irreducible ``representations'' $A_1$ of $\cM$. Technically, $\rho_1, \dots$ and $A_1, \dots$ are *-endomorphisms 
of $\cN$ and $\cM$, and we similarly also have endomorphisms $a_1, \dots$ from $\cN$ to $\cM$. These decompositions are in some ways analogous to a decomposition of a tensor product or restriction of ordinary representations (of, say, a group). They are called ``fusion''.

Now imagine iterating this process and consider 
$A_1 \i \cong a_2 \oplus \dots$, then $a_2 \bar \i \cong A_3 \oplus \dots$, etc. 
In this way we obtain a fusion tree as in fig. \ref{fig:18} where at each 
node we have to make a choice of the next sub-object to pick. Such a choice corresponds to an ``intertwiner'', 
\ben
\label{intertw}
t_1:\i  \bar \i \to A_1, \quad t_2: A_1 \i \to a_2, \quad t_3:   a_2 \bar \i \to A_3, \quad \dots .
\een
For example, intertwiner in the first case means that $t_1\i\bar \i(m) = A_1(m) t_1$ for all $m \in \cM$.
These subsequent choices of intertwiners 
are considered to define an abstract ket $|t_1, t_2, \dots, t_L\rangle$ in a Hilbert space $\sV^L$ of the chain of length $L$ (not to be confused 
with the Hilbert space on which $\cN,\cM$ act!). 

The concatenation of $t_i$'s defines an intertwiner $\underbrace{\i \bar \i \i \bar \i \dots \i \bar \i}_{L} \to id_\cM$ (taking the last object to be trivial\footnote{This corresponds in fact to a 
specific boundary condition on the chain.}), so 
$\sV^L$ may be thought of as a space intertwiners. So, if we have any 
element $m$ of $\cM$ that is also an intertwiner $\i \bar \i \i \bar \i \dots \i \bar \i \to \i \bar \i \i \bar \i \dots \i \bar \i$ then by composing 
with the intertwiner $\i \bar \i \i \bar \i \dots \i \bar \i \to id_\cM$ from $\sV^L$, we get another element from $\sV^L$, hence an 
$m$ of the sort described defines a linear operator on $\sV^L$. These $m$'s are by their very definition in the so-called ``relative commutant'',
$m \in \cM \cap (\cN_{L-1})'$, where a prime denotes the set of operators commuting with a given von Neumann algebra, and 
where $\cN_{L-1} := (\i \bar \i)^{L/2}(\cM)$, which is a von Neumann sub-algebra of $\cM$. Thus we have  
correspondences:
\ben
\begin{split}
\text{states in $\sV^L$} & \leftrightarrow \text{sequences of intertwiners $|t_1, \dots, t_L\rangle$,}\\
\text{operators on $\sV^L$} & \leftrightarrow \text{relative commutant $\cM \cap (\cN_{L-1})'$.}
\end{split}
\een
The sequence of inclusions $\cN_{L-1} \dots \subset \cN_3 \subset \cN_1 \subset \cM$ is the odd part of the ``Jones tunnel''
\cite{jones1983index} (with the even part being given by $\cN_{L}:= (\bar \i \i)^{L/2}(\cN)$, where $L$ is assumed even). 
Among many other things, Jones' work showed that the relative commutants $\cM \cap (\cN_{L-1})'$
are finite dimensional (hence isomorphic to direct sums of matrix algebras). By definition, they grow with increasing $L$, and 
each time we go down the tunnel by one unit, say from $\cN_x$ to $\cN_{x+1}$, they include a certain additional projection $e_x$ called a ``Jones projection'', 
meaning that $\cN_{x+1}$ is generated as a von Neumann algebra by $e_x$ and $\cN_x$.
The Jones projections play a crucial role in the entire theory and Jones showed that they satisfy a Temperly-Lieb algebra, 
\ben
\label{TLJ}
e_x e_{x \pm 1} e_x = d^{-1} e_x, \quad [e_x, e_y]=0 \quad \text{if $|x-y|>1$,}
\een
where $d \ge 1$ is the square root of the Jones index $[\cM:\cN]$ which he showed is quantized below $d=2$ as
\ben
\label{jonesq}
[\cM:\cN]^{1/2} = d \in \{ 2\cos[\pi / (k+2)] \ : \ \  k=1,2,3 \dots \} \cup [2,\infty].
\een
The Jones projections for $x=1, \dots, L$ are all in the relative commutant
$\cM \cap (\cN_{L-1})'$ and so can be viewed as operators on the chain Hilbert space $\sV^L$. In fact, they correspond in 
a sense to an action involving only neighboring sites $x,x+1$ on the chain, so it is natural to identify
\ben
\text{local operators on $\sV^L$}  \leftrightarrow \text{Jones projections $e_1, \dots, e_L$.}
\een
The sum of the $e_x$ is therefore a natural candidate for a local Hamiltonian,
\ben
H = J \sum_{x=1}^L e_x.
\een
In fact, the connection between the Temperly-Lieb algebra 
and anyonic Hamiltonians has been observed and used in the literature from the beginning in special models,
 see e.g. \cite{feiguin2007interacting,zini2018conformal}. There is evidence in special models that the algebra generated by the $e_x$
can be identified with a product of left-and right moving Virasoro algebras in a suitable 
conformal scaling limit of the chain, see \cite{zini2018conformal,stottmeister2022anyon} which uses ideas by \cite{koo1994representations}.

The Hilbert space $\sV^L$ of the chain also carries certain other special operators, called ``matrix product operators'' (MPO). These operators look 
schematically like the following fig. \ref{fig:MPO}. 
\begin{figure}
\centering
\begin{tikzpicture}[scale=.6]
\draw[thick] (-4,0) -- (0,0);
\draw[thick] (2,0) -- (6,0);
\draw[thick,dashed] (0,0) -- (2,0);
\draw[thick] (-3,1) -- (-3,-1);
\draw[thick] (-1,1) -- (-1,-1);
\draw[thick] (3,1) -- (3,-1);
\draw[thick] (5,1) -- (5,-1);

\draw (-8,0) node[anchor=east]{$\langle t_1',  \dots, t_L' | {\rm MPO} | t_1, \dots, t_L\rangle$};
\draw (-6,0) node[anchor=east]{$=$};
\draw (-3,-1) node[anchor=north]{$t_1$};
\draw (-1,-1) node[anchor=north]{$t_2$};
\draw (3,-1) node[anchor=north]{$t_{L-1}$};
\draw (5,-1) node[anchor=north]{$t_L$};
\draw (-3,1) node[anchor=south]{$t_1'$};
\draw (-1,1) node[anchor=south]{$t_2'$};
\draw (3,1) node[anchor=south]{$t_{L-1}'$};
\draw (5,1) node[anchor=south]{$t_L'$};
\filldraw[color=black, fill=white, thick](-3,0) circle (.5);
\filldraw[color=black, fill=white, thick](-1,0) circle (.5);
\filldraw[color=black, fill=white, thick](3,0) circle (.5);
\filldraw[color=black, fill=white, thick](5,0) circle (.5);
\end{tikzpicture}
  \caption{\label{fig:MPO} Schematic diagram for matrix element of an MPO.}
\end{figure}
The circles in this chain represent certain quantum $6j$-symbols (sometimes also called $F$-tensors or bi-unitary connections)
built from a fusion category associated with the inclusion $\cN \subset \cM$ and the legs are labelled 
by intertwiners $t_i$ \eqref{intertw}. The value of this concatenation of $6j$-symbols is thought of as the matrix element of an MPO, as described in more detail in the main text. 
For example, by closing the left and right horizontal wires, we obtain the operators $O_A^L$ on $\sV^L$ defined by \cite{bultinck2017anyons}, which are shown to satisfy 
\ben
\label{ooop}
O_A^L O_B^L = \sum_C N_{A,B}^C O_C^L
\een
with $N_{A,B}^C$ the fusion tensor (i.e. $N_{A,B}^C$ is the number of independent intertwiners $AB \to C$). 

The reason why these -- and the class of more general operators defined and investigated in this paper -- are called ``symmetry operators'' is that they 
commute with the local operators on the chain, hence in particular with the Hamiltonian, thus giving us 
global conservation laws. But commutation with even the local densities of the Hamiltonian is a much stronger property, so we may call 
them in a certain sense ``topological'', because they are invariant if we drive the evolution forward only locally. In 
summary, we have the correspondence
\ben
\text{symmetry operators on $\sV^L$}  \leftrightarrow \text{MPOs}
\een
The elaboration of this and related ideas will be the main theme of this work. 

\medskip

Jones' work was initially in the context of special von Neumann algebras of 
so-called type II$_1$, and the objects of the fusion category are in this case certain bimodules -- the natural notion of representation in this setting -- associated with the inclusion 
$\cN \subset \cM$ \cite{bisch1997algebras,bisch1997bimodules,evans1998quantum}. 
Jones work was soon generalized to inclusions of so-called type III \cite{hiai1988minimizing,kosaki1986extension}, where the invariants and fusion categories arising from the subfactor can be conveniently approached via the notion of an endomorphism \cite{longo1989index,longo1990index,izumi1991application}, a formalism which we also used in the above outline. 
The works \cite{longo1989index,longo1990index} and also \cite{fredenhagen1989superselection, fredenhagen1992superselection} brought to light in particular the close connection between the invariants of a factor-subfactor pair and the Doplicher-Haag-Roberts analysis of so-called superselection sectors in quantum field theory (QFT), see e.g. \cite{doplicher1989new,haag2012local}. From the viewpoint of QFT, an endomorphism corresponds to a representation of the observable algebras which is equivalent to the vacuum representation except in some subregion thought of the localization of the excitation. Since the localization can be translated by means of local operations, one gets a notion of exchange statistics of the excitations, which in low dimensions can be of anyonic type, thus endowing certain categories of localized endomorphisms with the structure of a so-called modular tensor category. Thus, one can say that low dimensional QFTs naturally provide as an output inclusions of von Neumann factors, and the associated objects in the corresponding fusion categories correspond to anyons. We may think of this as a {\em top down} direction because a QFT contains a lot more structure than the output fusion category due to the presence of local degrees of freedom -- they give {\em nets} of subfactors \cite{longo1995nets} in the terminology of algebraic QFT \cite{haag2012local}. 

\subsection{Main results}

In this paper, we further elaborate on the relation between the {\em top down} and {\em bottom up} connections between subfactors/fusion categories and QFTs. While our constructions do not touch the 
important analytical question of proving convergence of scaling limits of anyonic spin-chains to continuum QFTs (recently investigated in certain examples by \cite{osborne2021conformal,osborne2021quantum,morinelli2021scaling,zini2018conformal}, for an alternative program see 
\cite{Jones_2018}), we add non-trivial observations concerning the close analogy between vertical defects in anyonic spin-chains and transparent defects in $1+1$-dimensional rational CFTs. Our investigations crucially rely on harnessing the powerful machinery of 
subfactor theory, in particular $\alpha$-induction \cite{bockenhauer1998modular,bockenhauer1999modular,bockenhauer1999modular1,longo1995nets} and Ocneanu's double triangle algebra  \cite{ocneanu1988quantized,ocneanu1991quantum, bockenhauer1999alpha} in the context of anyonic spin-chains and MPOs. The main results of this paper are as follows:

\medskip
\noindent
{\bf Intertwiners, Jones tunnel, and anyonic chains.}
 We first relate the Hilbert space of an anyonic spin-chain to the Jones tunnel associated with a finite index inclusion $\cN \subset \cM$ of von Neumann factors and set up the basic connections
between the anyonic spin chain literature and subfactor theory such as intertwiners, $F$-symbols (referred to as $6j$-symbols in this paper), etc., as partly outlined in sec. \ref{subsec:intro1}. 
This part of the paper is hardly original and explained from a somewhat different perspective (type II factors) e.g. in \cite{kawahigashi2021projector,kawahigashi2020remark} which in turn partly builds on ideas of Ocneanu \cite{ocneanu1988quantized,ocneanu1991quantum}, see also \cite{evans1998quantum}. However, these observations are useful because we shall see how to import 
parts of the powerful machinery of subfactor theory to the study of anyonic spin chains. We also emphasize the connection with the formalism of endomorphisms and that, by contrast to most of the literature on anyonic spin chains, the construction that is most natural from the viewpoint of subfactor theory is to label the sites of the chains in an alternating way by the objects $\i \bar \i \i \bar \i \i \dots$, as opposed to a single object see fig. \ref{fig:18}. In simple examples of anyonic chains the objects $\i$ and $\bar \i$ typically can be identified; for example in the Ising category 
with objects $\sigma, \epsilon, id$ and fusion $\sigma^2 = id \oplus \epsilon$ etc., both $\i, \bar \i$ are identified with the self-conjugate object $\sigma$. 

From the perspective of subfactor theory, it is not natural to work with a single fusion category but with fusion categories 
$_\cN X_\cN$ (objects $\rho_1, \dots$) respectively $_\cM X_\cM$ (objects $A_1, \dots$) associated with $\cN$ respectively $\cM$, as well as the ``induction-restriction'' objects which form 
$_\cN X_\cM$ respectively ${}_\cM X_\cN$ (objects $a_1, \dots$). The objects $A$ associated with the symmetry operators $O_A^L$, which are typically discussed in the context of 
anyonic chains and related to defects \cite{buican2017anyonic} are from $_\cM X_\cM$, which is in general not a braided category. 
By contrast, we also have objects $\mu$ from $_\cN X_\cN$, which we assume to be a (non-degenerately) braided category. Thus, there is an imbalance between the 
properties of $_\cN X_\cN$ and $_\cM X_\cM$, which is mirrored by a corresponding imbalance, in general, between the category of primary fields of a 1+1 dimensional CFT and 
the category of defects, described in more detail below.

It would be interesting to connect our constructions to the discussion of string-net models in \cite{lootens2021matrix}, where modules of (different) categories appear. 
In our approach, the objects of ${}_\cN X_\cM$, ${}_\cM X_\cN$ likewise may be regarded as right/left-modules and ${}_\cM X_{\cM}$ as bimodulers 
associated to a ``Q-system'' that is closely related to ${}_\cN X_\cN$ \cite{bischoff2015tensor}, and as in our approach, the consistency of the fusion 
rules and associators enforced by the module properties plays a key role in \cite{lootens2021matrix}. However, in our constructions, detailed below, the braiding and $\alpha$-induction 
additionally play a central role and we do not see an overlap between our central results on vertical defects, outlined in the following, 
and the results by \cite{lootens2021matrix}.

\medskip
\noindent
{\bf Symmetry operators and double triangle algebra.}
The, in general very indirect, relationship between $_\cN X_\cN$, $_\cM X_\cM$, $_\cN X_\cM$, $_\cM X_\cN$ is at least partially encoded in the double triangle algebra \cite{ocneanu1988quantized,ocneanu1991quantum, bockenhauer1999alpha}. A crucial observation for the subsequent constructions in this paper is 
that there is a representation of this algebra on anyonic spin chains of arbitrary length $L$. In fact, 
the generators of the double triangle algebra are represented by MPOs built from chains of $6j$-symbols as in fig. \ref{fig:MPO}. Representers of particular members of the double triangle algebra 
yield the symmetry operators $O_A^L$ labelled by objects $A$ from $_\cM X_\cM$, discussed in the literature on anyonic spin chains and in particular \cite{bultinck2017anyons}. In fact, our proof of the representation property uses one of the main graphical ideas by \cite{bultinck2017anyons} called the ``zipper lemma'', which is a reflection of the pentagon identity for $6j$-symbols. 

The double triangle algebra contains also other special elements of interest from the perspective of anyonic spin-chains. This includes certain projections 
$q_{\mu,\lambda}$ which  are represented by MPOs $Q^L_{\mu,\lambda}$ on the spin chain of length $L$ and which are labelled by certain pairs of 
objects $\mu,\lambda$ from $_\cN X_\cN$. The structure of the 
double triangle algebra entails that we can write 
\ben
Q^L_{\mu,\nu} = \sum_A d^{-1}_A Y_{\mu,\nu, A} O^L_A
\een
in terms of the symmetry operators $O_A^L$ labelled by objects $A$ from ${}_\cM X_\cM$. The $d_A$'s are quantum dimensions of the simple objects $A$ and 
the coefficients $Y_{\mu,\nu,A}$ are defined in terms of $\alpha$-induction and relative braidings between $_\cN X_\cN$ and 
$_\cM X_\cM$, $_\cN X_\cM$, $_\cM X_\cN$. These coefficients are directly related to the 
Verlinde $S$-matrix $S_{A,B}$ in the rather special case that the fusion rules of ${}_\cM X_\cM$ are abelian (which happens e.g. in the ``Cardy case'', see below). 
But in the non-abelian case the fusion rules of ${}_\cM X_\cM$
cannot in general be diagonalized and thus the $Y_{\mu,\nu,A}$ may be seen as a generalized Verlinde tensor. 

Furthermore, the double triangle algebra also contains other, related, operators of interest for us whose representers are crucial to construct the

\medskip
\noindent
{\bf Defect algebra.}
The main result of this paper, which crucially relies on our observations related to the double triangle algebra, is the construction of certain MPOs associated with a bipartite anyonic spin chain, 
which we call $\Psi^{L_1,L_2}_{\lambda,\mu; w_1, w_2}$ and which are labelled by a pair $\lambda, \mu$ of simple objects from from $_\cN X_\cN$ and by $w_1, w_2^* \in {\rm Hom}(\alpha^-_\mu, \alpha^+_\lambda)$, where $\alpha^-_\mu, \alpha^+_\lambda$ denote the ``$\alpha$-induced'' object relative to a choice braiding in $_\cN X_\cN$ indicated by $\pm$. Thus, 
defining 
\ben
\label{Zmula}
Z_{\mu,\lambda} = \dim \Hom(\alpha^-_\mu, \alpha^+_\lambda),
\een 
there are $Z_{\mu,\lambda}^2$ generators $\Psi^{L_1,L_2}_{\lambda,\mu; w_1, w_2}$ for each pair $(\mu,\lambda)$. 
An important point is that, different from the symmetry operators $O^L_A$, the MPOs 
$\Psi^{L_1,L_2}_{\lambda,\mu; w_1, w_2}$ act on the tensor product of {\it two} chains of lengths $L_1$ respectively $L_2$, see fig. \ref{fig:220426_Fig-6b_TE} for a schematic 
drawing and see fig. \ref{fig:220413_Fig-6_TE} for a detailed definition.
\begin{figure}[h!]
\begin{center}
  \includegraphics[width=0.6\textwidth,]{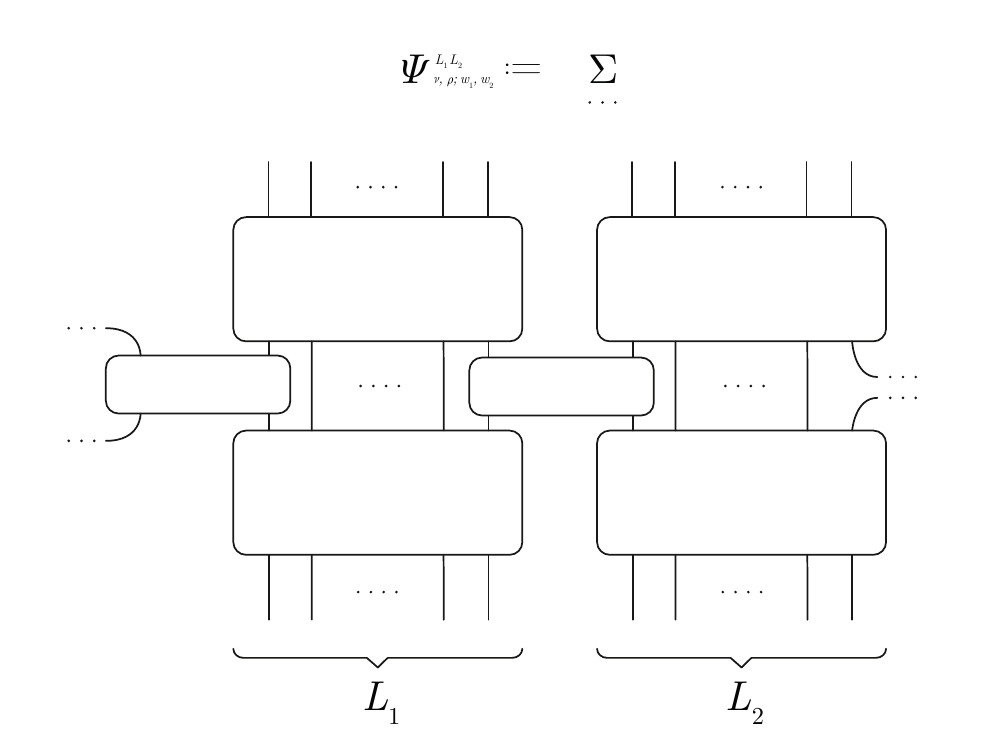}
\end{center}
  \caption{\label{fig:220426_Fig-6b_TE} Schematic form of the ``defect operators'' $\Psi^{L_1,L_2}_{\lambda,\mu; w_1, w_2}$ as MPOs on a bipartite chain of lenghts $L_1, L_2$. The large blocks 
  with $L_1,L_2$ strands are certain representers of the double triangle algebra.}
\end{figure}

We show that MPOs $\Psi^{L_1,L_2}_{\lambda,\mu; w_1, w_2}$ have several remarkable properties: (i) they commute among each other, i.e. for different $\lambda, \mu$ etc, (ii) they commute with 
the local densities of the chain Hamiltonian, (iii) they act on a ``conformal block''  (the range of the projector $Q^{L_1}_{\mu_1,\lambda_1} \otimes Q^{L_2}_{\mu_2,\lambda_2}$ on the bipartite chain) by fusing in the pair of charges $\lambda, \mu$, (iv) they generate an algebra $\cD^{L_1,L_2}$ isomorphic to the ``defect algebra'' that has been found \cite{frohlich2006correspondences,bischoff2015tensor,bischoff2016phase} in the context of $1+1$ dimensional CFTs\footnote{For a somewhat different approach to defects and fusion algebras in CFTs see \cite{bartels2019fusion}.}. 

\medskip
\noindent
The abelian nature of the defect algebra (i) shows that it is generated by commuting projections $D_A^{L_1,L_2}$, and (iv) permits us to import the
classification by \cite{frohlich2006correspondences,bischoff2015tensor,bischoff2016phase} implying that the minimal projections, 
$D_A^{L_1,L_2}$ are labelled precisely by the simple objects $A$ from $_\cM X_\cM$. By (ii), the subspace onto which a $D_A^{L_1,L_2}$ projects
 is left invariant by the local densities of the Hamiltonian on either side, so
it has a sort of ``topological character''. By (iii), the generators of the defect algebra $\cD^{L_1,L_2}$ of the bipartite chain 
of lenght $L=L_1+L_2$ are related to a decomposition 
\ben
\sV^{L_1,L_2} = \sV^{L_1} \otimes \sV^{L_2}, \quad 
\sV^{L_i} \cong \bigoplus_{\mu, \lambda: Z_{\lambda,\mu} \neq 0} \sV^{L_i}_{\lambda,\mu}
\een
of the Hilbert spaces $\sV^{L_i}$ of each of the two sub-chains $(i=1,2)$, and, on this tensor product Hilbert space the generators 
$\Psi^{L_1,L_2}_{\lambda,\mu; w_1, w_2}$ of $\cD^{L_1,L_2}$ act in a way that is precisely analogous to the ``braided product of two full centers'' in the 
CFT context \cite{bischoff2016phase}. [Each ``full center'' \cite{frohlich2006correspondences,bischoff2015tensor,bischoff2016phase} in the CFT-context corresponds 
to decomposing the Hilbert space of the CFT to the left/right
of the defect into conformal blocks, and on each side of the defect, the block associated with the primary with highest weight labels $\mu,\lambda$ appears $Z_{\mu,\lambda}$ 
times \eqref{Zmula}.] Together, (i)--(iv) suggest to us that the range of $D_A^{L_1,L_2}$ in the Hilbert space of
the bipartite anyonic chain precisely corresponds to a specific ``transparent boundary condition'' or ``defect'': This defect couples and 
sits in between the two parts of the bipartite chain, hence at ``constant position''.

It is known that the matrix $Z_{\lambda,\mu}$ \eqref{Zmula} is equal  modular invariant coupling matrix  in the CFT context appearing e.g. in the decomposition 
of the torus partition function $\sum_{\lambda, \mu} Z_{\lambda, \mu} \chi_\mu(q) \overline{\chi_\lambda(q)}$ or the Hilbert space 
$\oplus_{\mu,\lambda} Z_{\mu,\lambda} \sH^{\rm CFT}_\mu \otimes  \bar \sH^{\rm CFT}_\lambda$, which in our categorical setting is defined by \eqref{Zmula}
This suggests that, for a single chain with Hilbert space $\sV^L$, 
the subspaces $\sV^L_{\lambda,\mu}:= Q^L_{\lambda,\mu}\sV^L$ should be generated from $\sV^L_{0,0}$ by the action of certain local operators on the spin-chain which in the continuum limit 
ought to have the interpretation of conformal primary operators transforming in the representations $(\lambda,\mu)$ under the Virasoro algebra (assuming $_{\cN} X_{\cN}$ can be identified with a category of Virasoro representations to begin with). Unfortunately, we have as yet not been able to find convincing expressions for such operators. 

To summarize the discussion, our results suggest, but of course do not prove, the correspondences
\ben
\begin{split}
&\text{subspaces $Q^{L}_{\mu,\lambda} \sV^L$}  \leftrightarrow \text{CFT conformal block $(\lambda, \mu)$,}\\
&\text{subspaces $D^{L_1,L_2}_A \sV^{L_1,L_2}$}  \leftrightarrow \text{CFT states with defect $A$.}
\end{split}
\een
Here, the ``CFT'' refers to a continuum limit of corresponding anyonic spin chain, should such a limit exist. 

If the fusion algebra of ${}_\cN X_\cN$ is of so-called ``Cardy-'' or ``diagonal type'', 
\ben
Z_{\mu,\lambda} = \delta_{\mu,\lambda},
\een
corresponding in the CFT-context to a simple diagonal sum $\oplus_\mu  \sH^{\rm CFT}_\mu \otimes \bar \sH^{\rm CFT}_\mu$, 
then we show that our defect operators $\Psi^{L_1,L_2}_{\mu,\mu}$ (there are no `$w$'-labels in this case) can be labelled equivalently by the objects $A$ of the 
dual fusion category ${}_\cM X_\cM$, 
\ben
\mu \leftrightarrow A \quad \text{(diagonal models),}
\een
and our general results imply that they satisfy an algebra precisely isomorphic to \eqref{ooop}:
\ben
\label{ooop1}
\Psi_A^{L_1,L_2} \Psi_B^{L_1,L_2} = \sum_C N_{A,B}^C \Psi_C^{L_1,L_2}.
\een
 Of course, 
the operators $O^L_A$ and $\Psi^{L_1,L_2}_{A}$ are not at all the same: The first acts on a single chain of length $L$ whereas the second on a bipartite 
chain of lengths $L_1, L_2$, so one may tentatively think of the former as associated with ``horizontal defects'' (inserted at constant time) and of the latter with ``vertical defects'' (inserted at constant position). For general coupling matrices $Z_{\mu,\lambda}$ our abelian defect algebra $\cD^{L_1,L_2}$ appears to us in general different (as an algebra) from the, in general non-abelian,
algebra \eqref{ooop} generated by the symmetry operators $O_A^L$  discussed in connection with defects e.g. in \cite{buican2017anyonic,vanhove2018mapping,vanhove2021topological}, despite the fact that the central projections in our defect algebra $\cD^{L_1,L_2}$ are still labelled by the simple objects $A$ of the same category $_\cM X_\cM$. A similar remark applies to the 
defects constructed in the context of a wide class of lattice models by \cite{aasen2020topological}.
It would be interesting to understand the connection to these works better\footnote{The work \cite{buican2017anyonic} emphasizes the potential difference between the ``input category'', here $_\cM X_\cM$
and the ``output category'' of conformal defects and derives certain constraints. Our work on the other hand is concerned entirely with anyonic chains, although we discover 
new relationships between certain algebraic structures and algebraic structures appearing in CFTs. The works by \cite{vanhove2018mapping,vanhove2021topological} use a more Euclidean (imaginary time) description and the method of ``strange correlators''. The connection of that construction to ours is unclear to us.}. 

At any rate, we find it notable 
that an exact copy of the defect algebra of $1+1$-dimensional CFTs can be constructed for anyonic spin-chains before the continuum limit, especially in view of the fact that an anyonic spin chain may admit several continuum limits including non-conformal ones. It would be interesting to understand this point better. 

\subsection{Summary of notations and conventions} 

The von Neumann algebras appearing in this paper are always assumed to be infinite (type III) factors. 
$\cN \subset \cM$ always denotes a finite index inclusion of such von Neumann factors. 
Calligraphic letters $\cA, \dots, \cZ$ denote algebras, often von Neumann algebras. 
Operators from a 
von Neumann algebra will be denoted by lower case Roman letters. The adjoint operation in 
the von Neumann algebra is denoted by $x \mapsto x^*$ whereas the adjoint on the Hilbert space of the spin chain 
is denoted by $\dagger$. Operators on a spin chain of length $L$ are typically denoted by upper case letters $A^L, \Phi^L, \dots$. 
Upper case Roman indices $A,B,\dots$ also denote endomorphisms from
$\cM$ and are represented by thick solid lines in wire diagrams, Greek symbols $\mu,\nu, \dots$ denote endomorphisms from $\cN$
and are represented by dashed lines. Lower case Roman indices $a, b, \dots$ also denote endomorphisms from $\cN \to \cM$ and 
are represented by thin solid lines.

\section{Background}

\subsection{Von Neumann algebras}

{\bf General definitions:} See e.g. \cite{takesaki2003theory}. A von Neumann algebra is an ultra-weakly closed $*$-subalgebra of the algebra $\cB(\sH)$ of bounded operators 
on some Hilbert space, $\sH$. The commutant $\cN'$ of a von Neumann algebra is the von Neumann algebra of all bounded operators on 
$\sH$ commuting with each operator from $\cN$. The center of a von Neumann algebra is consequently $\cN \cap \cN'$. 
A von Neumann algebra with trivial center is said to be a factor. Factors can be classified into types I, II, III; the algebras considered in this work 
are assumed to be factors of type III. This property is used mainly to set up the calculus of endomorphisms (see below). In QFT, algebras 
are of this type \cite{haag2012local}, so the assumption is also natural if we have in mind relating the anyonic spin chain constructions back to QFT. 

An inclusion $\cN \subset \cM$ of von Neumann factors is said to be irreducible if $\cN' \cap \cM$ (also called the first relative commutant) consists of multiples 
of the identity. We will always assume that we are in this situation when considering inclusions. Associated with any inclusion of factors there is always the 
``dual inclusion'' $\cM' \subset \cN'$. For some constructions below, it is also convenient 
to assume that $\cM$ is $\sigma$-finite and in a ``standard form'', in the sense that there exists a vector $|\eta\rangle \in  \sH$ such that both $\cM |\eta\rangle$ and $\cM'|\eta\rangle$
are dense subspaces of $\sH$ (in such a case $|\eta\rangle$ is called cyclic and separating). For a von Neumann algebra in standard form 
one has an anti-linear, unitary, involutive operator $J_\cM$ exchanging $\cM$ with $\cM'$ under conjugation (i.e. $j_\cM(m) = {\rm Ad}(J_\cM) m := 
J_\cM m J_\cM \in \cM'$ iff $m \in \cM$). The existence of a cyclic and separating vector is a moderate assumption which we implicitly make throughout.

\medskip
\noindent
{\bf Conditional expectations:} See e.g.  \cite{longo1995nets, bisch1997bimodules, evans1998quantum, kosaki1986extension, hiai1988minimizing, takesaki2003theory}.
Let $\cN \subset \cM$ be an irreducible inclusion of two von Neumann factors. An ultraweakly continuous linear operator $E: \cM \to \cN$ is called 
a conditional expectation if it is positive, $E(m^*m) \ge 0$ for all $m \in \cM$, and if 
\ben
E(n_1 m n_2)= n_1 E(m) n_2
\een
for $m \in \cM, n_i \in \cN$. If there exists any conditional expectation $E$ at all, then the best constant $\lambda>0$ such that 
\ben
\label{eq:pp}
E(m^*m) \ge \lambda^{-1} m^*m \quad \text{for all $m \in \cM$}
\een
is called the index of $E$, and there exists one for which $\lambda$ is minimal. 
This $\lambda = [\cM:\cN]$ is the Jones-Kosaki index of the inclusion. In this paper we only consider 
inclusions with a {\em finite index}. In such a case one can always define a conditional expectation $E': \cN' \to \cM'$
for the dual inclusion.

\subsection{Fusion categories, endomorphisms and intertwiner calculus}
\label{sec:fusion}

See e.g. \cite{etingof2005fusion,longo1990index,bischoff2015tensor,muger2003subfactors,bockenhauer1999alpha,longo1996theory}. 
A fusion category over $\CC$ is a monoidal category with finitely many simple objects up to isomorphism and finite-dimensional $\CC$-linear morphism spaces, and such that the unit object is simple. 
Such a category is called unitary if it can be equipped with a $*$-operation such that it becomes a $C^*$-category in the sense of \cite{doplicher1989new}. Any such category can be 
realized as a category of (finite index) endomorphisms of an infinite (type III) von Neumann factor, $\cN$, and unitary fusion categories realized in this way will 
be denoted by $_\cN X_\cN$. An endomorphism $\mu \in \End(\cN)$ of a von Neumann algebra $\cN$
is an ultra-weakly continuous $*$-homomorphism such that $\mu(1)=1$. It is said to have finite index if $[\cN: \mu(\cN)]<\infty$, where the definition of index in the case of type III
von Neumann factors is as outlined above. In this paper, all endomorphisms considered are assumed to have a finite index.

\medskip
\noindent
{\bf Intertwiners ($\to$ Hom-spaces):} 
Given two endomorphisms $\mu,\nu$, one says that a linear operator $t \in \cN$  is an ``intertwiner'' if 
$t\mu(n)=\nu(n) t$ for all $n \in \cN$. The linear space of all such intertwiners is called $\Hom_\cN(\nu,\mu)$, 
but note that a given operator $t$ may belong to different Hom-spaces. The subscript `$\cN$' which indicates the algebra from which the 
intertwiners are taken is sometimes omitted where clear from the context. For the composition of two endomorphisms 
we write $\mu\nu := \mu \circ \nu$.
Two endomorphisms are called equivalent if there is a unitary 
intertwiner between them and irreducible (or simple) if there is no non-trivial self-intertwiner. If $t \in \Hom_\cN(\mu,\nu)$, then 
$t^* \in \Hom_\cN(\nu,\mu)$, 
and if $s \in \Hom_\cN(\lambda,\sigma)$, then 
\ben
\label{DHRprod}
s \times t := s \sigma(t) = \lambda(t) s \in \Hom_\cN(\lambda\mu,\sigma\nu)
\een
is called the Doplicher-Haag-Roberts (DHR) product. It is associative, satisfies $(s \times t)^* = s^* \times t^*$, and
gives the Hom-spaces the structure of a tensor category. Note that $1_\sigma \times t=\sigma(t)$ but $t \times 1_\sigma = t$ as 
operators, where $1_\sigma \in \Hom(\sigma,\sigma)$ is the trivial intertwiner (equal to the identity operator of $\cN$). 

\medskip
\noindent
{\bf Graphical calculus:} We often use a graphical calculus for manipulations involving intertwiners. In this calculus, an intertwiner is adjacent to 
a set of wires which correspond to the input (bottom) and output (top) endomorphisms $\mu, \nu, \dots$. The wires are typically drawn vertically. So for example, 
in fig. \ref{fig:220426_Fig-1_TE}, the box represents an intertwiner $t \in \Hom(\lambda_1 \dots \lambda_m,\mu_1 \dots \mu_n)$. 
\begin{figure}[h!]
\begin{center}
  \includegraphics[width=1.1\textwidth,]{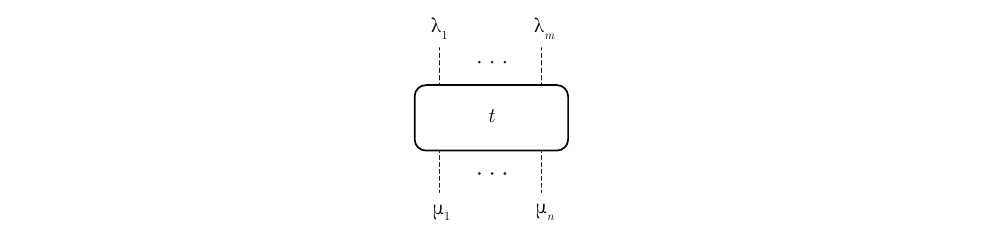}
\end{center}
  \caption{\label{fig:220426_Fig-1_TE} An intertwiner $t \in \Hom(\lambda_1 \dots \lambda_m,\mu_1 \dots \mu_n)$.}
\end{figure}
The identity intertwiner $1_\rho$ is always drawn as a vertical line. 
To represent the composition of (composable) intertwiners we stack them on top of each other and connect the wires as in fig. \ref{fig:220426_Fig-2_TE}. 
\begin{figure}[h!]
\begin{center}
  \includegraphics[width=1.1\textwidth,]{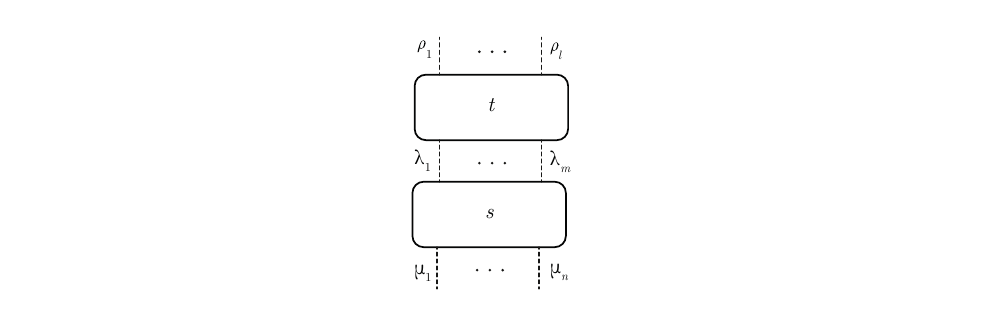}
\end{center}
  \caption{\label{fig:220426_Fig-2_TE} A product $ts$ of intertwiners. The diagram is read from bottom to top.}
\end{figure}
The DHR product \eqref{DHRprod} of two intertwiners is written by placing the two wire diagrams horizontally next to each other as in fig. \ref{fig:220426_Fig-3_TE}. 
\begin{figure}[h!]
\begin{center}
  \includegraphics[width=1.1\textwidth,]{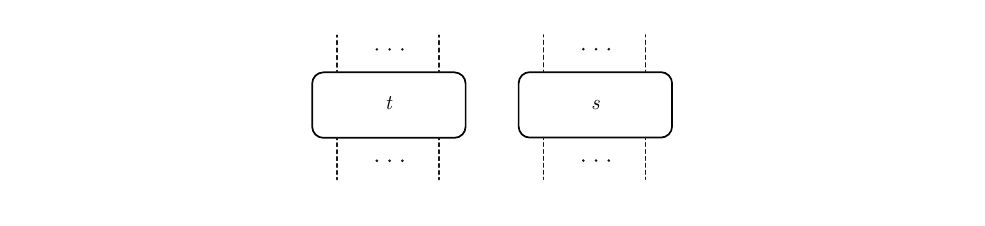}
\end{center}
  \caption{\label{fig:220426_Fig-3_TE} A DHR product $t \times s$ of intertwiners.}
\end{figure}
The diagrammatic notation (enriched by various special intertwiners in the following) is designed to automate certain identities. For example:
\begin{itemize}
\item We may slide boxes representing intertwiners in a DHR-product vertically upwards or downwards. This corresponds to the equivalent ways 
of writing the DHR product as in \eqref{DHRprod}.
\item If we have an identity between products of intertwiners represented by wired diagrams, we may place an arbitrary number of vertical wires
to the left or right (representing a DHR product with $1_{\mu_1} \times \cdots 1_{\mu_n} \times \cdots$ or $\cdots \times 1_{\mu_1} \cdots  \times 1_{\mu_n}$) to obtain 
a new identity. This follows from the homomorphism property of the $\mu_i$.
\item  Taking $*$ of an intertwiner represented by a wire diagram means reflection across the horizontal, or reading it top-to-bottom instead of bottom-to-top.
\end{itemize}
\medskip
\noindent
{\bf Decomposition:} Let $\theta \in \End(\cN)$. One writes $\theta \cong \oplus_{\mu} N_\mu \, \mu$ if there is 
a finite set of  irreducible and mutually inequivalent 
endomorphisms $\mu \in \End(\cN)$ and isometries $t_{\mu,i} \in \Hom_\cN(\theta, \mu), i = 1, \dots , N_\mu:=\dim_\CC \Hom_\cN(\theta, \mu)$ 
such that $\theta(n) = \sum_{\mu} \sum_{i=1}^{N_\mu} t_{\mu,i} \mu(n) t_{\mu,i}^*$ for all $n \in \cN$, and such that 
\ben\label{decomp}
t_{\mu,i}^* t_{\nu,j}^{} = \delta_{\mu,\nu} \delta^i_j 1, \quad 
\sum_{\mu} \sum_{i=1}^{N_\mu} t_{\mu,i}^{} t_{\mu,i}^*=1,
\een
see fig. \ref{fig:decomposition} for a wire diagram of the second identity and see fig. \ref{fig:decomposition2} for a wire diagram of the first identity. 
Note that for any pair $(\nu, \mu)$ of endomorphisms in $\End(\cN)$ with $\nu$ 
irreducible, the complex linear space $\Hom_\cN(\mu,\nu)$ is a Hilbert space with inner product 
\ben
\label{scalts}
t^* s = (t,s) 1, \quad t, s \in \Hom_\cN(\mu,\nu), 
\een
because $t^* s$ is a scalar. 
Then  \eqref{decomp} expresses that $\{t_{\mu,i} : i = 1, \dots, N_\mu\}$ is an orthonormal basis (ONB)
of $\Hom_\cN(\theta, \mu)$. 
\begin{figure}[h!]
\centering
\begin{tikzpicture}[scale=.8]
\draw (-2.5,1) node[anchor=west]{$\sum_t$};
\draw (2.5,1) node[anchor=west]{$=$};
\draw (0,.1) node[anchor=west]{$t^*$};
\draw (0,1.9) node[anchor=west]{$t$};
\draw (-1,-1) node[anchor=north]{$\mu$};
\draw (-1,3) node[anchor=south]{$\mu$};
\draw (1,-1) node[anchor=north]{$\nu$};
\draw (1,3) node[anchor=south]{$\nu$};
\draw (4,-1) node[anchor=north]{$\mu$};
\draw (5,-1) node[anchor=north]{$\nu$};
\draw[ line width=.03cm,
    dash pattern=on .09 cm off .04 cm] (0,0) -- (-1,-1);
\draw[ line width=.03cm,
    dash pattern=on .09 cm off .04 cm] (0,0) -- (1,-1);    
\draw[ line width=.03cm,
    dash pattern=on .09 cm off .04 cm] (0,0) -- (0,2);
\draw[ line width=.03cm,
    dash pattern=on .09 cm off .04 cm] (0,2) -- (-1,3);    
\draw[ line width=.03cm,
    dash pattern=on .09 cm off .04 cm] (0,2) -- (1,3);    
\draw[ line width=.03cm,
    dash pattern=on .09 cm off .04 cm] (4,-1) -- (4,3);  
\draw[ line width=.03cm,
    dash pattern=on .09 cm off .04 cm] (5,-1) -- (5,3);   
\end{tikzpicture}
  \caption{\label{fig:decomposition} Wire diagram for \eqref{decomp}. 
  }
\end{figure}
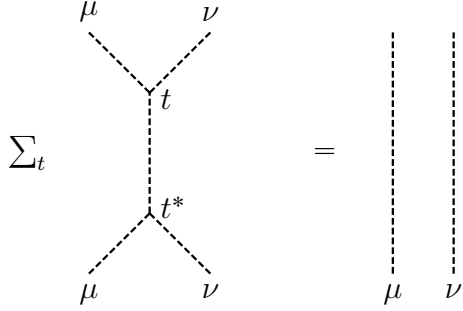
\begin{figure}[h!]
\centering
\begin{tikzpicture}[scale=.8]
\draw (2.5,1) node[anchor=west]{$= (s,t) \delta_{\lambda,\sigma}$};
\draw (0,-0.1) node[anchor=west]{$t$};
\draw (0,2.1) node[anchor=west]{$s^*$};
\draw (6,-1) node[anchor=north]{$\lambda$};
\draw (0,-1) node[anchor=north]{$\lambda$};
\draw (0,3) node[anchor=south]{$\sigma$};
\draw (1.5,1) node[anchor=east]{$\mu$};
\draw (-1.5,1) node[anchor=west]{$\nu$};
\draw[ line width=.03cm,
    dash pattern=on .09 cm off .04 cm] (0,0) .. controls (1.1,1)  .. (0,2);
\draw[ line width=.03cm,
    dash pattern=on .09 cm off .04 cm] (0,0) .. controls (-1.1,1)  .. (0,2);    
\draw[ line width=.03cm,
    dash pattern=on .09 cm off .04 cm] (0,0) -- (0,-1);
\draw[ line width=.03cm,
    dash pattern=on .09 cm off .04 cm] (0,2) -- (0,3);    
 \draw[ line width=.03cm,
    dash pattern=on .09 cm off .04 cm] (6,-1) -- (6,3);       
\end{tikzpicture}
  \caption{\label{fig:decomposition2} Wire diagram for \eqref{scalts}. 
  }
\end{figure}
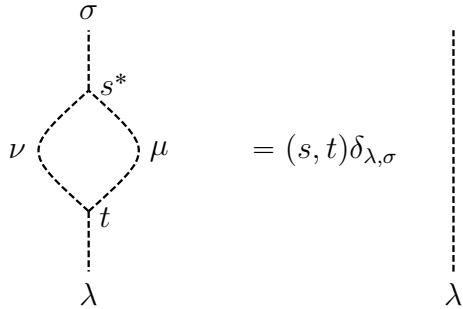

\medskip
\noindent
{\bf Fusion:} Given irreducible $\mu,\nu$ from some unitary fusion category $_\cN X_\cN \subset \End(\cN)$, the decomposition into irreducible endomorphisms 
$\sigma \in _\cN X_\cN$ is a finite sum as in
\ben
\mu \nu \cong \bigoplus_\sigma N_{\mu,\nu}^\sigma \ \sigma, 
\een
and the non-negative integers $N_{\mu,\nu}^\sigma$ are called the fusion coefficients. They satisfy 
obvious associativity-type conditions resulting from the associativity of the composition of endomorphisms. 
Note that it need not be the case that $\mu\nu$ is unitarily equivalent to $\nu\mu$, so the fusion 
matrices $N_{\mu,\nu}^\sigma$ need not be symmetric in the lower indices. 

\medskip
\noindent
{\bf Conjugate and Dimension:} Let $\mu \in \End(\cN)$ be an irreducible endomorphism of the von Neumann factor $\cN$.
One calls $\bar \mu \in \End(\cN)$ a conjugate endomorphism if the fusion of $\mu \bar \mu$ and $\bar \mu \mu$ contain 
the identity endomorphism. We assume throughout that the dimension of the endomorphisms considered is finite. 
This dimension is defined (for example) via the Jones-Kosaki index as $d_\mu = [\cN, \mu(\cN)]^{1/2}$. 
When the index of $\mu$ is finite, then there exist $r_\mu,\bar r_\mu \in \cN$ such that 
$\bar r_\mu  \in \Hom_\cN(\mu \bar \mu, id),  
r_\mu \in \Hom_\cN(\bar \mu \mu, id)$ and, 
\ben
\label{conjgacyrel}
{\bar{\mu}}(r^*_\mu) \bar r_\mu = 1 = \mu(\bar r^*_\mu) r_\mu, 
\quad r^*_\mu r_\mu = d_{\mu} \cdot 1 = \bar r^*_\mu \bar r_\mu.
\een
In particular, $r_\mu,\bar r_\mu$ are multiples of isometries.
The second relation gives a relation with the dimension $d_\mu \equiv d(\mu) \ge 1$ of $\mu$. 
One has $d_{\mu} = d_{\bar \mu}$ and for $\mu,\nu,\sigma$ irreducible endomorphisms of $\cN$,
\ben
d_\mu d_\nu = \sum_\sigma N_{\mu,\nu}^\sigma d_\sigma.
\een
These formulas express the additivity/multiplicativity of the quantum dimension under decomposition/fusion
and the invariance under conjugation and are the basic justification for the usage of the term ``dimension'' 
even though $d_\mu$ in general does not have to be integer.

The graphical representation of $r_\mu, \bar r_\mu$ is given in fig. \ref{fig:r}.
\begin{figure}
\centering
\begin{tikzpicture}[scale=.6]
\draw (-2,0) node[anchor=north]{$\bar \lambda$};
\draw (0,0) node[anchor=north]{$\lambda$};
\draw (2,1) node[anchor=south]{$\bar \lambda$};
\draw (4,1) node[anchor=south]{$\lambda$};
\draw (10,0) node[anchor=north]{$\bar \lambda$};
\draw (8,0) node[anchor=north]{$\lambda$};
\draw (12,1) node[anchor=south]{$ \lambda$};
\draw (14,1) node[anchor=south]{$\bar \lambda$};
\draw[ line width=.03cm,
    dash pattern=on .09 cm off .04 cm]  (0,0) arc (0:180:1);
\draw[ line width=.03cm,
    dash pattern=on .09 cm off .04 cm]  (2,1) arc (180:359:1);
\draw[ line width=.03cm,
    dash pattern=on .09 cm off .04 cm]  (10,0) arc (0:180:1);
\draw[ line width=.03cm,
    dash pattern=on .09 cm off .04 cm]  (12,1) arc (180:359:1);
\end{tikzpicture}
  \caption{\label{fig:r} Wire diagrams for $r_\lambda^*, r_\lambda, \bar r_\lambda^*, \bar r_\lambda$.}
\end{figure}
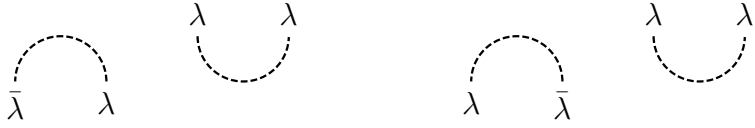
 The wire diagram for the conjugacy relation is 
depicted in fig. \ref{fig:220426_Fig-4_TE}.
\begin{figure}[h!]
\begin{center}
  \includegraphics[width=1.1\textwidth,]{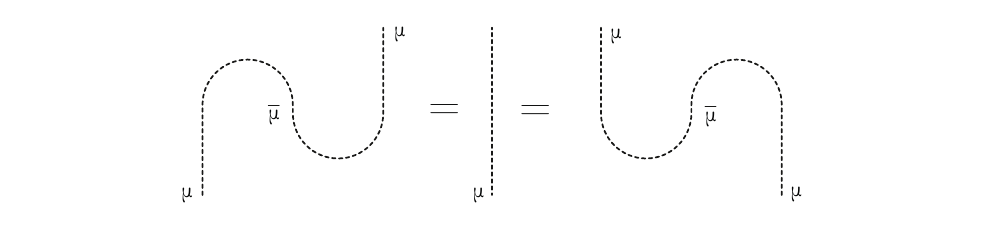}
\end{center}
  \caption{\label{fig:220426_Fig-4_TE} Topological invariance of wire diagram representing the conjugacy relation \eqref{conjgacyrel}.}
\end{figure}
The wire diagram for the normalization and isometry property of $\bar r_\lambda$ (similarly for $r_\lambda$) is depicted in fig. \ref{fig:circle}.
\begin{figure}
\centering
\begin{tikzpicture}[scale=.8]
\draw (-5,0.5) node[anchor=north]{$d_\lambda \cdot 1 = $};
\draw (-2.4,0.5) node[anchor=north]{$\lambda$};
\draw (.4,0.5) node[anchor=north]{$\bar \lambda$};
\draw[ line width=.03cm,
    dash pattern=on .09 cm off .04 cm] (0,0) arc (0:360:1);
\end{tikzpicture}
  \caption{\label{fig:circle} Wire diagram for $\bar r_\lambda^* \bar r_\lambda = d_\lambda 1$. 
  }
\end{figure}
Similar constructions can be 
made if $a$ is an endomorphism from $\cN \to \cM$, in which case $\overline{a}$ is an endomorphism $\cM \to \cN$.
In either case, one can achieve that $\overline{a \bar b} = b \bar a$. A generalization to 
reducible endomorphisms is also possible.  

\medskip
\noindent
{\bf Frobenius duality:} Let $\mu, \nu, \lambda \in \End(\cN)$ be irreducible and $t \in \Hom_\cN(\mu\lambda,\nu)$ be 
normalized to one, $t^* t = 1$. Then we define a ``Frobenius-dual'' endomorphism $\tilde t \in \Hom_\cN(\nu \bar \lambda,\mu)$, see fig. \ref{fig:frobenius}, 
\ben
\label{frobenius}
\tilde t := \left( \frac{d_\nu}{d_\mu} \right)^{-1/2} (t \times 1_{\bar \lambda})^*(1_\mu \times \bar r_\lambda).
\een
The normalization factor has been chosen so that $\tilde t^* \tilde t = 1$. One shows that 
Frobenius duality is involutive, so we get an anti-isometric identification of intertwiner spaces and corresponding identities between fusion coefficients. 
\begin{figure}
\centering
\begin{tikzpicture}[scale=.8]
\draw (-3.5,0) node[anchor=west]{$\left( \frac{d_\nu}{d_\mu} \right)^{-1/2}$};
\draw (0,0) node[anchor=west]{$t^*$};
\draw (-2,-2) node[anchor=north]{$\mu$};
\draw (0,2) node[anchor=south]{$\nu$};
\draw (2,2) node[anchor=south]{$\bar \lambda$};
\draw[ line width=.03cm,
    dash pattern=on .09 cm off .04 cm] (0,0) -- (-2,-2);
\draw[ line width=.03cm,
    dash pattern=on .09 cm off .04 cm] (0,0) .. controls (0.5,-2) and (2.3,-2) .. (2,2);
\draw[ line width=.03cm,
    dash pattern=on .09 cm off .04 cm] (0,0) -- (0,2);
\end{tikzpicture}
  \caption{\label{fig:frobenius} Wire diagram for $\tilde t$ given $t$. 
  }
\end{figure}

\medskip
\noindent
{\bf Conjugate intertwiner:} Let $\mu, \nu, \lambda \in \End(\cN)$ be irreducible and $t \in \Hom_\cN(\mu\lambda,\nu)$ be 
normalized to one, $t^* t = 1$. Then we define a ``conjugate'' endomorphism 
\ben
\bar t := 
(1_{\bar \nu} \times \bar r_\mu)^*
(1_{\bar \nu} \times \bar r_\lambda \times 1_{\bar \mu})^*
(1_{\bar \nu} \times t \times 1_{\bar \mu})^*(r_\nu \times 1_{\bar \lambda \bar \mu})  \in \Hom_\cN(\bar \lambda \bar \mu,\bar \nu).
\een
One shows that the normalization is chosen so that $\bar t^* \bar t = 1$, and that 
conjugation is involutive, so we get an anti-isometric identification 
of intertwiner spaces and corresponding identities between fusion coefficients. %

\medskip
\noindent
{\bf Braiding:} Let $_\cN X_\cN$ be a unitary fusion category. If $\mu \nu \cong \nu \mu$
for any $\mu,\nu \in \, _\cN X_\cN$ we say the system is braided if a consistent choice of the unitaries implementing the equivalence, 
called $\epsilon^\pm(\mu,\nu) \in \Hom_\cM(\mu\nu,\nu\mu)$, i.e. $\mu \nu = {\rm Ad}[\epsilon^\pm(\mu,\nu)] \nu\mu$, can be made. 
Here $\pm$ refer to over- and under-crossing which are the adjoints of each other. 
If we do not have a superscript as in $\epsilon(\lambda,\mu)$ then by convention ``$+$'' is meant. Consistency means that we have 
the so-called braiding-fusion relations (BFE) and the Yang-Baxter relations (YBE). The YBEs are 
\ben
(1_\rho \times \epsilon(\lambda,\mu)) ( \epsilon(\lambda,\rho) \times 1_\mu)(1_\lambda \times \epsilon(\mu,\rho)) =
(\epsilon(\mu,\rho) \times 1_\nu) (1_\mu \times \epsilon(\lambda,\rho))(\epsilon(\lambda,\mu) \times 1_\rho).
\label{YBE}
\een
They correspond to the wire diagram in fig. \ref{fig:220426_Fig-5_TE}.
\begin{figure}[h!]
\begin{center}
  \includegraphics[width=1.1\textwidth,]{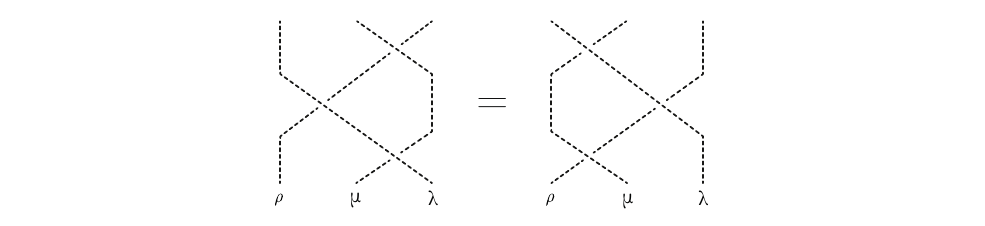}
\end{center}
  \caption{\label{fig:220426_Fig-5_TE} Topological invariance of wire diagram representing the YBE.}
\end{figure}
The BFEs are ($t \in \Hom(\mu\nu,\lambda)$)
\ben
\begin{split}
\epsilon(\rho,\lambda)
(1_\rho \times t^*) &=
(t^* \times 1_\rho)
(1_\mu \times \epsilon(\rho,\nu))
(\epsilon(\rho,\mu) \times 1_\nu)\\
\epsilon(\lambda, \rho)
(t^* \times 1_\rho) &=
(1_\rho \times t^*)
(\epsilon(\mu,\rho) \times 1_\nu)
(1_\mu \times \epsilon(\nu,\rho))
\end{split}
\label{BFE}
\een
the first of which corresponds to the wire diagram in fig. \ref{fig:220426_Fig-6_TE}.
\begin{figure}[h!]
\begin{center}
  \includegraphics[width=1.1\textwidth,]{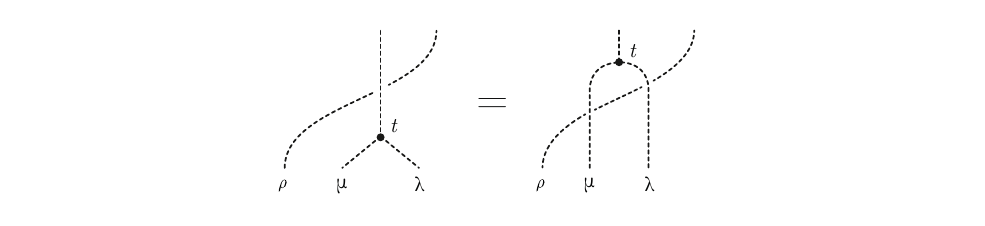}
\end{center}
  \caption{\label{fig:220426_Fig-6_TE} Topological invariance of wire diagram representing the BFE. There is a similar diagram with a crossing from right to left and 
  there are similar relations with over-crossings associated with $\epsilon^-(\mu,\nu)$.}
\end{figure}
We will assume in the following that the braiding is non-degenerate, i.e. $\epsilon^\pm(\lambda, \mu)$ are equal for $\pm$ and all $\mu \in \, _\cN X_\cN$ 
if and only if $\lambda = id$.
A braiding can be generalized in a natural way to 
any endomorphism that is decomposable into irreducible endomorphisms from $_\cN X_\cN$. 
The braiding implies an obvious symmetry of the fusion coefficients. 

One has ($\mu \in \, _\cN X_\cN$): 
\ben
(1_\mu \times \bar r_\mu)^* (\epsilon(\mu,\mu) \times 1_{\bar \mu})(1_\mu \times \bar r_\mu) = e^{2\pi i h(\mu)} ,
\een
where $h(\mu) \in \RR$ is called the statistics phase. For $t \in \Hom_\cM(\sigma,\mu\nu), \mu,\nu,\sigma \in \, _\cN X_\cN$, one can show using 
the BFE and YBE that
\ben
t\epsilon(\nu,\mu) \epsilon(\mu,\nu) = e^{2\pi i(h(\mu) + h(\nu) - h(\sigma))} t.
\een 
The Rehren matrix, see fig. \ref{fig:Y1}
\ben
\label{YRehren}
Y_{\mu,\nu} := d_\mu (1_\mu \times \bar r_\nu)^* (\epsilon(\nu,\mu) \epsilon(\mu,\nu) \times 1_{\bar \nu})(1_\mu \times \bar r_\nu)
\een
is shown to satisfy $Y_{\mu,\nu} = \sum_\lambda d_\lambda e^{2\pi i(h(\mu) + h(\nu) - h(\lambda))} N_{\mu,\nu}^\lambda$. If the braiding is nondegenerate, 
$_\cN X_\cN$ is said to be ``modular''. In such a case, $Y_{\mu,\nu}$ it is equal
up to a prefactor to the Verlinde matrix $S_{\mu,\nu}$ which diagonalizes the fusion coefficients. 

\subsection{Q-systems and subfactors}\label{sec:Qsys}

\noindent
{\bf Q-systems:} See \cite{longo1994duality,longo1995nets,bischoff2015tensor}.
A Q-system is a way to encode an inclusion of properly infinite von Neumann factors $\cN \subset \cM$
possessing a minimal conditional expectation $E:\cM \to \cN$ such that the index, denoted here by $d^2$, is finite. 
An important point is that the data in the Q-system only refer to the smaller factor, $\cN$.

\begin{definition}\label{Qsysdef}
A Q-system is a triple $(\theta, x, w)$ where: $\theta \cong \oplus_i \rho_i$ is an endomorphism of $\cN$, 
$w \in \Hom(\theta, id) \subset \cN$ and $x \in \Hom(\theta^2, \theta) \subset \cN$ such that
\ben\label{Q1}
w^* x = \theta(w^*)x = 1, \quad 
x^2 = \theta(x)x, \quad
\theta(x^*)x = xx^* = x^* \theta(x), 
\een
see figs. \ref{fig:X1}, \ref{fig:X3}, \ref{fig:X4} as well as 
\ben\label{Q2} 
w^* w = d \cdot 1, \quad 
x^* x =  d \cdot 1,
\een
see fig. \ref{fig:X2}.
\end{definition}

\begin{figure}[h!]
\centering
\begin{tikzpicture}[scale=.6]
\draw (-2,-2) node[anchor=north]{$\theta$};
\draw (2,2) node[anchor=south]{$\theta$};
\draw (-1,2) node[anchor=south]{$\theta$};
\draw (1,-2) node[anchor=north]{$\theta$};
\draw (-1,1) node[anchor=north]{$x^*$};
\draw (1,-1) node[anchor=south]{$x$};
\draw (6,2) node[anchor=south]{$\theta$};
\draw (4,2) node[anchor=south]{$\theta$};
\draw (6,-2) node[anchor=north]{$\theta$};
\draw (4,-2) node[anchor=north]{$\theta$};
\draw (5,.7) node[anchor=west]{$x^*$};
\draw (5,-.7) node[anchor=west]{$x$};
\draw (9,-2) node[anchor=north]{$\theta$};
\draw (8,2) node[anchor=south]{$\theta$};
\draw (11,2) node[anchor=south]{$\theta$};
\draw (12,-2) node[anchor=north]{$\theta$};
\draw (11,1) node[anchor=north]{$x^*$};
\draw (9,-1) node[anchor=south]{$x$};
\draw (3,0) node{$=$};
\draw (7,0) node{$=$};
\draw[ line width=.03cm,
    dash pattern=on .09 cm off .04 cm]  (0,0) arc (0:180:1);
\draw[ line width=.03cm,
    dash pattern=on .09 cm off .04 cm]  (0,0) arc (180:359:1);
\draw[ line width=.03cm,
    dash pattern=on .09 cm off .04 cm]  (6,-2) arc (0:180:1);
\draw[ line width=.03cm,
    dash pattern=on .09 cm off .04 cm]  (6,2) arc (0:-180:1);
\draw[ line width=.03cm,
    dash pattern=on .09 cm off .04 cm]  (-2,0) -- (-2,-2);
\draw[ line width=.03cm,
    dash pattern=on .09 cm off .04 cm]  (2,0) -- (2,2);
\draw[ line width=.03cm,
    dash pattern=on .09 cm off .04 cm]  (1,-1) -- (1,-2);
\draw[ line width=.03cm,
    dash pattern=on .09 cm off .04 cm]  (-1,1) -- (-1,2);
 \draw[ line width=.03cm,
    dash pattern=on .09 cm off .04 cm]  (5,-1) -- (5,1);
\draw[ line width=.03cm,
    dash pattern=on .09 cm off .04 cm]  (8,0) arc (180:359:1);
\draw[ line width=.03cm,
    dash pattern=on .09 cm off .04 cm]  (12,0) arc (0:180:1);    
\draw[ line width=.03cm,
    dash pattern=on .09 cm off .04 cm]  (9,-1) -- (9,-2);
\draw[ line width=.03cm,
    dash pattern=on .09 cm off .04 cm]  (8,0) -- (8,2);
\draw[ line width=.03cm,
    dash pattern=on .09 cm off .04 cm]  (11,1) -- (11,2);
\draw[ line width=.03cm,
    dash pattern=on .09 cm off .04 cm]  (12,0) -- (12,-2);
\end{tikzpicture}
  \caption{\label{fig:X1} Wire diagrams for $x^* \theta(x)  = xx^* = \theta(x^*)x$, 
  which is equivalent to $(x^* \times 1_\theta)(1_\theta \times x) = xx^* = (1_\theta \times x^*)(x \times 1_\theta)$ in DHR notation.}
\end{figure}
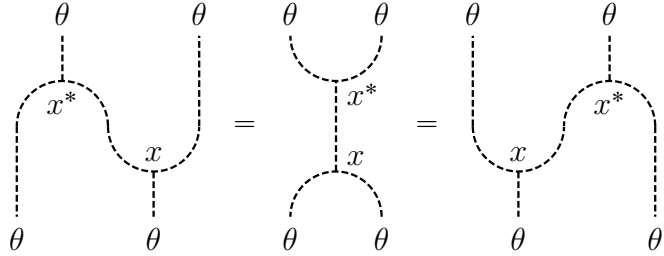
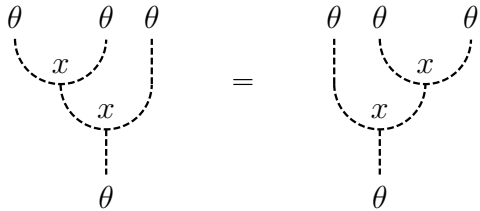
\begin{figure}[h!]
\centering
\begin{tikzpicture}[scale=.6]
\draw (-7,1) node[anchor=south]{$\theta$};
\draw (-5,1) node[anchor=south]{$\theta$};
\draw (-5,-1) node[anchor=south]{$x$};
\draw (-6,0) node[anchor=south]{$x$};
\draw (-4,1) node[anchor=south]{$\theta$};
\draw (-5,-2) node[anchor=north]{$\theta$};
\draw (-2,0) node{$=$};
\draw[ line width=.03cm,
    dash pattern=on .09 cm off .04 cm]  (-6,0) arc (180:359:1);
\draw[ line width=.03cm,
    dash pattern=on .09 cm off .04 cm]  (-7,1) arc (180:359:1);    
\draw[ line width=.03cm,
    dash pattern=on .09 cm off .04 cm]  (-5,-1) -- (-5,-2);
\draw[ line width=.03cm,
    dash pattern=on .09 cm off .04 cm]  (-4,0) -- (-4,1);   
\draw (0,1) node[anchor=south]{$\theta$};
\draw (1,1) node[anchor=south]{$\theta$};
\draw (1,-1) node[anchor=south]{$x$};
\draw (2,0) node[anchor=south]{$x$};
\draw (3,1) node[anchor=south]{$\theta$};
\draw (1,-2) node[anchor=north]{$\theta$};
\draw[ line width=.03cm,
    dash pattern=on .09 cm off .04 cm]  (0,0) arc (180:359:1);
\draw[ line width=.03cm,
    dash pattern=on .09 cm off .04 cm]  (1,1) arc (180:359:1);    
    
\draw[ line width=.03cm,
    dash pattern=on .09 cm off .04 cm]  (1,-1) -- (1,-2);
\draw[ line width=.03cm,
    dash pattern=on .09 cm off .04 cm]  (0,0) -- (0,1);    
\end{tikzpicture}
  \caption{\label{fig:X3} Wire diagrams for $x^2 = \theta(x)x$, which is equivalent to 
  $(x \times 1_\theta)x = (1_\theta \times x)x$ in DHR-notation.}
\end{figure}
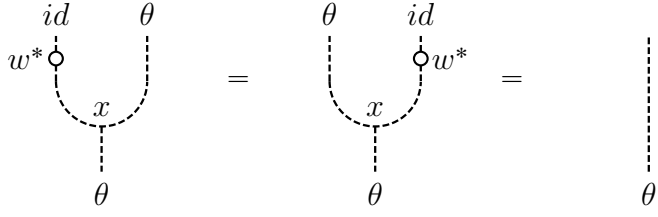
\begin{figure}[h!]
\centering
\begin{tikzpicture}[scale=.6]
\draw (-5,-2) node[anchor=north]{$\theta$};
\draw (-5,-1) node[anchor=south]{$x$};
\draw (-4,1) node[anchor=south]{$\theta$};
\draw (-6,1) node[anchor=south]{$id$};
\draw (-6,.5) node[anchor=east]{$w^*$};
\draw (-2,0) node{$=$};
\draw (4,0) node{$=$};
\draw[ line width=.03cm,
    dash pattern=on .09 cm off .04 cm]  (-6,0) arc (180:359:1); 
\draw[ line width=.03cm,
    dash pattern=on .09 cm off .04 cm]  (-5,-1) -- (-5,-2);
\draw[ line width=.03cm,
    dash pattern=on .09 cm off .04 cm]  (-4,0) -- (-4,1);   
 \draw[ line width=.03cm,
    dash pattern=on .09 cm off .04 cm]  (-6,0) -- (-6,1);  
 \filldraw[color=black, fill=white, thick](-6,0.5) circle (.15);      
\draw (0,1) node[anchor=south]{$\theta$};
\draw (1,-1) node[anchor=south]{$x$};
\draw (1,-2) node[anchor=north]{$\theta$};
\draw (7,-2) node[anchor=north]{$\theta$};
\draw (2,1) node[anchor=south]{$id$};
\draw (2,.5) node[anchor=west]{$w^*$};
\draw[ line width=.03cm,
    dash pattern=on .09 cm off .04 cm]  (0,0) arc (180:359:1);   
\draw[ line width=.03cm,
    dash pattern=on .09 cm off .04 cm]  (1,-1) -- (1,-2);
\draw[ line width=.03cm,
    dash pattern=on .09 cm off .04 cm]  (0,0) -- (0,1);   
\draw[ line width=.03cm,
    dash pattern=on .09 cm off .04 cm]  (2,0) -- (2,1);     
\draw[ line width=.03cm,
    dash pattern=on .09 cm off .04 cm]  (7,-2) -- (7,1);     
 \filldraw[color=black, fill=white, thick](2,0.5) circle (.15);    
\end{tikzpicture}
  \caption{\label{fig:X4} Wire diagrams for $w^*x= \theta(w^*)x=1$, which is equivalent to 
  $(w^* \times 1_\theta)x = (1_\theta \times w^*)x$ in DHR notation.}
\end{figure}
\begin{figure}[h!]
\centering
\begin{tikzpicture}[scale=.8]
\draw (2.5,1) node[anchor=west]{$= d$};
\draw (0,-0.1) node[anchor=west]{$x$};
\draw (0,2.2) node[anchor=west]{$x^*$};
\draw (4,-1) node[anchor=north]{$\theta$};
\draw (0,-1) node[anchor=north]{$\theta$};
\draw (0,3) node[anchor=south]{$\theta$};
\draw (1.5,1) node[anchor=east]{$\theta$};
\draw (-1.5,1) node[anchor=west]{$\theta$};
\draw[ line width=.03cm,
    dash pattern=on .09 cm off .04 cm] (0,0) .. controls (1.1,1)  .. (0,2);
\draw[ line width=.03cm,
    dash pattern=on .09 cm off .04 cm] (0,0) .. controls (-1.1,1)  .. (0,2);    
\draw[ line width=.03cm,
    dash pattern=on .09 cm off .04 cm] (0,0) -- (0,-1);
\draw[ line width=.03cm,
    dash pattern=on .09 cm off .04 cm] (0,2) -- (0,3);    
 \draw[ line width=.03cm,
    dash pattern=on .09 cm off .04 cm] (4,-1) -- (4,3);       
\end{tikzpicture}
  \caption{\label{fig:X2} Wire diagram for $x^* x = d1$. 
  }
\end{figure}
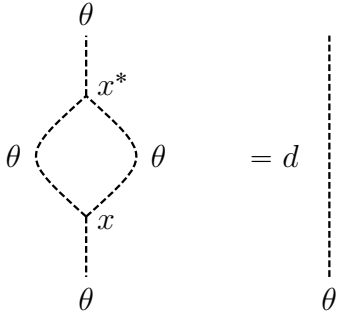

Given a Q-system, one defines an extension $\cM$  as follows. As a set, $\cM$ consists of 
all symbols of the form $nv$, where $n \in \cN$ with the product, $*$-operation, and unit defined by, respectively
\ben
\label{qsystem1}
n_1 v n_2 v = n_1 \theta(n_2) xv, \quad (nv)^* = w^* x^* \theta(n^*) v, \quad 1 = w^* v. 
\een
Associativity and consistency with the $*$-operation follow from the defining relations. 
The conditional expectation is related to the data by $E(nv)=d^{-1} nw$ and is used to induce the operator norm on $\cM$.
Conversely, given an inclusion of infinite (type III) factors $\cN \subset \cM$, the data of the Q-system and $v \in \cM$ can be found
by a canonical procedure as follows. 
If $\i: \cN \to \cM$ is the inclusion map, we can define a conjugate $\bar \i: \cM \to \cN$ where  $r \in \Hom_\cM(\bar \i\i, id)$ and 
$\bar r \in \Hom_\cM(\i \bar \i, id)$ such that \eqref{conjgacyrel} holds with $d=d_{\i} = [\cM:\cN]^{1/2}$, namely
\ben
(\bar r^* \times 1_\i)(1_\i \times r) = 1_\i, \quad ( r^* \times 1_{\bar \i})(1_{\bar \i} \times \bar r) = 1_{\bar \i}, \quad r^* r = d 1 = \bar r^* \bar r
\label{rconjugacy}
\een 
using the notation of the DHR product \eqref{DHRprod}. Here, and throughout these notes, we use the notation
 \ben
 d=[\cM:\cN]^{1/2} \quad (=d_\imath=d_\theta^{1/2})
 \een
 for the square-root of the Jones index.
In terms of $r, \bar r$, the Q-system for $\cN \subset \cM$ is now
\ben
\label{canQsys}
\theta = \bar \i \i, \quad 
x=1_{\bar \i} \times \bar r \times 1_\i = \bar \i(\bar r), \quad w=r.  
\een
The relations for this Q-system follow from the relations \eqref{rconjugacy}.
We work out ``explicitly'' the very well-known case of an inclusion associated with the action of a finite group $G$ in the appendix
as an illustration for the interested reader in app. \ref{appB}.

The defining relations of $\cM$ can also be written in a more suggestive way resembling the operator product expansion (OPE) in QFT. Let 
$\theta \cong \oplus_i \rho_i$ be the decomposition of $\theta$ into irreducible objects, where a given object may appear multiple times. This means that there are isometries $w_i \in \Hom(\rho_i,\theta)$ in $\cN$
such that
\ben
\theta(n) = \sum_i w_i \rho_i(n) w_i^*.
\een
Next, define $\phi_i := w_i^* v \in \cM$. The defining relations of 
a Q-system imply 
\ben
\phi_i \phi_j = \sum_k c_{ij}^k \, \phi_k, \quad
c_{ij}^k := (w_i^* \times w^*_j) x w_k \in \Hom(\rho_i\rho_j,\rho_k) \cap \cN.
\een

\medskip
\noindent
{\bf Braided product of Q-systems.} See \cite{bischoff2015tensor}.
Assume that we have a braiding between the sub-objects $\rho_i$ of $\theta$, we can define a braiding operator of $\theta^2$ as
\ben
\epsilon^\pm(\theta,\theta) := \sum_{i,j} (w_j \times w_i)^* \epsilon^\pm(\rho_i,\rho_j) (w_i \times w_j) \in \Hom(\theta^2, \theta^2). 
\een
A Q-system for a von Neumann algebra $\cA$ is called {\em commutative} iff $\epsilon^\pm(\theta,\theta)x = x$. Given two Q-systems $Q_1 = (\theta_1, x_1, w_1), Q_2=(\theta_2, x_2, w_2)$ and a braiding between the subobjects of $\theta_1, \theta_2$, we can define the {\em braided product Q-system}, which has the data 
\bena
\label{qsystem2}
& (\theta, x^\pm, w) = Q_1 \times^\pm Q_2 \quad
\Longleftrightarrow \\ 
&\theta = \theta_1 \theta_2, \quad x^\pm := 
(1_{\theta_1} \times \epsilon^\pm (\theta_1, \theta_2) \times 1_{\theta_2})
(x_1 \times x_2), \quad
w = w_1 \times w_2. 
\eena 
The braided product Q-system is commutative if $Q_1$ and $Q_2$ both are commutative. The braided product Q-systems define two extension $\cA \subset \cB_{12}^\pm$, which can more 
explicitly be described as follows: $\cB_{12}^\pm$ is generated by $\cA$ and two isometries $v_1, v_2$ satisfying the relations
\ben
v_2 v_1 = \i(\epsilon^\pm(\theta_1,\theta_2)) v_1 v_2, 
\een
in addition to the relations for $v_1, v_2$ analogous to \eqref{qsystem1}. It follows that the linear space $\cB_i^\pm:=\cA v_i, i=1,2$ define von Neumann algebras intermediate 
to the braided product extension in the sense that $\cA \subset \cB_1^\pm, \cB_2^\pm \subset \cB_{12}^\pm$. The braided product extension plays a role in CFT in the context of defects.

\medskip
\noindent
{\bf Conditional expectations and canonical endomorphism.} 
See \cite{longo1995nets}. The minimal conditional expectation $E$ and its dual $E'$ can be described more explicitly using Q-systems. 
In terms of the Q-systems, $\cM$ is generated by $\cN$ together with a single operator, $v$, 
and $\cN'$ is generated by $\cM'$ together with a single operator, $v'$. The operator $v'$ can be defined as follows. 
Let $|\Omega\rangle$ be a cyclic and separating vector for both $\cN, \cM$ (which exists for type III), and let $|\eta\rangle$ be a 
vector such that $\omega_\eta = \omega_\Omega \circ E$. Then $v': n |\Omega\rangle \mapsto n|\eta\rangle$ is 
seen to be an isometry in $\cN'$, and the dual construction is made for $v \in \cM$. 

The operators $w=j_\cN(v') \in \cN$, $w'=j_\cM(v) \in \cM'$ and the ``canonical'' endomorphisms
\ben
\gamma=j_\cN j_\cM : \cM \to \cN, \quad \gamma' = j_\cM j_\cN: \cN' \to \cM'
\een
can be defined, where $j_\cN(n)=J_\cN n J_\cN$ and $J_\cN$ is the modular conjugation\footnote{With respect to a
fixed natural cone ${\mathscr P}^\sharp_\cN$ defined by $|\Omega\rangle$.} of $\cN$, etc. The expectations $E,E'$ are then given by
\ben
E(m)=\frac{1}{d} w^* \gamma(m) w, \quad E'(n')= \frac{1}{d}  w^{\prime *} \gamma'(n') w^\prime.
\een
They have the property that $J_\cM v'= v' J_\cN, J_\cM v= v J_\cN$. The restricted 
canonical endomorphism $\theta$ \eqref{canQsys} (and similarly for the dual inclusion) is given by
\ben
\theta = \gamma |_{\cN} \in {\rm End}(\cN), \quad 
\theta' = \gamma' |_{\cM'} \in {\rm End}(\cM'), 
\een
for a suitable choice of $\bar \i$ (and $\bar \i'$), and for such a choice
\ben
\theta = \bar \i \i, \quad \gamma = \i \bar \i.
\een

\medskip
\noindent
{\bf $\alpha$-induction:}
Let $\cN \subset \cM$ be an irreducible inclusion of subfactors with 
finite index and associated canonical endomorphism $\theta \in \End(\cN)$.
Given a braided unitary fusion category ${}_\cN X_\cN$ and an irreducible endomorphism
$\mu \in {}_\cN X_\cN$ of $\cN$, we can define the $\alpha$-induced endomorphisms (of $\cM$)
\ben
\label{alphadef}
\alpha^\pm_\mu := \overline{\imath}^{-1} \circ {\rm Ad}(\epsilon^\pm(\mu,\theta)) \circ \lambda \circ \overline{\imath}, 
\een
which is an in general reducible endomorphism of $\cM$ (even though $\mu$ is by definition irreducible). If we describe the inclusion $\imath: \cN \to \cM$ by a Q-system
$(\theta=\bar \imath \imath, x, w)$, see sec. \ref{sec:Qsys}, then $\cM$ is spanned linearly by 
elements of the form $nv$ where $n \in \cN$ and where $v$ is the generator with the relations recalled in sec. \ref{sec:Qsys}. 
Then we can write
\ben
\alpha^\pm_\lambda(nv) = \lambda(n) \epsilon^\pm(\lambda,\theta)^* v.
\een

One 
can derive the following naturality/functorial relations:
\begin{enumerate}
    \item (Conjugate) $\alpha^\pm_{\overline \lambda} = \overline{\alpha^\pm_{\lambda}}$,
    \item (Dimension) $d_{\alpha^\pm_{\lambda}} = d_\lambda$, 
    \item (Composition) $\alpha^\pm_{\lambda} \alpha^\pm_{\mu} = \alpha^\pm_{\lambda \mu}$,
    \item (Functoriality 1) If $t \in \Hom(\lambda\mu,\nu)$, then $\imath(t) \in \Hom(\alpha^\pm_{\lambda\mu},\alpha_\nu^\pm)$,
    \item (Braiding) Even though $\cM$ is in general not braided, we have ( as endomorphisms of $\cM$) $\alpha_\mu^\pm \alpha^\pm_\lambda = {\rm Ad}(\imath[ \epsilon^\pm(\lambda,\mu)]) \alpha_\mu^\pm \alpha^\pm_\lambda$,
    \item (Functoriality 2) if $v \in \Hom(\alpha^-_\rho, \alpha^+_\nu), 
    v' \in \Hom(\alpha^-_{\rho'}, \alpha^+_{\nu'})$, then 
    \ben
    \label{funct2}
    \imath[ \epsilon^-(\rho',\rho)] (v' \times v) = 
     (v \times v') \imath[ \epsilon^+(\nu,\nu')].
    \een
\end{enumerate}
The matrix $Z_{\lambda,\mu} := \dim \Hom(\alpha^-_\lambda, \alpha^+_\mu)$
commutes with the matrix $Y_{\mu,\nu}$ and 
\ben
\label{hleftright}
Z_{\lambda,\mu} \neq 0 \quad \Longrightarrow \quad h(\lambda) - h(\mu) \in \bZ.
\een

\subsection{Relative braiding and systems of endomorphisms}
\label{relative}

See \cite{bockenhauer1999alpha,bockenhauer1998modular, bockenhauer1999modular,bockenhauer1999modular1,izumi2000structure}.
When studying a finite index inclusion $\cN \subset \cM$ of von Neumann factors, special systems of endomorphisms often arise. We let 
$\imath: \cN \to \cM$ be the embedding and 
$\overline{\imath}: \cM \to \cN$ be a conjugate endomorphism. We
consider finite sets 
\bena
{}_\cN X_\cN &\subset& \End(\cN,\cN)\, \\
{}_\cM X_\cN &\subset& \End(\cM,\cN)\, \\
{}_\cN X_\cM &\subset& \End(\cN,\cM)\, \\
{}_\cM X_\cM &\subset& \End(\cM,\cM)\,
\eena
of equivalence classes of endomorphisms with the following properties:

\begin{itemize}
\item Any two members of any of the sets ${}_\cN X_\cM, {}_\cM X_\cN, {}_\cM X_\cM, {}_\cN X_\cN$ are mutually inequivalent as endomorphisms, irreducible, 
and have finite index. [The index of $a \in \End(\cM, \cN)$ is defined as $d_a = [\cM: a(\cN)]^{1/2}$.] 

\item ${}_\cN X_\cN$ is a unitary fusion category and so in particular is closed under 
fusion and taking conjugates, and so in particular has a unit, 
the identity endomorphism of $\cN$. Additionally, it is assumed to be non-degenerately braided, so the fusion in $\cN$ is in particular commutative.
Each endomorphism appearing in the decomposition of the canonical endomorphism $\theta = \gamma |_\cN = \bar \i \i$ is required to be 
contained in ${}_\cN X_\cM$, and $\epsilon^\pm(\theta, \theta)x=x$, so the Q-system $(\theta, x, w)$ corresponding to $\cN \subset \cM$ (see sec. \ref{sec:Qsys}) 
is commutative in the terminology introduced above. 
Irreducible objects of $_\cN X_\cN$ will be denoted by lower case Greek letters such as $\mu, \nu, \lambda, \dots$.

\item ${}_\cN X_\cM$ consists of all irreducible endomorphisms $\bar b$ (without multiplicities) appearing in the decomposition of $\lambda{\overline{\imath}}$, where $\lambda
\in{} _\cN X_\cN$.

\item ${}_\cM X_\cN$ consists of all irreducible endomorphisms $a$ (without multiplicities) appearing in the decomposition of $\imath \lambda$, where $\lambda \in {}_\cN X_\cN$.

\item ${}_\cM X_\cM$ consists of all irreducible endomorphisms $B$ (without multiplicities) appearing in the decomposition of $\i \lambda {\overline{\imath}}$, where $\lambda \in {}_\cN X_\cN$. 
Note that by the other assumptions, ${}_\cM X_\cM$ is by itself a unitary fusion category. But it need not have have a braiding, for example, so the assumptions on 
${}_\cM X_\cM$ respectively ${}_\cN X_\cN$ are not symmetrical. 
\end{itemize}

Even though the fusion of general endomorphisms of $\cM$ may not be commutative (so in particular not braided), we 
can define a kind of relative braiding between endomorphisms from the sets $\, {}_\cN X_\cM, \, {}_\cM X_\cN, \, {}_\cN X_\cN$ with the 
alpha-induced endomorphisms in $\, {}_\cM X_\cM$. These 
braiding operators are denoted by
\ben
\label{epsrel}
\begin{split}
\epsilon^\pm(\lambda,\bar b) &\in \Hom(\bar b \alpha_\lambda^\pm, \lambda \bar b)\\
\epsilon^\pm(a,\lambda) &\in \Hom(\alpha_\lambda^\pm a, a\lambda), 
\end{split}
\een
see \cite{bockenhauer1999alpha} sec. 3.3 for the definitions and proofs. Recall that we have defined braiding operators for the alpha induced endomorphisms $\alpha^\pm_\lambda$ above in 
\eqref{alphadef}. Together with \eqref{epsrel}, these satisfy the expected braiding-fusion (BF) relations, used throughout later parts of this paper, often implicitly when manipulating diagrams.
The wire diagrams for the relative braiding intertwiners \eqref{epsrel} are depicted in fig. \ref{fig:220426_Fig-7_TE}. Our conventions for the wires are, see fig. \ref{fig:lines},
\begin{figure}
\centering
\begin{tikzpicture}[scale=.6]
\draw (0,1) node[anchor=west]{$\beta$};
\draw (4,1) node[anchor=west]{$b$};
\draw (8,1) node[anchor=west]{$B$};
\draw[ line width=.03cm,
    dash pattern=on .09 cm off .04 cm]  (0,0) -- (0,2);
\draw[ thin ]  (4,0) -- (4,2);
\draw[ very thick]  (8,0) -- (8,2) ;
\end{tikzpicture}
  \caption{\label{fig:lines} Types of lines.}
\end{figure}
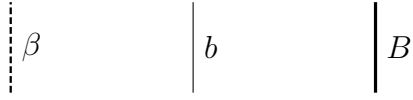
\begin{itemize}
\item Thick solid lines: Endomorphisms $A,B,\dots$ from ${}_\cM X_\cM$.
\item Thin solid lines: Endomorphisms $a,b, \dots$ or $\bar a, \bar b, \dots$ from ${}_\cM X_\cN$ or ${}_\cN X_\cM$.
\item Dashed lines: Endomorphisms $\mu, \lambda, \dots$ from ${}_\cN X_\cN$.
\end{itemize}
\begin{figure}[h!]
\begin{center}
  \includegraphics[width=1.1\textwidth,]{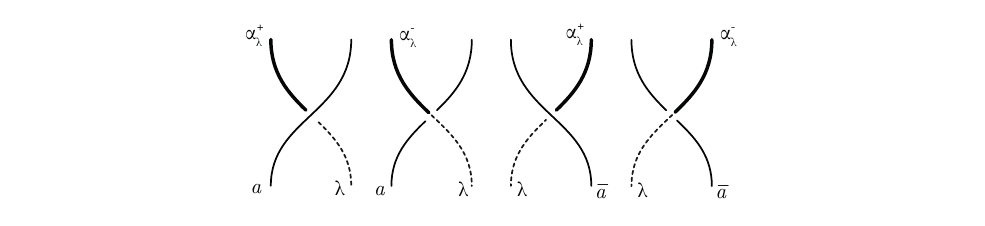}
\end{center}
  \caption{\label{fig:220426_Fig-7_TE} Wire diagrams for $\epsilon^+(a,\lambda), \epsilon^-(a,\lambda), \epsilon^+(\lambda,\bar a), \epsilon^-(\lambda,\bar a)$ from left to right.}
\end{figure}
Using these conventions, the braiding fusion (BF) relatins are depicted in fig. \ref{fig:220426_Fig-8_TE}.
\begin{figure}[h!]
\begin{center}
  \includegraphics[width=1.1\textwidth,]{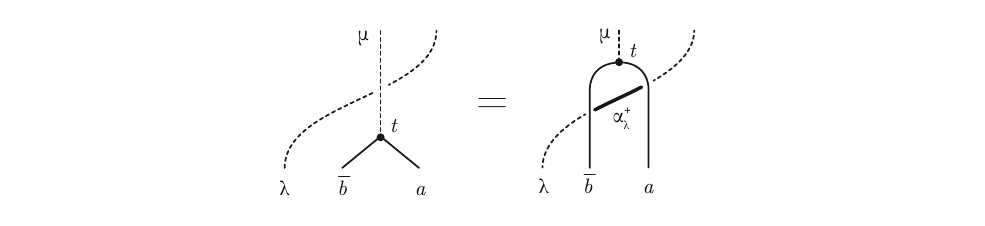}
\end{center}
  \caption{\label{fig:220426_Fig-8_TE} Topological invariance of wire diagram representing a BF relation with relative braiding. There is a similar BF relation 
  with crossing from right to left.}
\end{figure}

\subsection{Jones basic construction and Jones projections}

See e.g. \cite{kosaki1986extension, hiai1988minimizing, longo1990index, longo1994duality, longo1995nets} for the extension of Jones' theory 
\cite{jones1983index} to type III algebras. 
Consider a finite index inclusion $\cN \subset \cM$ of type III von Neumann factors in standard form with conditional expectation $E:\cM \to \cN$.
Let $\omega$ be a faithful normal state with cyclic and separating vector $|\Omega\rangle$ for both $\cN$ and $\cM$
implementing $\omega$. Then $\omega \circ E$ is invariant under $E$, with corresponding vector representative $|\eta\rangle$. 
Note that $|\eta\rangle$ is not cyclic. $e_1 := e_{\cN}:= [\cN \eta] \in \cN'$ is called the Jones projection associated with the inclusion and 
one defines $\cM_1 := \cM \vee \{ e_{\cN} \}$, leading to a new inclusion $\cM \subset \cM_1$. This process is iterated setting $e_2 := e_{\cM} = 
[\cM \eta] \in \cM'$ and $\cM_2 := \cM_1 \vee \{ e_{\cM} \}$, then $e_3 := e_{\cM_1} = 
[\cM_1 \eta] \in \cM_1'$ and $\cM_3 := \cM_2 \vee \{ e_{\cM_1} \}$ etc. This gives the Jones tower
\ben
\cN \subset \cM \subset_{e_1} \cM_1 \subset_{e_2} \cM_2 \subset_{e_3} \cM_3 \subset \dots .
\een
One also defines the corresponding Jones tunnel, e.g. by setting
\ben
\cN_1 = j_\cN j_\cM(\cM), \quad \cN_2 = j_{\cN_1} j_\cN(\cN), \quad \cN_3 = j_{\cN_2} j_{\cN_1}(\cN_1), \dots 
\een
which gives
\ben
\dots \cN_3 \subset_{e_{-3}} \cN_2 \subset_{e_{-2}} \cN_1 \subset_{e_{-1}} \cN \subset_{e_0} \cM.
\een
Alternatively, one can construct the Jones tunnel by forming the commutant of the Jones tower for the dual inclusion $\cM' \subset \cN'$.
$e_0 = e_{\cN_1}$ is the projection extending $\cN$ to $\cM$, $e_{-1} = e_{\cN_2}$ is that extending $\cN_1$ to $\cN$, 
$e_{-2} = e_{\cN_3}$ is that extending $\cN_2$ to $\cN_1$ etc.  The maps $j_\cN j_\cM = j_{\cN_1} j_\cN = j_{\cN_2} j_{\cN_1} = \dots$ correspond to a 2-shift of the tunnel to the left, 
so establish that the inclusions $\cN \subset \cM, \cN_2 \subset \cN_1, \dots$ are all isomorphic. The same applies to the inclusions 
$\cN_1 \subset \cN_2, \cN_3 \subset \cN_4, \dots$. One thereby sees that the endomorphism $\theta = j_\cN j_\cM$ is a leftwards 
2-shift\footnote{A further linear map corresponding to a 1-shift is given by the quantum Fourier transform which is relevant also in the 
context of the double triangle algebra but not discussed here, see \cite{bisch1997bimodules}.} of the 
even part of the tunnel, giving $\theta^k(\cN) = \cN_{2k}$. Likewise $\gamma = j_{\cM} j_{\cM_1}$ is a leftwards 2-shift of the odd part of the tunnel, 
giving $\gamma^k(\cM) = \cN_{2k-1}$.  

Remembering that the conjugate endomorphism ${\overline{\imath}}$ of the embedding $\imath$ may be chosen such that 
$\gamma =\imath  {\overline{\imath}}$, $\theta = \overline{\imath} \imath$, we therefore get 
\ben
\begin{split}
\cN_1 =& \imath {\overline{\imath}}  (\cM), \\
\cN_2 =& \imath {\overline{\imath}} \imath (\cN), \\
\cN_3 =& \imath {\overline{\imath}} \imath{\overline{\imath}} (\cM),\\
\cN_4 =& \imath {\overline{\imath}} \imath{\overline{\imath}} \i (\cN),\\
&\dots,
\end{split}
\een
and so on. The inclusion $\cN \subset \cM$ is said to be of finite depth if in the subsequent decompositions of $\imath {\overline{\imath}} \i {\overline{\imath}} \i {\overline{\imath}} \cdots$
into (equivalence classes of) irreducible endomorphisms $\mu \in \End(\cM)$, no new irreducible endomorphisms 
appear after some ``depth'' $k$, and this is implied by our 
standing assumptions formulated in sec. \ref{relative}. (This condition is 
independent of the condition $d<\infty$ of finite index.) 

It is possible to obtain more ``concrete'' expressions for the Jones projections of the tunnel in terms $r \in \Hom_\cM(\bar \i\i, id)$ and 
$\bar r \in \Hom_\cM(\i \bar \i, id)$ appearing in the the conjugacy relations \eqref{rconjugacy} associated with $\i, \bar \i$ as follows. 
First, one can derive the dual identities 
\ben
e_{\cN_1} = d^{-1} \, vv^* = e_{\cM'}, \quad 
e_{\cM_1'} = d^{-1} \, v' v^{\prime *} = e_{\cN}.
\een
Now, $e_{\cN_1} = e_0$ whereas $j_\cM j_{\cM_1}(e_{\cM_1'}) = j_\cN j_{\cM}(e_{\cM_1'}) = j_\cN(e_{\cM_1'}) = e_{-1}$. 
On the other hand, we can show that $v = \bar r$ (by using \eqref{rconjugacy}), and we have already mentioned that 
$j_\cN(v') = w$, where $v,w$ refer to the Q-system for the extension $\cN \subset \cM$, and $v',w'$ to the dual 
extension $\cM' \subset \cN'$. Finally $w = \i(r)$, and together 
this gives the first two  of the following formulas. The remaining formulas follow by the above observation that 
$\theta = \bar \i \i$ ($=$ DHR left mutliplication by $1_{\bar \i} \times 1_\i \times \dots$) respectively $\gamma=\i \bar \i$ (($=$ DHR left mutliplication by $1_{\i} \times 1_{\bar \i} \times \dots$))
represent a leftward 2-shift of the Jones tunnel:
\ben
\begin{split}
e_0 &= d^{-1} \, \bar r \bar r^* \times 1_{\i} \times 1_{\bar \i} \times \cdots \in \cM \cap \cN_1' \\ 
e_{-1} &= d^{-1} \, 1_{\i} \times rr^* \times 1_{\bar \i} \times 1_{\i} \times \cdots  \in \cN \cap \cN_2' \\
e_{-2} &= d^{-1} \, 1_{\i} \times 1_{\bar \i } \times \bar r \bar r^* \times 1_{\i} \times 1_{\bar \i} \times \cdots \in \cN_1 \cap \cN_3' \\
e_{-3} &=  d^{-1} \, 1_{\i} \times 1_{\bar \i } \times 1_\i \times rr^* \times 1_{\bar \i} \times 1_{\i} \times \cdots \in \cN_2 \cap \cN_4',\\
&\dots ,
\end{split}
\een
and so on, where $\times$ is the DHR product, see fig. \ref{fig:ej}. Note the alternating pattern of $\bar \i, \i$ and $\bar r, r$. 
The identity  intertwiners $1_\i$ and $1_{\bar \i}$ are inserted to the right in the DHR products to match the vertical lines to the right of the cup-cap pairs in 
the wire diagram fig. \ref{fig:ej}, but they do not affect the actual value of $e_{-j}$.
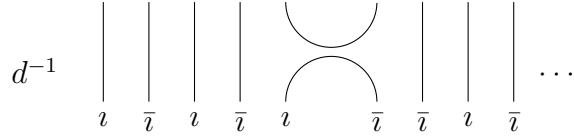
\begin{figure}
\centering
\begin{tikzpicture}[scale=.6]
\draw (-2,0) node[anchor=north]{$ \i$};
\draw (0,0) node[anchor=north]{$\bar \i$};
\draw[thin] (0,0) arc (0:180:1);
\draw[thin] (-2,2.2) arc (180:359:1);
\draw[thin] (-3,0) -- (-3,2.2);
\draw[thin] (-4,0) -- (-4,2.2);
\draw[thin] (-5,0) -- (-5,2.2);
\draw[thin] (-6,0) -- (-6,2.2);
\draw (-3,0) node[anchor=north]{$\bar \i$};
\draw (-4,0) node[anchor=north]{$\i$};
\draw (-5,0) node[anchor=north]{$\bar \i$};
\draw (-6,0) node[anchor=north]{$\i$};
\draw[thin] (1,0) -- (1,2.2);
\draw[thin] (2,0) -- (2,2.2);
\draw[thin] (3,0) -- (3,2.2);
\draw (4,1) node[anchor=north]{$\cdots$};
\draw (-7.5,1.2) node[anchor=north]{$d^{-1}$};
\draw (1,0) node[anchor=north]{$\bar \i$};
\draw (2,0) node[anchor=north]{$ \i$};
\draw (3,0) node[anchor=north]{$\bar \i$};
\end{tikzpicture}
  \caption{\label{fig:ej} Wire diagram for $e_{-4}$ using the wire diagram for $\bar r, \bar r^*$ as in fig. \ref{fig:r} with $\lambda \to \i$.}
\end{figure}
Either from \eqref{rconjugacy}, or using the diagrams 
fig. \ref{fig:220426_Fig-4_TE} (applied to $\i, \bar \i$ instead of $\lambda, \bar \lambda$) and the value $d=d_\i = d_{\bar \i}$ for the circle, fig. \ref{fig:circle}, 
one then gets the Temperly-Lieb-Jones relations \eqref{TLJ}.  

\subsection{Higher relative commutants and paths}

See e.g. \cite{bisch1997bimodules,evans1998quantum} who mainly consider type II case and use bimodule language.
The intersections $\cM \cap \cN_k'$ and $\cN \cap \cN'_{k+1}$ are called relative commutants. 
For $k=0$ they are trivial for an irreducible inclusion $\cN \subset \cM$ of factors, and for $k>0$ they are finite-dimensional matrix algebras if $d<\infty$. 
The latter can be seen by giving the relative commutants an ``explicit'' description in terms of intertwining operators. 

We first consider the first non-trivial relative commutant $\cN \cap \cN_2'$. 
Consider the decomposition of the endomorphism $\theta = \bar \i \i \in \End(\cN)$ into irreducibles $\mu$ as in \eqref{decomp}
using an ONB of intertwiners $t_{\mu,j} \in \Hom_\cM(\theta, \mu)$. 
Let $n_{\mu j}^i := t_{\mu,i} (t_{\mu,j})^*$. These are matrix units
\ben
n_{\mu j}^i n_{\nu k}^l = \delta_j^l \delta_{\mu,\nu} n_{\mu k}^i.
\een
By construction $n_{\nu j}^i \in \cN \cap \cN_2'$ because they are elements of $\cN$ commuting with $\cN_2 = \theta(\cN)$ since they are in 
$\Hom_\cN(\theta,\theta)$. In fact these matrix units generate $\cN \cap \cN_2'$, so 
\ben
\cN \cap \cN_2' \cong \bigoplus_{\mu \subset \theta} M_{N_\mu}(\CC).
\een
We may apply the analogous reasoning to the relative commutant $\cM \cap \cN_1'$ using $\cN_1 = \gamma(\cM)$ 
and a decomposition of $\gamma$ into irreducibles $M$ with intertwiners $t_{M,j} \in \Hom(\gamma, M)$. This gives
\ben
\cM \cap \cN_1' \cong \bigoplus_{M \subset \gamma} M_{N_{M}}(\CC).
\een
Bases for the higher relative commutants $\cM \cap \cN_{2k+1}'$ resp. $\cN \cap \cN_{2k}'$, etc. are obtained by considering 
the ONBs of intertwiners appearing in subsequent decompositions of $\i {\overline{\imath}} \i {\overline{\imath}} \cdots \i$ resp. 
${\overline{\imath}} \i {\overline{\imath}} \i \cdots {\overline{\imath}}$, etc. We describe the relative commutants $\cM \cap \cN_{2k-1}'$, the other cases are similar. 

By definition, an element of $\cM \cap \cN_{2k-1}'$ is an element of $\cM$ that is an intertwiner in the space $\Hom((\i \bar \i)^k, (\i \bar \i)^k)$.
We produce such intertwiners as follows. First, we decompose $\i \bar \i$ into 
irreducibles $M_1$ an ONB of intertwiners $\{t_1\} \subset \Hom(\i {\overline{\imath}},M_1)$ (more generally we could start with 
$a_0 \in \,  _\cM X_\cN$ and decompose $a_0 \bar \i$).
Next we multiply by $\i$ from the right, and similarly consider an ONB of intertwiners $\{t_2\} \subset \Hom(M_1 \i, a_2)$, 
after which we multiply by ${\overline{\imath}}$ from the right, and consider an orthonormal set of intertwiners $\{t_3\} \subset \Hom(a_2 {\overline{\imath}},M_3)$, and so on
until $\{t_{2k-1}\} \subset \Hom(a_{2k-2} \bar \i,M_{2k-1})$. We denote the space of such sequences of 
isometric intertwiners $(t_1, t_2, \dots, t_{2k-1})$ by ${\rm Path}^{2k-1}_{a_0,M_{2k-1}}$. [The subscript $(a_0,M_{2k-1})$ means 
that we start with the object $a_0$ and end with $M_{2k-1}$.]
Then
\ben
\label{tP}
t_P:=(t_1 \times 1_{\i} \times 1_{\overline{\imath}} \times \cdots 1_{\overline{\imath}})
\cdots
(t_{2k-3} \times 1_{\i} \times 1_{\bar \i})
(t_{2k-2} \times 1_{\bar \i})
t_{2k-1} \in \Hom(a_0 (\bar \i \i)^{k-1} \bar \i, M_{2k-1} ),
\een
where $P$ can be thought of as a ``path label'' denoting a compatible sequence of 
orthonormal intertwiners with suitable source and target endomorphisms. 

By construction we have (generalizing \eqref{decomp})
\ben
t_P^* t_{P'}^{} = \delta_{P,P'}1, \quad \sum_P t_P^{} t_P^* = 1. 
\een
The desired basis of $\cM \cap \cN_{2k-11}'$ is then 
$\{t_{P'} t_P^* : P,P' \in 
{\rm Path}^{2k-1}_{\imath,M_{2k-1}} \}$, i.e. the matrix units are labelled by 
pairs of paths $P,P' \in {\rm Path}^{2k-1}_{\imath,M_{2k-1}}$ 
with the same final object, $M_{2k-1} = M_{2k-1}'$. 

By fairly obvious variations of the above construction, we could have ended instead with an  $a_{2k} \in {}_\cM X_\cN$ after $2k$ 
decompositions, or we could have started with $M_0 \in {}_\cM X_\cM$, or both. The corresponding path spaces will be denoted 
accordingly, and this would be related to the other higher relative commutants.  

We obviously have a freedom in which order we perform the subsequent decompositions of $\i {\overline{\imath}} \i {\overline{\imath}} \cdots \i$, and a different 
order of the decomposition gives a different basis, e.g. of $\cM \cap \cN_{2k-1}'$. As in the classical case of group representation, 
we can pass back and fourth between these bases via $6j$-symbols, which are described below. 

\subsection{$6j$-symbols}
See e.g. \cite{evans1998quantum} (ch. 10, 11) or \cite{kawahigashi2020remark,kawahigashi2021projector} for the type II case. 
These references use bimodule language, which can be translated to sectors
as outlined in \cite{evans1998quantum}, sec. 10.8., but with not many details given.
We now consider such $6j$-symbols (also called quantum $6j$-symbols, bi-unitary connections or F-symbols depending on the literature) 
and discuss some of their properties needed in the sequel. We consider sets 
of endomorphism ${}_\cN X_\cN, {}_\cN X_\cM, {}_\cM X_\cN, {}_\cM X_\cM$ as in sec. \ref{relative}.
Let $B, M_1, M_2 \in \, {}_\cM X_\cM$, $a_1, a_2 \in {}_\cM X_\cN$. 
Then we consider the following two ways of decomposing $B a_1 {\overline{\imath}}$ as in $(B a_1) {\overline{\imath}}$ or $B (a_1 {\overline{\imath}})$. 
First, we pick an ONB of intertwiners $t_1 \in \Hom_\cM(a_1 {\overline{\imath}}, M_1)$ and subsequently an ONB of intertwiners
$t_2 \in \Hom_\cM(BM_1 , M_2)$. We get an intertwiner
\ben
(1_B \times t_1)
t_2 \in \Hom_\cM(B a_1 {\overline{\imath}}, M_2).
\een
Second, we pick an ONB of intertwiners $t_3 \in \Hom_\cM(B a_1,a_2)$ and subsequently an ONB of intertwiners
 $t_4 \in \Hom_\cM(a_2 {\overline{\imath}}, M_2)$. We get an intertwiner
\ben
(t_3 \times 1_{\overline{\imath}})
t_4 \in \Hom_\cM(B a_1 {\overline{\imath}}, M_2).
\een
The intertwiner
\ben
U_B
\begin{pmatrix}
& t_1 & \\
t_3 & & t_2 \\
& t_4 &
\end{pmatrix}
:=
[(t_3 \times 1_{\overline{\imath}})t_4]^*
(1_B \times t_1)t_2 \in \Hom_\cM(M_3, M_3) = \CC 1
\een
is a multiple of the identity and identified with a scalar $\CC$. It is called a $6j$-symbol, for a wire diagram see fig. \ref{fig:1}

\begin{figure}
\begin{center}
\hspace*{-1cm}
  \includegraphics[width=1.2\textwidth,]{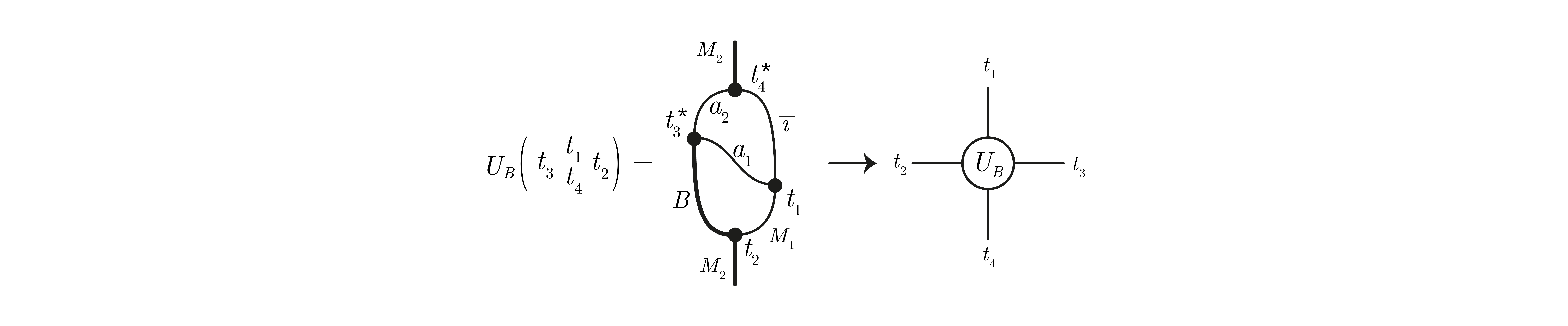}
  \end{center}
  \caption{\label{fig:1} The $6j$-symbol $U_B$.}
\end{figure}

Similarly, let $t_1 \in \Hom(M_1 \imath, a_1), t_2 \in \Hom(BM_1, M_2), 
t_3 \in \Hom(Ba_1, a_2), t_4 \in \Hom(M_2 \imath, a_2)$. Then we set
\ben
\overline U_B
\begin{pmatrix}
& t_1 & \\
t_2 & & t_3 \\
& t_4 &
\end{pmatrix}
:=
[(t_2 \times 1_{\i})t_4]^*
(1_B \times t_1)t_3 \in \Hom_\cM(a_2, a_2) = \CC 1. 
\een
\begin{lemma}\label{lem:6j}
These $6j$-symbols have a number of properties:
\begin{enumerate}
\item (Unit): Writing $0=id$ for the identity endomorphism, we have
\ben
U_0
\begin{pmatrix}
& t_1 & \\
t_3 & & t_2 \\
& t_4 &
\end{pmatrix}
= \delta_{t_2,1} \delta_{t_3,1} \delta_{t_1,t_4}.
\een
\item (Unitarity): We have
\ben
\sum_{t_1, t_2}
U_B
\begin{pmatrix}
& t_1 & \\
t_3 & & t_2 \\
& t_4 &
\end{pmatrix}
 \overline {U_B
\begin{pmatrix}
& t_1 & \\
t_3' & & t_2 \\
& t_4' &
\end{pmatrix}}
= \delta_{t_4,t_4'} \delta_{t_3,t_3'}
\een
as well as 
\ben
\sum_{t_3, t_4}
U_B
\begin{pmatrix}
& t_1 & \\
t_3 & & t_2 \\
& t_4 &
\end{pmatrix}
 \overline {U_B
\begin{pmatrix}
& t_1' & \\
t_3 & & t_2' \\
& t_4 &
\end{pmatrix}}
= \delta_{t_1,t_1'} \delta_{t_2,t_2'}
\een
where the sums are over an ONB of intertwiners with the appropriate source and target endomorphisms. 
The contragredient $6j$-symbol is also unitary.
\item (Conjugate): We have
\ben\label{conjugate}
U_{\bar B}
\begin{pmatrix}
& t_4 & \\
\tilde t_3 & & \tilde t_2 \\
& t_1 &
\end{pmatrix}
=
\left[ \frac{d(a_1) d(M_2)}{d(a_2) d(M_1)} \right]^{1/2} \overline {U_B
\begin{pmatrix}
& t_1 & \\
t_3 & & t_2 \\
& t_4 &
\end{pmatrix}},
\een
and similarly
\ben
\overline U_{\bar B}
\begin{pmatrix}
& t_4 & \\
\tilde t_3 & & \tilde t_2 \\
& t_1 &
\end{pmatrix}
=
\left[ \frac{d(M_1) d(a_2)}{d(M_2) d(a_1)} \right]^{1/2} \overline {
\overline U_B
\begin{pmatrix}
& t_1 & \\
t_3 & & t_2 \\
& t_4 &
\end{pmatrix}}.
\een
\end{enumerate}
\end{lemma}
\begin{proof}
(Unit) This follows from the ONB properties of the intertwiners and the irreducibility of $M_i, a_i, B, {\overline{\imath}}, \i$. 

(Unitarity) This follows from the ONB properties of the intertwiners, and their Frobenius duals for the 
contragredient $6j$-symbols. 

(Conjugate) The reader is invited to carry out the following steps in 
a graphical manner. 
We begin with the definition of $U_{\bar B}$ and the Frobenius 
dual intertwiners $\tilde t_2, \tilde t_3$, which allows us to write
\ben\label{duality1}
\begin{split}
&U_{\bar B}
\begin{pmatrix}
& t_4 & \\
\tilde t_3 & & \tilde t_2 \\
& t_1 &
\end{pmatrix}\\
=& \left[ \frac{d(a_1)}{d(a_2) d(M_1) d(M_2)} \right]^{1/2}
r^*_{M_1}(1_{\bar M_1} \times r_B^* \times 1_{M_1})x
(1_{\bar M_1 \bar B} \times t_2^*)(1_{\bar M_1} \times r_B \times 1_{M_1})r_{M_1}. 
\end{split}
\een
Here we defined $x \in \Hom(BM_1, M_2)$ as 
\ben
x:= 
(1_B \times t_1^*)
(t_3 \times 1_{\bar \imath})
t_4.
\een
Next we insert a summation over an ONB $s$ of 
$\Hom(BM_1,M_3)$ so that the conjugate intertwiner $\bar s$
runs over an ONB of $\Hom(\bar M_1 \bar B, \bar M_3)$. This gives us 
\ben
\left[ \frac{d(a_1)}{d(a_2) d(M_1) d(M_2)} \right]^{1/2}
\sum_{s,M_3} r^*_{M_1}(1_{\bar M_1} \times r_B^* \times 1_{M_1})x
(\bar s \times 1_{M_2})
(\bar s^* \times t_2^*)(1_{\bar M_1} \times r_B \times 1_{M_1})r_{M_1}. 
\een
The last three factors in parenthesis yield a $\delta_{s,t_2}$ using 
the definition and isometric property of the conjugate intertwiner, so 
the summation collapses to $[\dots]^{1/2} 
r^*_{M_1}(1_{\bar M_1} \times r_B^* \times 1_{M_1})(\bar t_2 \times x)r_{M_2}$. Next we insert again a summation over an ONB $s$ of 
$\Hom(BM_1,M_3)$, turning this into
\ben
\left[ \frac{d(a_1)}{d(a_2) d(M_1) d(M_2)} \right]^{1/2}
\sum_{s,M_3} r^*_{M_1}(1_{\bar M_1} \times r_B^* \times 1_{M_1})
(\bar t_2 \times s)(1_{\bar M_2} \times s^* x) r_{M_2}.
\een
We get a $\delta_{s,t_2}$ for the same reason as before, 
so the summation reduces to
\ben
\left[ \frac{d(a_1)}{d(a_2) d(M_1) d(M_2)} \right]^{1/2}
r_{M_2}^*(1_{\bar M_2} \times t_2^* x)r_{M_2}
=
\left[ \frac{d(a_1) d(M_2)}{d(a_2) d(M_1)} \right]^{1/2}
\overline{
U_{B}
\begin{pmatrix}
& t_1 & \\
t_3 & & t_2 \\
& t_4 &
\end{pmatrix}
}. 
\een
The relation for $\overline U_B$ is demonstrated in the same manner. Note that the intertwiners are not from the same spaces here as in the case of 
$U_B$ and consequently we get a different prefactor.
\end{proof}

\subsection{Double triangle algebra}

See \cite{bockenhauer1999alpha}, which is partly based on ideas by Ocneanu \cite{ocneanu1991quantum, ocneanu1988quantized}.
 Let $\cN \subset \cM$ 
be an inclusion of infinite (type III) factors with finite index and finite depth. We consider finite sets ${}_\cN X_\cM, {}_\cM X_\cN, {}_\cM X_\cM, {}_\cN X_\cN$ 
of endomorphisms with the  properties described in 
sec. \ref{relative}, where $\i$ is the embedding endomorphism from $\cN \to \cM$ and $\bar \i$ a conjugate endomorphism from $\cM \to \cN$.
Note that $a \in {}_\cM X_\cN$ implies that $\bar a \in {}_\cN X_\cM$.
\begin{definition}
As a finite dimensional vector space, the double triangle algebra is defined by 
\ben
\lozenge = \bigoplus_{a,b,c,d \in {}_\cM X_\cN} \Hom_\cM(c \bar d, a \bar b).
\een
\end{definition}
We note that a given intertwiner might appear in multiple spaces and is considered as different in such a case. It follows from the definition that 
a basis of $\lozenge$ is given by the elements
\ben
e_{B; t,b,a}^{s,d,c} :=  (d_a d_b d_c d_d)^{1/4} st^*, \
\quad 
s \in \Hom_\cM(c \bar d, B), t \in \Hom_\cM(a \bar b, B)
\een
where $s,t$ run through an ONB of intertwiners. Our conventions for graphically representing such generators are described in fig. \ref{fig:dta}.

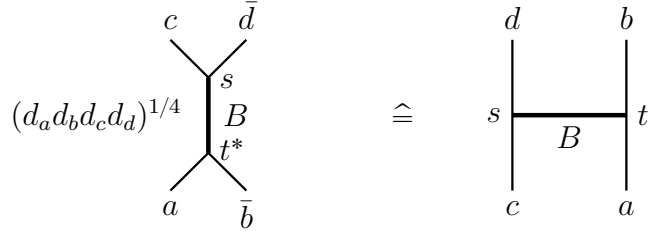
\begin{figure}[h!]
\centering
\begin{tikzpicture}[scale=.5]
\draw (-5.5,1) node[anchor=west]{$(d_a d_b d_c d_d)^{1/4}$};
\draw (4.5,1) node[anchor=west]{$\widehat{=}$};
\draw (0,.1) node[anchor=west]{$t^*$};
\draw (0,1.9) node[anchor=west]{$s$};
\draw (-1,-1) node[anchor=north]{$a$};
\draw (-1,3) node[anchor=south]{$c$};
\draw (1,-1) node[anchor=north]{$\bar b$};
\draw (1,3) node[anchor=south]{$\bar d$};
\draw (.1,1) node[anchor=west]{$B$};
\draw (8,-1) node[anchor=north]{$c$};
\draw (8,3) node[anchor=south]{$d$};
\draw (11,-1) node[anchor=north]{$a$};
\draw (11,3) node[anchor=south]{$b$};
\draw (9.5,1) node[anchor=north]{$B$};

\draw (8,1) node[anchor=east]{$s$};
\draw (11,1) node[anchor=west]{$t$};
\draw[ line width=.03cm] (0,0) -- (-1,-1);
\draw[ line width=.03cm] (0,0) -- (1,-1);    
\draw[ line width=.06cm] (0,0) -- (0,2);
\draw[ line width=.03cm] (0,2) -- (-1,3);    
\draw[ line width=.03cm] (0,2) -- (1,3);    
\draw[ line width=.03cm] (8,-1) -- (8,3);
\draw[ line width=.03cm] (11,-1) -- (11,3);
\draw[ line width=.06cm] (8,1) -- (11,1);
\end{tikzpicture}
  \caption{\label{fig:dta}
Left side: wire diagram for a basis element in $\lozenge$, read from bottom to top. Right side: compared to the left diagram, it is rotated by 90 degrees in agreement with 
the conventions of \cite{bockenhauer1999alpha} and stripped off the prefactor. Their conventions 
are adapted in the following in connection with the double triangle algebras 
and match the conventions for MPOs and spin chains, which are typically drawn horizontally.
  }
\end{figure}
The product structure $\star$ of $\lozenge$ is defined by:
\ben
e_{B; t,b,a}^{s,b',a'} \star e_{B'; t',d,c}^{s',d',c'} := \delta_{b',c'} \delta_{b,c}  
(1_{a'} \times r_{b'}^* \times 1_{ \bar d'})
(e_{B; t,b,a}^{s,b',a'} \times e_{B'; t',a,c}^{s',a',c'})
(1_{a} \times r_b \times 1_{\bar d}).
\een
Here $r_a \in \Hom_\cN(\bar aa, id)$ are solutions to the conjugacy relations normalized so that $r_a^* r_a = d_a 1$. The product structure will be depicted
graphically by wire diagrams such as fig. \ref{prod3}. 
%
%
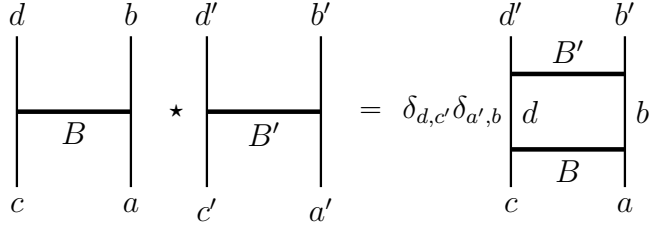
\begin{figure}[h!]
\centering
\begin{tikzpicture}[scale=.5]
\draw (0,-1) node[anchor=north]{$c$};
\draw (0,3) node[anchor=south]{$d$};
\draw (3,-1) node[anchor=north]{$a$};
\draw (3,3) node[anchor=south]{$b$};
\draw (1.5,1) node[anchor=north]{$B$};
\draw[ line width=.03cm] (0,-1) -- (0,3);
\draw[ line width=.03cm] (3,-1) -- (3,3);
\draw[ line width=.06cm] (0,1) -- (3,1);
\draw (3.7,1) node[anchor=west]{$\star$};
\draw (5,-1) node[anchor=north]{$c'$};
\draw (5,3) node[anchor=south]{$d'$};
\draw (8,-1) node[anchor=north]{$a'$};
\draw (8,3) node[anchor=south]{$b'$};
\draw (6.5,1) node[anchor=north]{$B'$};
\draw[ line width=.03cm] (5,-1) -- (5,3);
\draw[ line width=.03cm] (8,-1) -- (8,3);
\draw[ line width=.06cm] (5,1) -- (8,1);
\draw (8.7,1) node[anchor=west]{$= \, \, \delta_{d,c'} \delta_{a',b}$};
\draw (13,-1) node[anchor=north]{$c$};
\draw (13,3) node[anchor=south]{$d'$};
\draw (16,-1) node[anchor=north]{$a$};
\draw (16,3) node[anchor=south]{$b'$};
\draw (14.5,0) node[anchor=north]{$B$};
\draw (14.5,2) node[anchor=south]{$B'$};
\draw (13,1) node[anchor=west]{$d$};
\draw (16,1) node[anchor=west]{$b$};
\draw[ line width=.03cm] (13,-1) -- (13,3);
\draw[ line width=.03cm] (16,-1) -- (16,3);
\draw[ line width=.06cm] (13,0) -- (16,0);
\draw[ line width=.06cm] (13,2) -- (16,2);
\end{tikzpicture}
\caption{The (vertical) product $\star$ in $\lozenge$.}
\label{prod3}
\end{figure}

The unit of $\lozenge$ with respect to the above product structure is given by $\oplus_{a} \bar r_a^{} \bar r_a^*$,
and the structure constants of the double triangle algebra may be obtained by expanding the right side in the the given basis using the intertwiner calculus. 

In the literature $\star$ is called the ``vertical'' product. A ``horizontal'' product $\cdot$
may be defined by simply using the product structure on (compatible) intertwiners
induced by the algebra structure of $\cM$, i.e. by 
$
e_{B; t,b,a}^{s,b',a'} \cdot e_{B'; t',d,c}^{s',d',c'} := \delta_{a,c'} \delta_{b,d'}  e_{B; t,b,a}^{s,b',a'} e_{B'; t',a,c}^{s',a',c'}.
$
Although this will not be used in this work, we mention that 
the horizontal and vertical products are related by the ``quantum Fourier transform'' \cite{bisch1997bimodules} in a similar way as the pointwise product and convolution 
of ordinary functions are related by the standard Fourier transfrom. 
\begin{definition}
$\cZ_h$ is the center of $\lozenge$ with respect to the horizontal product.
\end{definition}
\cite{bockenhauer1999alpha} have analyzed $\cZ_h$ in terms of the braiding and fusion relations in 
$_\cN X_\cN$. As a first result we quote the following. We define
\ben
\label{eBdef}
e_B := \sum_{t,a,b} e^{t,b,a}_{B;t,b,a} =  \bigoplus_{a,b \in {}_\cM X_\cN} \sqrt{d_a d_b} \sum_{t \in \Hom(a\bar b,A)}  tt^*, 
\een
where the second expression emphasizes the sum is understood as an orthogonal sum as in the orthogonal sum of intertwiner spaces defining $\lozenge$.
Then $\cZ_h = \{e_B : B \in _\cM X_\cM\}$ and it is shown (\cite{bockenhauer1999alpha}, thm. 4.4) that 
\ben
\label{84}
e_A \star e_B = \sum_{C} \frac{d_A d_B}{d_C} N_{A,B}^C e_C
\een
where $A,B,C \in _\cM X_\cM$ and $N_{A,B}^C$ the fusion coefficients. 
Thus, $\cZ_h$ is a representation of the fusion ring of $_\cM X_\cM$ under the vertical product. 

If ${}_\cN X_\cN$ is braided as we are assuming, then $\cZ_h \subset \lozenge$ also contains representations of the fusion rules for 
${}_\cN X_\cN$, as discussed in \cite{bockenhauer1999alpha}. Recall from sec. \ref{sec:fusion} that 
$\alpha^\pm_\lambda(m) = \bar \imath^{-1} \circ \epsilon^\pm(\lambda, \theta) \circ \lambda \circ \bar \imath(m), \lambda \in {}_\cN X_\cN, m \in \cM$ 
are the alpha-induced endomorphisms of $\cM$ and $\epsilon^\pm(\lambda, \theta)$ the braiding operators of ${}_\cN X_\cN$. We define the shorthand
\ben
\langle \alpha^\pm_\lambda, B \rangle := \dim \Hom_\cM(\alpha_\lambda^\pm, B), \quad B \in {}_\cM X_\cM, 
\een
and then 
\ben
\label{pdef}
p_\lambda^\pm := d_\lambda \sum_{B \in _\cM X_\cM} d_B^{-1} \langle \alpha^\pm_\lambda, B \rangle e_B \in \lozenge.
\een
Then clearly $p_\lambda^\pm \in \cZ_h$ and it is shown (\cite{bockenhauer1999alpha} thm. 5.3 and cor. 5.4) that 
\ben\label{86}
p_\mu^\pm \star p_\nu^\pm = \sum_{\lambda} \frac{d_\mu d_\nu}{d_\lambda} N_{\mu,\nu}^\lambda p_\lambda^\pm
\een
where $\mu,\nu,\lambda \in {}_\cN X_\cN$ and $N_{\mu,\nu}^\lambda$ the fusion coefficients for ${}_\cN X_\cN$, i.e. for endomorphisms of $\cN$. 
Thus, $\cZ_h$ contains two ($\pm$) copies of the fusion ring of ${}_\cN X_\cN$.

Now assume that ${}_\cN X_\cN$ is in addition non-degenerately braided. Following \cite{bockenhauer1999alpha}, we define an element $q_{\mu,\nu} \in \lozenge, 
\mu,\nu \in {}_\cN X_\cN$ by the following expression\footnote{Our definition is seen to be equivalent to that given in \cite{bockenhauer1999alpha} if we use lem. 6.2
of that reference.} 
(see fig. \ref{fig:220413_Fig-19_TE})
 \ben
 \label{qdef}
 \begin{split}
 q_{\mu,\lambda} :=& \sqrt{d_\mu d_\lambda} D_X^{-1} \bigoplus_{a,b \in {}_\cN X_\cM} \sqrt{d_a d_b}
 \sum_{v \in \Hom(\alpha^-_\mu, \alpha_\lambda^+)}
 (1_a \times r_\mu^* \times 1_{\bar b})
 (\epsilon^-(a, \bar \mu)^* \times \epsilon^-(\mu,\bar b)^*)\\
 & \hspace{2cm} (\bar v \times 1_{a\bar b} \times v)
 (\epsilon^+(a,\bar \lambda) \times \epsilon^+(\lambda, \bar b))
 (1_a \times r_\lambda \times 1_{\bar b}) \\
 =& 
\bigoplus_{a,b \in {}_\cN X_\cM} \sqrt{d_a d_b}
 \sum_{v \in \Hom(\alpha^-_\mu, \alpha_\lambda^+)} q_{\mu,\lambda,v,a,b}
\end{split}
\een
where $v \in \Hom(\alpha^-_\mu, \alpha_\lambda^+)$ run through an 
ONB in the sense that for two $v,v'$ we have $\i(r_\mu^*)(\bar v \times v') \i(r_\lambda) = \delta_{v,v'}$.
Here $\bar v = 
(\i(r_\lambda^*) \times 1_{\bar \alpha^-_\mu})
(1_{\bar \alpha_\lambda^+} \times v^* \times 1_{\bar \alpha^-_\mu})
(1_{\bar \alpha_\lambda^+} \times \i(\bar r_\mu))
\in \Hom(\bar \alpha^-_\mu, \bar \alpha_\lambda^+)$ is the conjugate intertwiner.
Furthermore, e defined the global index as 
\ben
D_X:= \sum_{\mu \in _{\cN} X_\cN} d_\mu^2 = \sum_{B \in _{\cM} X_\cN} d_B^2.
\een

\begin{figure}[h!]
\begin{center}
  \includegraphics[width=0.8\textwidth,]{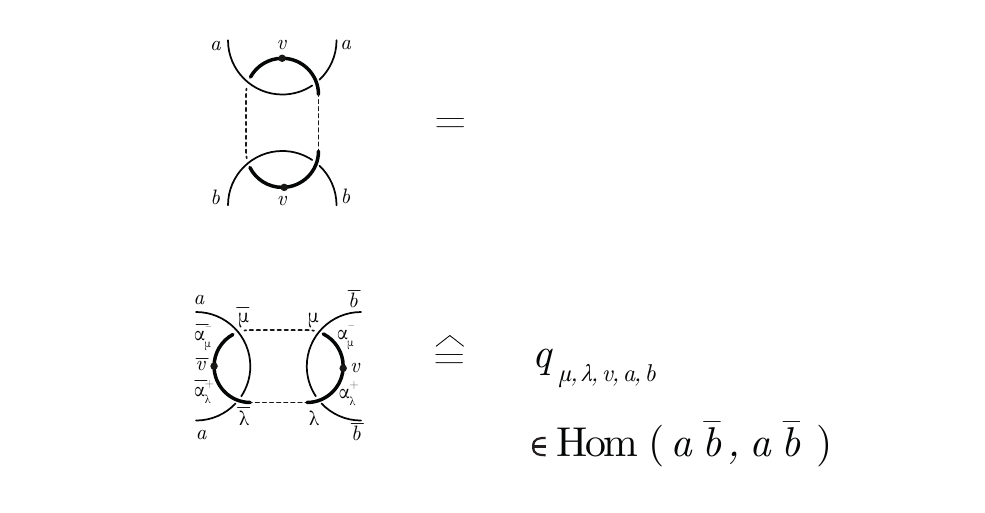}
  \end{center}
  \caption{  \label{fig:220413_Fig-19_TE} Graphical expression of $q_{\lambda,\mu,v,a,b}$ stripped of the numerical prefactors. In the upper panel, we draw the symbol for 
  this double triangle algebra element rotated by 90 degrees following the conventions of \cite{bockenhauer1999alpha}. In the lower 
  panel we draw the normal wire diagram for the intertwiners in $q_{\lambda,\mu,v,a,b}$, to be read from bottom to top.}
\end{figure}

One of the main results of \cite{bockenhauer1999alpha} (thm. 6.9) is that the element $q_{\mu,\nu} \in \lozenge, \mu,\nu \in _\cN X_\cN$ are 
mutually commuting idempotents which coincide precisely with the minimal central projections of $\cZ_h$. 
Furthermore, identifying $\cZ_h$ with an orthogonal sum of full matrix algebras as in 
\ben
\cZ_h \cong \bigoplus_{\mu,\nu \in _\cN X_\cN} M_{Z_{\mu,\nu}}(\CC), 
\een
the size of each block is precisely $Z_{\mu,\nu} = \langle \alpha^+_\mu, \alpha^-_\nu \rangle$. Thus, we have, in particular, 
\ben\label{qrel}
[q_{\mu,\nu}, p^\pm_\lambda]_\star = 0, \quad q_{\mu,\nu} \star q_{\mu',\nu'} = \delta_{\mu,\mu'} \delta_{\nu,\nu'} q_{\mu,\nu}.
\een
Since $q_{\lambda,\mu} \in \cZ_h$ which is spanned linearly by the $e_A$, there must exist complex coefficients $Y_{\lambda,\mu,A}$ such that
\ben
\label{YlmA}
q_{\lambda,\mu} = \sum_A d^{-1}_A Y_{\lambda,\mu,A} e_A.
\een
By \cite{bockenhauer1999alpha}, thm. 6.9, the fusion ring $_{\cM} X_{\cM}$ is abelian (i.e. $AB$ is unitarily equivalent to $BA$ for all $A,B \in {}_{\cM} X_{\cM}$)
if and only if $Z_{\lambda, \mu} \in \{0,1\}$, and in such a case, the simple objects $A \in {}_{\cM} X_{\cM}$ are in one-to-one correspondence 
with the pairs of simple objects $(\lambda,\mu)$ from the fusion ring $_{\cN} X_{\cN}$. 

We record some properties of $Y_{\lambda,\mu,A}$ in the following lemma for later.
\begin{lemma}\label{Ylem}
The coefficients $Y_{\lambda,\mu,A} \in \CC$ as in \eqref{YlmA} satisfy:
\begin{enumerate}
\item We have, see fig. \ref{fig:Y},
\ben
\label{eq:Ydef}
Y_{\lambda,\mu,A} = D_X^{-1} \sum_{v,a,b,t} d_a d_b \, t^* q_{\mu,\lambda,v,a,b} t,
\een
 where $v$ is running through an ONB of 
$\Hom(\alpha^+_\lambda, \alpha^-_\mu)$, and for fixed $a,b \in {}_\cM X_\cN$, $t$ is running over an ONB of 
$\Hom(a\bar b,A)$. $q_{\mu,\lambda,v,a,b}$ is as in \eqref{qdef}.
\item $\bar Y_{\lambda,\mu,A} = Y_{\lambda,\mu, \bar A}$.
\item If the fusion ring $_{\cM} X_{\cM}$ is abelian, then the matrix $Y_{\lambda,\mu,A}$ labelled by 
simple objects $A \in _{\cM} X_{\cM}$ and pairs of simple objects $(\lambda,\mu) \in {}_{\cN} X_{\cN}^2$ such that $Z_{\lambda, \mu} \neq 0$
is invertible and unitary up to normalization 
and their inverse  diagonalize the fusion rules of $_{\cM} X_{\cM}$.
\end{enumerate}
\end{lemma}

{\bf Remark:} $Y_{\lambda,\mu,A}$ should be thought of as a generalization of Rehren's matrix $Y_{\rho,\nu}$ \eqref{YRehren}, under the correspondences
$\rho \leftrightarrow (\lambda, \mu), \nu \leftrightarrow A$, as
one can appreciate by comparing figs. \ref{fig:Y} and \ref{fig:Y1}.
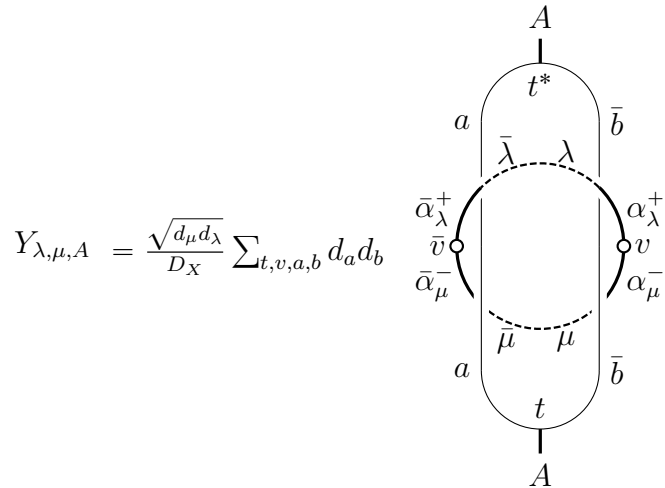
\begin{figure}[h!]
\centering
\begin{tikzpicture}[scale=.55]
\draw (-12.5,-0) node[anchor=east]{$Y_{\lambda,\mu,A}$};
\draw (-9,0) node{$=\frac{\sqrt{d_\mu d_\lambda}}{D_X}\sum_{t,v,a,b} d_a d_b$};
\draw[thin]  (-.59,-3) -- (-.59,3);
\draw[thin]  (-3.41,-3) -- (-3.41,3);   
\draw[line width=.2cm,
    color = white]  (0,0) arc (0:180:2); 
\draw[very thick]  (0,0) arc (0:360:2); 
\draw[ line width=.2cm,
    color = white]  (-.59,1.41) arc (45:135:2); 
\draw[ line width=.03cm,
    dash pattern=on .09 cm off .04 cm]  (-.59,1.41) arc (45:135:2);    
\draw[ line width=.2cm,
    color = white]  (-3.41,-1.41) arc (225:315:2); 
\draw[ line width=.03cm,
    dash pattern=on .09 cm off .04 cm]  (-3.41,-1.41) arc (225:315:2);     
 \draw[line width=.2cm,
    color = white]  (-.59,-3) -- (-.59,0);
\draw[line width=.2cm,
    color = white]  (-3.41,-3) -- (-3.41,0);    
 \draw[thin]  (-.59,-3) -- (-.59,0);
\draw[thin]  (-3.41,-3) -- (-3.41,0); 
\draw[thin]  (-3.41,-3) arc (180:360:1.41); 
\draw[thin]  (-3.41,3) arc (180:0:1.41); 
 \draw[very thick]  (-2,-5) -- (-2,-4.41);
\draw[very thick]  (-2,5) -- (-2,4.41); 
 \filldraw[color=black, fill=white, thick](0,0) circle (.15);  
  \filldraw[color=black, fill=white, thick](-4,0) circle (.15);  
\draw (-2,-4.41) node[anchor=south]{$t$};
\draw (-2,4.41) node[anchor=north]{$t^*$};
\draw (-2,-5) node[anchor=north]{$A$};
\draw (-2,5) node[anchor=south]{$A$};
\draw (-3.41,-3) node[anchor=east]{$a$};
\draw (-0.59,-3) node[anchor=west]{$\bar b$};
\draw (-3.41,3) node[anchor=east]{$a$};
\draw (-0.59,3) node[anchor=west]{$\bar b$};
\draw (-4,0) node[anchor=east]{$\bar v$};
\draw (-0,0) node[anchor=west]{$v$};
\draw (-3.8,0.9) node[anchor=east]{$\bar \alpha_\lambda^+$};
\draw (-0.2,.9) node[anchor=west]{$\alpha^+_\lambda$};
\draw (-3.8,-0.9) node[anchor=east]{$\bar \alpha_\mu^-$};
\draw (-0.2,-.9) node[anchor=west]{$\alpha^-_\mu$};

\draw (-3.3,2.3) node[anchor=west]{$\bar \lambda$};
\draw (-.9,2.3) node[anchor=east]{$\lambda$};
\draw (-3.3,-2.3) node[anchor=west]{$\bar \mu$};
\draw (-.9,-2.3) node[anchor=east]{$\mu$};
\end{tikzpicture}
  \caption{\label{fig:Y} Wire diagram for $Y_{\lambda, \mu,A}$, read from from bottom to top.}
\end{figure}
\begin{figure}[h!]
\centering
\begin{tikzpicture}[scale=.55]
\draw (-17,-0) node[anchor=east]{$Y_{\lambda,\mu}= d_\mu$};
\draw[ line width=.03cm,
    dash pattern=on .09 cm off .04 cm]  (-16,0) arc (180:360:2);    
\draw[line width=.4cm,
    color = white]  (-14,0) -- (-14,-3);   
\draw[line width=.03cm,
    dash pattern=on .09 cm off .04 cm]  (-14,-3) -- (-14,3);

\draw[ line width=.4cm,
   color=white]  (-12,0) arc (0:180:2);  
   \draw[ line width=.03cm,
    dash pattern=on .09 cm off .04 cm ]  (-12,0) arc (0:180:2);      
\draw (-16,0) node[anchor=east]{$\bar \lambda$};
\draw (-12,0) node[anchor=west]{$\lambda$};
\draw (-14,-3) node[anchor=north]{$\mu$};
\draw (-14,3) node[anchor=south]{$\mu$};
\draw (-7.5,0) node{$=\frac{1}{D_X}\sum_{t,\alpha,\beta} d_\alpha d_\beta$};
\draw[ line width=.03cm,
    dash pattern=on .09 cm off .04 cm]  (-.59,-3) -- (-.59,3);
\draw[ line width=.03cm,
    dash pattern=on .09 cm off .04 cm]  (-3.41,-3) -- (-3.41,3);   
\draw[line width=.2cm,
    color = white]  (0,0) arc (0:180:2); 
\draw[ line width=.03cm,
    dash pattern=on .09 cm off .04 cm]  (0,0) arc (0:360:2); 
\draw[ line width=.2cm,
    color = white]  (-.59,1.41) arc (45:135:2); 
\draw[ line width=.03cm,
    dash pattern=on .09 cm off .04 cm]  (-.59,1.41) arc (45:135:2);    
\draw[ line width=.2cm,
    color = white]  (-3.41,-1.41) arc (225:315:2); 
\draw[ line width=.03cm,
    dash pattern=on .09 cm off .04 cm]  (-3.41,-1.41) arc (225:315:2);     
 \draw[line width=.2cm,
    color = white]  (-.59,-3) -- (-.59,0);
\draw[line width=.2cm,
    color = white]  (-3.41,-3) -- (-3.41,0);    
 \draw[ line width=.03cm,
    dash pattern=on .09 cm off .04 cm]  (-.59,-3) -- (-.59,0);
\draw[ line width=.03cm,
    dash pattern=on .09 cm off .04 cm]  (-3.41,-3) -- (-3.41,0); 
\draw[ line width=.03cm,
    dash pattern=on .09 cm off .04 cm]  (-3.41,-3) arc (180:360:1.41); 
\draw[ line width=.03cm,
    dash pattern=on .09 cm off .04 cm]  (-3.41,3) arc (180:0:1.41); 
 \draw[ line width=.03cm,
    dash pattern=on .09 cm off .04 cm]  (-2,-5) -- (-2,-4.41);
\draw[ line width=.03cm,
    dash pattern=on .09 cm off .04 cm]  (-2,5) -- (-2,4.41); 
%
%
\draw (-2,-4.41) node[anchor=south]{$t$};
\draw (-2,4.41) node[anchor=north]{$t^*$};
\draw (-2,-5) node[anchor=north]{$\mu$};
\draw (-2,5) node[anchor=south]{$\mu$};
\draw (-3.41,-3) node[anchor=east]{$\alpha$};
\draw (-0.59,-3) node[anchor=west]{$\bar \beta$};
\draw (-3.41,3) node[anchor=east]{$\alpha$};
\draw (-0.59,3) node[anchor=west]{$\bar \beta$};

\draw (-3.3,2.3) node[anchor=west]{$\bar \lambda$};
\draw (-.9,2.3) node[anchor=east]{$\lambda$};
\draw (-3.3,-2.3) node[anchor=west]{$\bar \lambda$};
\draw (-.9,-2.3) node[anchor=east]{$\lambda$};
\end{tikzpicture}
  \caption{\label{fig:Y1} Wire diagram for Rehren's matrix $Y_{\lambda, \mu}$, read from from bottom to top. To go to the rightmost diagram 
  intended to indicate the analogy to fig. \ref{fig:Y} for $Y_{\lambda,\nu,A}$,
  we have used the BF relations.}
\end{figure}
\begin{proof}
1) It follows by combining \eqref{YlmA}, \eqref{qdef}, and \eqref{eBdef} after multiplying \eqref{YlmA} from the left with 
any non-trivial isometry $t \in \Hom(a\bar b,A)$ that $d_A^{-1} Y_{\lambda,\mu,A} =   \sum_{v} t^* q_{\mu,\lambda,v,a,b} t$. 
Now multiply by $d_a d_b/D_X$ and take a sum over an ONB of $t \in \Hom(a\bar b,A)$ and then a sum over $a,b$. Then we 
get $(\sum_{a,b} N_{a,\bar b}^A d_a d_{b}) d_A^{-1} D_X^{-1} Y_{\lambda,\mu,A}$ on the left and the claimed formula on the right. 
The result then follows from $\sum_{a,b} N_{a,\bar b}^A d_a d_b= \sum_{a,b} N_{b, A}^a d_a d_b=d_A \sum_b d_b^2 = d_A D_X$, 
using the multiplicative property of the dimension and Frobenius reciprocity. 

2) Item 1) shows that $\bar Y_{\lambda,\mu,A} =$ the sum over $v$ of the following expression
(in the following we us the shorthand $D=D_X$):
\ben
\begin{split}
&d_A t^* q_{\mu,\lambda,v,a,b}^* t\\
=&\sqrt{d_\mu d_\lambda} D^{-1} d_A \, t^* 
(1_a \times r_\lambda^* \times 1_{\bar b})
 (\epsilon^+(a, \bar \lambda)^* \times \epsilon^+(\lambda,\bar b)^*)\\
 & (\bar v^* \times 1_{a\bar b} \times v^*)
 (\epsilon^-(a,\bar \mu) \times \epsilon^-(\mu, \bar b))
 (1_a \times r_\mu \times 1_{\bar b})t\\
 =&\sqrt{d_\mu d_\lambda} D^{-1} \, \bar r_a^*(1_a \times r_b^* \times 1_{\bar a})[\\
 &(1_a \times r_\lambda^* \times 1_{\bar b})
  (\epsilon^+(a, \bar \lambda)^* \times \epsilon^+(\lambda,\bar b)^*)
  (\bar v^* \times 1_{a\bar b} \times v^*)\\
  &(\epsilon^-(a,\bar \mu) \times \epsilon^-(\mu, \bar b))
 (1_a \times r_\mu \times 1_{\bar b})t \times \bar t
 ]\bar r_A\\
 =&\sqrt{d_\mu d_\lambda} D^{-1} \, 
 \i(r_\lambda^*) 
 (1_{\bar\alpha^+_\lambda} \times \bar r_a^* \times 1_{\alpha^+_\lambda})
 (1_{\bar\alpha^+_\lambda a} \times r_b^* 1_{\bar a\alpha^+_\lambda})\\
 &
 (1_{\bar\alpha^+_\lambda a \bar b \alpha^+_\lambda} \times \epsilon^+(\lambda, \bar a)^* \times \epsilon^+(b, \bar \lambda)^*)
 (\bar v^* \times 1_{a\bar b} \times v^* \times 1_{b\bar a})\\
 &[(\epsilon^-(a,\bar \mu) \times \epsilon^-(\mu, \bar b))
 (1_a \times r_\mu \times 1_{\bar b}) \times 1_{b\bar a}](t \times \bar t) \bar r_A\\
 =&\sqrt{d_\mu d_\lambda} D^{-1} \, 
 \i(r_\mu^*) 
 (1_{\bar\alpha^-_\mu} \times \bar r_a^* \times v)
 (1_{\bar\alpha^+_\lambda a} \times r_b^* 1_{\bar a\alpha^+_\lambda})\\
 &
 (1_{\bar\alpha^+_\lambda a \bar b \alpha^+_\lambda} \times \epsilon^+(\lambda, \bar a)^* \times \epsilon^+(b, \bar \lambda)^*)
 (1_{\bar \alpha^-_\mu a\bar b} \times v^* \times 1_{b\bar a})\\
 &[(\epsilon^-(a,\bar \mu) \times \epsilon^-(\mu, \bar b))
 (1_a \times r_\mu \times 1_{\bar b}) \times 1_{b\bar a}](t \times \bar t) \bar r_A\\
 =&\sqrt{d_\mu d_\lambda} D^{-1} \, 
 \i(r_\mu^*) 
 (1_{\bar\alpha^-_\mu} \times \bar r_a^* \times 1_{\alpha^-_\mu})
 (\epsilon^-(a,\bar \mu)^* \times r_b^* 1_{\bar a} \times v)\\
 &
 (1_{a\bar \mu} \times \epsilon^-(\bar \mu,\bar b)^* \times 1_{b\bar a \alpha^+_\lambda})
 (1_{a\bar b \bar \alpha^-_\mu} \times \bar v \times 1_{b\bar a\alpha^+_\lambda})\\
 &[1_{a\bar b} \times (\epsilon^+(b,\bar \lambda) \times \epsilon^+(\lambda, \bar a))
 (1_b \times r_\lambda \times 1_{\bar a})](t \times \i(r_\lambda) \times \bar t) \bar r_A\\
 =&\sqrt{d_\mu d_\lambda} D^{-1} \, \bar r_a^*(1_a \times r_b^* \times 1_{\bar a})[\\
 &t\times (1_b \times r_\mu^* \times 1_{\bar a})
  (\epsilon^-(b, \bar \mu)^* \times \epsilon^-(\mu,\bar a)^*)
  (\bar v \times 1_{b\bar a} \times v)\\
  &(\epsilon^+(b,\bar \lambda) \times \epsilon^+(\lambda, \bar a))
 (1_b \times r_\lambda \times 1_{\bar a})\bar t 
 ]\bar r_A\\
 =&\sqrt{d_\mu d_\lambda} D^{-1} d_{\bar A} \, \bar t^*(1_b \times r_\mu^* \times 1_{\bar a})
 (\epsilon^-(b, \bar \mu)^* \times \epsilon^-(\mu,\bar a)^*)\\
 &(\bar v \times 1_{b\bar a} \times v)
 (\epsilon^+(b,\bar \lambda) \times \epsilon^+(\lambda, \bar a))
 (1_b \times r_\lambda \times 1_{\bar a}) \bar t \\
 =&d_A \bar t^* q_{\mu,\lambda,v,b,a} \bar t. 
\end{split}
\een
Here we have used repeatedly the BF relations for the relative braiding operators as 
described in \cite{bockenhauer1999alpha}, and the definition of the conjugate intertwiner
and the conjugacy relations. The proof becomes more transparent using graphical notation starting from fig. \ref{fig:Y}, and we will in the following frequently use such 
a graphical calculus, see e.g. \cite{bockenhauer1999alpha}. Now we take the sum over an ONB of $v$'s and then the right side gives exactly $Y_{\lambda,\mu,\bar A}$
by item 1) and we are done.

3) This is obvious by \cite{bockenhauer1999alpha}, thm. 6.9.
\end{proof}

\subsection{Subfactors and CFTs on $1+1$ Minkowski spacetime}

See \cite{haag2012local,longo1995nets, longo1990index, longo1989index, fredenhagen1992superselection, fredenhagen1989superselection, rehren2000canonical}.
In the Haag-Kastler approach to QFT, one can define a $1+1$ CFT by starting from a pair of left-and right-moving copies of a chiral CFT, each given by 
the Virasoro net $\{ {\mathcal V}_c(I)\}$ where $I$ runs through the open intervals of $S^1$. $\{ {\mathcal V}_c(I)\}$ in turn is generated by 
``quantized diffeomorphisms'' of $S^1$ which act non-trivially only within 
the given interval $I$. 

More precisely \cite{fewster2005quantum}, each $\{ {\mathcal V}_c(I)\}$
is a von Neumann algebra acting on a common Hilbert space, $\sH$, 
which is generated by unitary operators $U(f), f \in {\rm Diff}(S^1)$
subject to the relations $U(f) U(f') = e^{icB(f,f')} U(f\circ f')$, 
$U(id) = 1$, where $f,f'$ act non-trivially only within $I$ and where 
$B: {\rm Diff}(S^1) \times {\rm Diff}(S^1) \to \RR$ is the Bott cocycle.
It is customary to identify the circle $S^1 \setminus \{-1\}$ minus 
a point with $\RR$ via the Caley transform, and then we can consider the
algebras $\cV_c(I)$ as labelled by open intervals $I \subset \RR$.
The product of the left and right moving chiral CFTs is labelled by diamonds 
$D:=\{ (t,x) \in \RR^2 \mid t-x \in I_-, t+x \in I_+\} \cong I_+ \times I_-$ where $I_\pm$ are open intervals of $\RR$, 
and the algebra of observables of the combined left- and right chiral observables is (see fig. \ref{fig:diamond})
\ben
\label{extension1}
\cA(D) := \cV_c(I_+) \otimes \cV_c(I_-)^{\rm opp}, 
\quad D \cong I_+ \times I_-,
\een
with the appropriate notion of tensor product for von Neumann algebras. Here ``opp'' means the opposite algebra, 
which is identical as a vector space but has the opposite product structure
$n_1 \cdot^{\rm opp} n_2 := n_2 n_1$ and same $*$-structure.
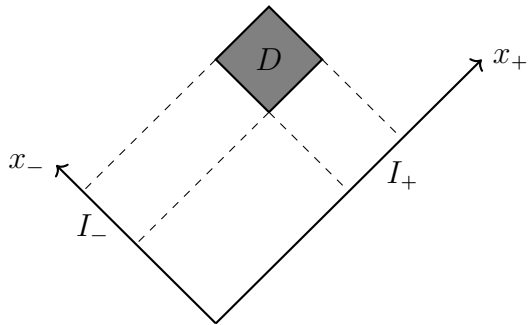
\begin{figure}[h!]
\centering
\begin{tikzpicture}[scale=.7]
\filldraw[color=black, fill=gray, thick](0,0) -- (1,1) -- (0,2) -- (-1,1) -- (0,0);
\draw[->, thick] (-1,-4) -- (-4,-1);
\draw[->,thick] (-1,-4) -- (4,1);
\draw[dashed] (-1,1) -- (-3.5,-1.5);
\draw[dashed] (0,0) -- (-2.5,-2.5);
\draw[dashed] (1,1) -- (2.5,-.5);
\draw[dashed] (0,0) -- (1.5,-1.5);
\draw (-4,-1) node[anchor=east]{$x_-$};
\draw (4,1) node[anchor=west]{$x_+$};
\draw (0,1) node{$D$};
\draw (-2.8,-2.2) node[anchor=east]{$I_-$};
\draw (2,-1.2) node[anchor=west]{$I_+$};
\end{tikzpicture}
  \caption{\label{fig:diamond} The causal diamond $D=I_+ \times I_-$. 
  }
\end{figure}
These algebras are mutually local, i.e. $[\cA(D_1), \cA(D_2)] = \{0\}$
when $D_1$ and $D_2$ are space- or timelike related diamonds.

For each open interval $I \subset \RR$, we have a system $\Delta_{\rm DHR}(I)$ of irreducible, endomorphisms  $\mu$ of ${\mathcal V}_c(I)$ which 
are in 1-1 correspondence with irreducible, unitary, highest weight representations of the Virasoro algebra for 
the central charge $c$. The endomorphisms in $\Delta_{\rm DHR}(I)$
are localized within $I$, i.e. they can be extended to the quasi-local 
$C^*$-algebra generated by all open intervals $J \subset \RR$
so as to act as the identity on each $\cV_c(J)$ for any $J \cap I = \emptyset$. It is known that $\Delta_{\rm DHR}(I)$
forms a modular tensor category. 

In the Haag-Kastler approach, a {\em rational CFT} in 1+1 Minkowski spacetime is a finite index extension of 
the net $\{\cA(D)\}$, i.e. collection of von Neumann factors 
\ben
\label{extension2}
\cC(D) \supset \cA(D), 
\een
which are again supposed to be mutually local, meaning that $[\cC(D_1), \cC(D_2)] = 0$ for spacelike $D_1, D_2$.  
It is understood how to construct such an extension from the modular tensor category of transportable endomorphisms $\mu$ of ${\mathcal V}_c(I)$. First, 
one shows that due to the transportability and conformal invariance, this extension problem is solved once it is solved for a {\em single}, arbitrary reference 
diamond $D_0 = I \times I$, that is we only need $\cC(D_0) \supset \cA(D_0)$. The intuitive reason for this is that this 
diamond can be appropriately moved to any other 
diamond by means of elements from the group
${\rm Mob} \times {\rm Mob}$, where
${\rm Mob} = \widetilde{{\rm PSL_2}(\RR)}/\{\pm 1\}$,
for which one automatically has a strongly continuous unitary
representation on $\sH \otimes \sH$ which acts geometrically on the original net $\{\cA(D)\}$ and therefore also on $\{ \cC(D) \}$.

Accordingly, we set $\cN:= \cV_c(I)$, and consider
$\cN \otimes \cN^{\rm opp} (= \cA(D_0))$ where $\cN^{\rm opp}$ is the opposite 
algebra. Then we consider the chosen endomorphisms $\mu$ of $\cN$ as defining the modular tensor category $_\cN X_\cN:= \Delta_{\rm DHR}(I)$. Accordingly, we assume that this contains a canonical endomorphism (i.e. a $Q$-system $(\theta, x, w)$) and call $\imath:\cN \to \cM$ the corresponding extension of $\cN$ with embedding $\imath$.
The desired extension 
\ben
\cC \supset \cN \otimes \cN^{\rm opp}
\een
is defined by the following $Q$-system $Z[X]:=(\theta_{\rm R},x_{\rm R},w_{\rm R})$ \cite{rehren2000canonical}.
First, the endomorphism $\theta_{\rm R}$ of
$\cN \otimes \cN^{\rm opp}$ is given by
\ben
\label{Thetadef}
\theta_{\rm R} := \bigoplus_{\lambda, \mu \, _\cN X_\cN} Z_{\lambda,\mu} \, \lambda \otimes \mu^{\rm opp} \ \ .
\een
 The direct sum is 
understood as follows. For each triple $l=(\lambda,\mu,v)$
with $v$ running through an ONB of $\Hom(\alpha^-_\mu, \alpha^+_\lambda)$, 
we pick an isometry $T_l \in \cN \otimes \cN^{\rm opp}$ such that the relations of a Cuntz algebra 
are fulfilled,
\ben
T_l^* T_{l'} = \delta_{l,l'}1, \quad \sum_l T_l T_l^* = 1.
\een
Then we set $\theta_{\rm R}(n \otimes n^{\rm opp}) = \sum_l T_l(\lambda(n) \otimes \mu^{\rm opp}(n^{\rm opp}))T_l^*$.
(Note that the numbers 
$Z_{\mu,\lambda} \in \bN_0$ are equal to the number of basis elements $v \in \Hom(\alpha^+_\mu,\alpha^-_\lambda)$.) 
The operator $x_{\rm R} \in \Hom(\theta_{\rm R},\theta_{\rm R}^2)$ is
\ben
x_{\rm R} := \sum_{l,l',l''} \sum_{e,f}
\eta_{e,f;v,v''}^{v'} (T_l \times T_{l'})[\i(e) \otimes \i^{\rm opp}(f^*)] T_{l''}^*
\een
where $e$ respectively $f$
run through ONBs of $\Hom(\lambda \lambda',\lambda'')$ respectively 
$\Hom(\mu \mu',\mu'')$, so that the
embeddings $\i: \cN \to \cM, \i^{\rm opp}: \cN^{\rm opp} \to \cM^{\rm opp}$
produce intertwiners in $\Hom(\alpha^-_\lambda \alpha_{\lambda'}^-, \alpha^-_{\lambda''})$ respectively 
$\Hom(\alpha^+_\mu \alpha^+_{\mu'},\alpha^+_{\mu''})$. 
$v'$ runs through an ONB of 
$\Hom(\alpha^-_{\lambda'}, \alpha^+_{\mu'})$ (similarly for $v''$), 
$l'=(\lambda', \mu', v')$ (similarly for $l''$) and 
Rehren's structure constants are defined as
\ben
\label{rehren}
\eta_{e,f;v,v''}^{v'} := 
\left[\frac{d(\mu) d(\mu')}{d(\theta) d(\mu'')} \right]^{1/2}
E [\i(e)^* (v \times v')\i(f) (v'')^*] , 
\een
with $E:\cM \to \cN$ the minimal conditional expectation. 
The properties of $E$ and $\alpha$-induction imply that 
$\eta_{e,f;v,v''}^{v'}$ is an intertwiner in $\Hom_\cN(\lambda,\lambda)$, hence a scalar. The isometry $w_{\rm R}$ is defined to be $T_{id}$. The formula
\ben
\epsilon_{\rm R}(\theta_{\rm R},\theta_{\rm R}) := 
\sum_{l,l'} (T_l \times T_{l'}) [\i(\epsilon^+(\mu,\mu')) \times \i^{\rm opp}(\epsilon^-(\lambda,\lambda')^*)] (T_{l'} \times T_{l})^*
\een
defines a braiding such that $\epsilon_{\rm R}(\theta_{\rm R},\theta_{\rm R}) x_{\rm R} = x_{\rm R}$, i.e. the Q-system $Z[X]:=(\theta_{\rm R},x_{\rm R},w_{\rm R})$ is commutative, and it is in a sense also the maximal such Q-system. It is called the ``full center'' for this reason and there exists a more abstract, equivalent categorical definition 
\cite{bischoff2014characterization}.
The summands in the canonical endomorphism $\theta_{\rm R}$ correspond to the ``primary fields'' by which $\cN \otimes \cN^{\rm opp}$ 
has been extended, see sec. \ref{sec:Qsys}. 

We now set $\cC(D_0) := \cC$, and by applying appropriate 
representers of the conformal group, these can be transported from 
the given diamond $D \subset \RR^2$ to any other diamond of any shape.
Locality of the net $\{ \cC(D) \}$ is then shown to follow from the commutativity of the Q-system $Z[X]$. 
Since the numbers $Z_{\mu,\lambda} \in \bN_0$ are equal to the number of basis elements 
$v \in \Hom(\alpha^+_\mu,\alpha^-_\lambda),$ the independent primary fields are labelled by $(\mu,\lambda,v)$. 

\subsection{Defects for CFTs on $1+1$ Minkowski spacetime}
\label{sec:transparent}

See \cite{bischoff2016phase,bischoff2015tensor,frohlich2006correspondences}.
A variant of the above full center Q-system construction can also 
be employed to construct a $1+1$ CFT in the presence of a transparent defect. 
In the language of conformal nets this situation is encoded in the following way. 
We have {\em two} finite index 
nets $\{ \cD_l(D) \}$ and $\{\cD_r(D)\}$ labeled by causal diamonds $D$ extending the Virasoro-net $\{ \cA(D) \}$, 
as above in \eqref{extension2}, \eqref{extension1}. However, different from above we require Einstein causality only in the restricted sense
\ben
[\cD_l(D_1), \cD_r(D_2)] = \{0\} \quad \text{if $D_1$ is left-local to $D_2$.}
\een
Here, ``left-local'' means that $D_1$ is not only in the causal complement of $D_2$ but also to the left, see fig. \ref{fig:diamond1}. 
\begin{figure}[h!]
\centering
\begin{tikzpicture}[scale=.7]
\filldraw[color=black, fill=lightgray, thin] (1,1) -- (-2,4) -- (-2,-2) -- (1,1);
\filldraw[color=black, fill=lightgray, thin](4,2) -- (6,4) -- (6,0) -- (4,2);
\filldraw[color=black, fill=gray, thick](0,0) -- (1,1) -- (0,2) -- (-1,1) -- (0,0);
\filldraw[color=black, fill=gray, thick](4,2) -- (4.5,2.5) -- (5.5,1.5) -- (5,1) -- (4,2);
\draw[->, thick, dashed] (-2,-3) -- (6,5);
\draw[->,thick, dashed] (6,-3) -- (-2,5);
\draw (-2,5) node[anchor=east]{$x_-$};
\draw (6,5) node[anchor=west]{$x_+$};
\draw (0,1) node{$D_1$};
\draw (4.8,1.8) node{$D_2$};
%
\end{tikzpicture}
  \caption{\label{fig:diamond1} $D_1$ is left-local to $D_2$. Operators in $\cD_l(D_1)$ can be thought of as having a ``left shadow'', 
  and operators in $\cD_r(D_2)$ a ``right shadow''. These shadows must not overlap for commutativity.
  }
\end{figure}
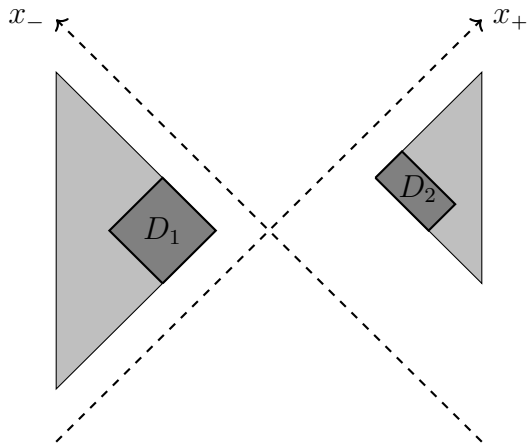
The restricted locality encodes that there is somehow a division of the system into a left and a right part and this division is caused by a defect. The defect is topological (so does not have a precise location) and furthermore invisible from the viewpoint of the underlying Virasoro net $\{\cA(D)\}$ which remains local in the unrestricted sense.

It turns out that the classification of defects, i.e. finite-index extensions of the Virasoro net subject to the above restricted locality property is achieved by the ``braided product''\footnote{Alternatively, one could take $Z[X] \times^+ Z[X]$.}. First, we fix a diamond $D_0 = I \times I$, where $I$ is some open interval, and we set $\cN = \cV_c(I)$, with $\cV_c$ the Virasoro net, so that $\cA(D_0)$ can 
be identified with $\cN \otimes \cN^{\rm opp}$, see sec. \ref{sec:Qsys}. In $\cN \otimes \cN^{\rm opp}$, we have the full center Q-system, $Z[X]:=(\theta_{\rm R},x_{\rm R},w_{\rm R})$. Then we take the, say $+$, braided product
$Z[X] \times^+ Z[X]$ of two full center Q-systems, see sec. \ref{sec:Qsys}. The braided product is a new (commutative) Q-system and thus corresponds to an extension $\cD_{12}^+$ of $\cN \otimes \cN^{\rm opp}$. Additionally, we have the intermediate extensions $\cD_1^+$ and $\cD_2^+$, see sec. \ref{sec:Qsys}. 
By construction, $\cN \otimes \cN^{\rm opp}$ is (isomorphic to) $\cA(D_0)$ for a fixed diamond $D_0$ and we set $\cD_l(D_0):=\cD_1^+$ and $\cD_r(D_0):=\cD_2^+$, 
as well as $\cD(D_0):= \cD^+_{12}$. This produces only a single pair of inclusions of von Neumann algebras, $\cA(D_0) \subset \cD_{l,r}(D_0) \subset \cD(D_0)$, 
and not a net of inclusions. But this can again be constructed since the endomorphisms of ${}_\cN X_\cN$ are transportable. An irreducible defect then corresponds to a central projection in the von Neumann algebra $\cD_{12}^+$; more precisely, that part of the common Hilbert space $\sH \otimes \sH$ on which all the algebras $\cD_{l,r}(D), \cA(D)$ are acting which is the range of the central 
projection, corresponds to quantum states with a specific ``irreducible'' topological defect.

A full characterization of the center $\cZ(\cD):=\cD_{12}^+ \cap (\cD_{12}^+)'$ has been given in \cite{bischoff2016phase,bischoff2015tensor}; in particular 
these authors have given a presentation of this center as an algebra in terms of generators and relations which we use below. 
These generators, $B_{\lambda,\mu,w_1,w_2}$,
are labelled by a pair of simple objects $\lambda, \mu \in \, _\cN X_\cN$
and a pair $w_i \in \Hom(\alpha^-_\mu, \alpha^+_\lambda)$. Thus, the dimension of $\cZ(\cD)$ is 
$\sum_{\mu,\lambda \in \, _\cN X_\cN} Z_{\mu,\lambda}^2$ since $\dim \Hom(\alpha^-_\mu, \alpha^+_\lambda) =
Z_{\mu,\lambda}$, and, because this is also equal to the number of simple objects in $\, _\cM X_\cM$ \cite{bockenhauer1999alpha}, 
it suggests that the minimal projections in 
$\cZ(\cD)$ (i.e. irreducible defects) are constructible from the simple
objects $A \in \, {}_\cM X_\cM$. That construction has been given in 
\cite{bischoff2016phase,bischoff2015tensor} (thm. 5.5 resp. thm. 4.44), building partly on work by \cite{frohlich2006correspondences}. 

The details of that construction
are broadly as follows. Let $A \in \, _\cM X_\cM$, so $\beta:=\bar \i B \i
\in \,
{}_\cN X_\cN$ can be defined which defines $m \in \Hom(\theta \beta \theta, \beta)$ as $m:= 1_{\bar \i} \times \bar r \times 1_B \times \bar r \times 1_\i$, where $\bar r \in \Hom(\i\bar \i,id)$ is a solution to the conjugacy relations and where $\theta = \bar \i \i$ is the canonical endomorphism 
associated with $\cN \subset \cM$. Then one can define
\ben
p_B := r_\beta^* \bar \beta\left(
\epsilon^+(\theta,\beta) \theta\beta(r_B^*)m
\right) r_\beta
\een
which is a projection in $\Hom(\theta,\theta)$ and the $p_B$'s realize a copy of the $_\cM X_\cM$-fusion ring. As is shown in more detail in 
\cite{bischoff2016phase,bischoff2015tensor} (thm. 5.5 resp. thm. 4.44), one can obtain from 
the $p_B$ the minimal projections in $\cZ(\cD)$ by a canonical procedure based on the observation that $\cZ(\cD) \cong \Hom(\Theta_{\rm R}, \Theta_{\rm R})$. 

These minimal projections are not actually used in the present paper but we will show that the for an anyonic spin chain based on the non-degenerately braided fusion category ${}_\cN X_\cN$ associated with an inclusion of von Neumann factors $\cN \subset \cM$, we can construct matrix product operators on a bi-partite anyonic spin chain obeying the very same algebraic relations as the generators $B_{\lambda,\mu,w_1,w_2}$ of 
$\cZ(\cD)$. The corresponding projections are then also MPOs, and they determine an orthogonal decomposition of the Hilbert space of the bipartite spin chain. Each summand is invariant under the local operators built from the energy densities, and based on these analogies, we consider the summands, labelled by the simple objects $A \in \, _\cM X_\cM$, as corresponding to a transparent defect at the level of the chain. 

\section{MPOs}

\subsection{$6j$-symbols and MPOs}

This subsection follows \cite{bultinck2017anyons} and \cite{kawahigashi2021projector, kawahigashi2020remark}, which is translated to sector language. 
We consider sets of endomorphism ${}_\cN X_\cN, {}_\cN X_\cM,  {}_\cM X_\cN, {}_\cM X_\cM$ as in sec. \ref{relative}.

Let $A, B, M_{i} \in {}_\cM X_\cM$, $a,b, a_{i} \in {}_\cM X_\cN$. 
Similar to our discussion of relative commutants, 
we first decompose $M_0 \i$ into 
irreducibles $a_1$ using an ONB of intertwiners $\{t_1\} \subset \Hom_\cM(M_0\i,a_1)$. 
Next we multiply by ${\overline{\imath}}$ from the right, and similarly consider an ONB of intertwiners $\{t_2\} \subset \Hom_\cM(a_1 {\overline{\imath}}, M_2)$, 
after which we multiply by $\i$ from the right, and consider an ONB of intertwiners $\{t_3\} \subset \Hom_\cM(M_2 {\overline{\imath}},a_3)$, and so on until $\{t_{2k}\} \subset \Hom_\cM(a_{2k-1} {\overline{\imath}},M_{2k})$, see fig. \ref{fig:18}. We denote the space of such sequences of isometric intertwiners $(t_1, t_2, \dots, t_{2k})$ 
with the property $M_0 = A, M_{2k} = B$ by ${\rm Path}^{2k}_{A,B}$. 
We also set 
\ben
\label{eq:clopa}
{\rm Path}^{2k}_{c} := \bigcup_{A \in \, _\cM X_\cM} {\rm Path}^{2k}_{A,A}
\een 
(here ``$c$'' indicates that we have ``closed'' paths in a sense). 
In a similar way, we may start by decomposing $a_0 \bar \imath$ 
into irreducible objects $M_1$ with an ONB $t_1 \subset \Hom_\cM(a_0 \bar \imath, M_1)$, then we decompose $M_1 \imath$ into irreducible objects 
$a_2$ with an ONB $t_1 \subset \Hom_\cM(M_1 \imath, a_2)$, and so on until $\{t_{2k}\} \subset \Hom_\cM(M_{2k-1} \imath,a_{2k})$. We denote the space of such sequences of isometric intertwiners $(t_1, t_2, \dots, t_{2k})$ 
with the property $a_0 = a, a_{2k} = b$ by ${\rm Path}^{2k}_{a,b}$, 
and we again define the set of closed paths as in \eqref{eq:clopa}.

Fixing $k$ and $L=2k$, we consider the following matrices labelled by closed paths $P=(t_1, t_2, \dots, t_{L}), P'=(t_1', t_2', \dots, t_{L}') \in {\rm Path}^{L}_c$ and $A \in {}_\cM X_\cM$:
{\footnotesize
\ben\label{eq:Odef}
\begin{split}
& \hspace{3cm} (O^L_A)^P_{P'} := \\
& 
\sum_{s_1,\dots,s_{2L}} 
\overline U_A
\begin{pmatrix}
& t_1 & \\
s_1 & & s_2 \\
& t_1' &
\end{pmatrix}
 U_A
\begin{pmatrix}
& t_2 & \\
s_2 & & s_3 \\
& t_2' &
\end{pmatrix}
\cdots 
\overline U_A
\begin{pmatrix}
& t_{L-1} & \\
s_{2L-1} & & s_{2L} \\
& t_{L-1}' &
\end{pmatrix}
U_A
\begin{pmatrix}
& t_{L} & \\
s_{2L} & & s_1 \\
& t_{L}' &
\end{pmatrix}.
\end{split}
\een
}
Here, the $s_i, i=1, \dots, L/2$ run through ONBs of intertwiners with suitable source and target space. 
These matrices are identified with linear operators on 
a finite-dimensional vector space $\sV^L$:

\begin{definition} (see fig. \ref{fig:18})
$\sV^L$ is the vector space whose basis elements $|t_1, \dots, t_{L}\rangle =: |P\rangle$ are labelled by $P \in {\rm Path}^{L}_{c}$. $\sV^L_{\rm open}$ is the is the vector space whose basis elements are labelled by paths $P \in {\rm Path}^{L}_{A,B}$ with $A,B \in \, _\cM X_\cM$ or by 
paths $P \in {\rm Path}^{L}_{a,b}$ with $a,b \in \, _\cM X_\cN$, depending on the context.
\end{definition}

Thus, we have for example
\ben
O^L_A |P\rangle := \sum_{P'} (O^L_A)^P_{P'} |P'\rangle,
\een
where $P,P' \in {\rm Path}^{L}_{c}$.
We think of the states $|t_1, \dots, t_{L}\rangle =: |P\rangle 
\in \sV^L_{\rm open}$ as describing the configurations of a 
spin-chain of length $L$, with the intertwiner $t_x$ corresponding the ``spin'' at ``site'' $x$. States in the subspace $\sV^L$ have ``periodic
boundary conditions'', hence describe states on closed chain, see fig. \ref{fig:18}. The MPOs
$O_A$ thus naturally act on $\sV^L$ as indicated in fig. \ref{fig:MPO2}, 
but can be extended to $\sV^L_{\rm open}$ simply by defining their matrix elements to be zero for any non-closed path. 
\begin{figure}[h!]
\centering
\begin{tikzpicture}[scale=.6]
\draw (-7.5,0) node[anchor=east]{$O_A^{L} = $};

\draw[thick]  (0.16,0) arc (0:359:3.16);

\draw[thick] (-1,1) -- (1,1);
\filldraw[color=black, fill=white, thick](0,1) circle (.6);
\draw (0,1) node[rotate=90]{$U_A$};

\draw[thick] (-1,-1) -- (1,-1);
\filldraw[color=black, fill=white, thick](0,-1) circle (.6);
\draw (0,-1) node[rotate=90]{$\bar U_A$};

\draw[thick] (-2,-4) -- (-2,-2);
\filldraw[color=black, fill=white, thick](-2,-3) circle (.6);
\draw (-2,-3) node[rotate=0]{$U_A$};

\draw[thick] (-4,-4) -- (-4,-2);
\filldraw[color=black, fill=white, thick](-4,-3) circle (.6);
\draw (-4,-3) node[rotate=0]{$\bar U_A$};

\draw[thick] (-7,-1) -- (-5,-1);
\filldraw[color=black, fill=white, thick](-6,-1) circle (.6);
\draw (-6,-1) node[rotate=270]{$U_A$};

\draw[thick] (-7,1) -- (-5,1);
\filldraw[color=black, fill=white, thick](-6,1) circle (.6);
\draw (-6,1) node[rotate=270]{$\bar U_A$};

\draw[thick] (-4,4) -- (-4,2);
\filldraw[color=black, fill=white, thick](-4,3) circle (.6);
\draw (-4,3) node[rotate=180]{$U_A$};

\draw[thick] (-2,4) -- (-2,2);
\filldraw[color=black, fill=white, thick](-2,3) circle (.6);
\draw (-2,3) node[rotate=180]{$\bar U_A$};
\end{tikzpicture}
  \caption{\label{fig:MPO2} Schematic diagram for the MPO $O_A$ representing a map from the outside legs to the inside legs. Here $L=8$.}
\end{figure}
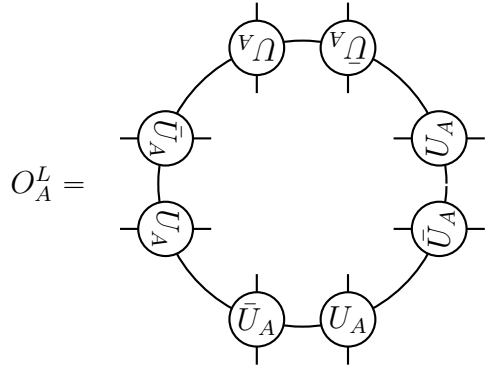

Note that if $P \in {\rm Path}^{L}_{A,B}$ for some $A,B \in \, _\cM X_\cM$, the formula \eqref{tP} gives an intertwiner 
$t_P \in \Hom_\cM(A(\i {\overline{\imath}})^k, B)$. Hence, for $P,P' \in {\rm Path}^{L}_{A,B}$, the element $t_P^* t_{P'}$ is 
a scalar which we define as the inner product $t_P^* t_{P'} = \langle P,P' \rangle 1$ on $\sV^L_{\rm open}$ (or on its subspace $\sV^L$), turning it into a Hilbert space.
The adjoint of a linear operator with respect to this inner product is denoted by a $\dagger$. In view of $\langle P | P' \rangle = \delta_{P,P'}$
using the ONB properties of the intertwiners, we may simply state this as 
\ben
(A^\dagger)^P_{P'} = \overline{A^{P'}_{P}}.
\een
The main properties of these matrix product operators (MPOs)
are:
\begin{theorem}\label{thm1}
The MPOs satisfy:
\begin{enumerate}
\item (Identity) $O^L_0 = I^L$, where ``$0$'' indicates the identity endomorphism and $I^L$ is the identity on $\sV^L$.
\item (Fusion) For $A,B \in {}_\cM X_\cM$ and any $k \ge 1, L=2k$:
\ben
\label{Ofusion}
O^L_A O^L_B = \sum_{C \in {}_\cM X_\cM } N_{A,B}^C O^L_C.
\een
\item (Conjugate) For $A \in {}_\cM X_\cM$ and any $k \ge 1, L=2k$:
\ben
(O^L_A)^\dagger = O^L_{\bar A}. 
\een
\item (Projector) Let $D_X$ be the global dimension of $X$. Then for any $k \ge 1, L=2k$, the matrix 
\ben
P^L := \sum_{A \in _\cM X_\cM} \frac{d_A}{D_X} O^L_A
\een
is an orthogonal projector, $P^L = (P^L)^\dagger, (P^L)^2 = P^L$.
\end{enumerate}
\end{theorem}

\begin{proof}
(Fusion) Let  $A,B, C,M_{i}, M_i' \in {}_\cM X_\cM, a_{i} \in {}_\cM X_\cN$ and consider ONBs of isometric intertwiners
{\footnotesize
\ben
\begin{split}\label{eq:ONB1}
t_0 \in & \Hom_\cM(C,B A)\\
t_1 \in & \Hom_\cM(M_1,a_1 {\overline{\imath}}),\\
t_2 \in & \Hom_\cM(M_2,A M_1),\\
t_3 \in & \Hom_\cM(a_2,A a_1),\\
t_4 \in & \Hom_\cM(M_2,a_2 {\overline{\imath}}),\\
t_5 \in & \Hom_\cM(a_3,B a_2),\\
t_6 \in & \Hom_\cM(M_2,a_3 {\overline{\imath}}),\\
t_7 \in & \Hom_\cM(M_3,B M_2),\\
t_8 \in & \Hom_\cM(M_3,C M_1),\\
t_1' \in & \Hom_\cM(M_3',a_1 {\overline{\imath}}),\\
t_3' \in & \Hom_\cM(M_1',A M_3'),\\
t_4' \in & \Hom_\cM(M_2',a_2 {\overline{\imath}}),\\
t_5' \in & \Hom_\cM(M_2',B M_1'),\\
t_6' \in & \Hom_\cM(M_2',a_3 {\overline{\imath}}),\\
t_7' \in & \Hom_\cM(a_3,C a_1),\\
t_8' \in & \Hom_\cM(a_3,C a_1),\\
t_8'' \in & \Hom_\cM(M_2',C M_3').
\end{split}
\een
}
To prove the fusion relation, we need the following ``zipper lemma'', which is one of the main graphical observations of \cite{bultinck2017anyons}
and algebraically a reinterpretation of the pentagon identity for $6j$-symbols which by definition holds in any unitary fusion category. 
The ``zipper tensors'' are defined  as 
\ben
Y_{A,B}^C
\begin{pmatrix}
& t_2 & \\
t_0; & & t_8 \\
& t_7 &
\end{pmatrix}
:=
t_7^*(1_B \times t_2^*)(t_0 \times 1_{M_1})t_8
, \quad
\overline Y_{A,B}^C
\begin{pmatrix}
& t_3 & \\
t_0; & & t_8' \\
& t_5 &
\end{pmatrix}
:=
t_5^*(1_B \times t_3^*)(t_0 \times 1_{a_1})t_8'
\een
and are course also $6j$-symbols (note that the arguments of the 
unbarred and barred zipper tensor are from different intertwiner spaces).
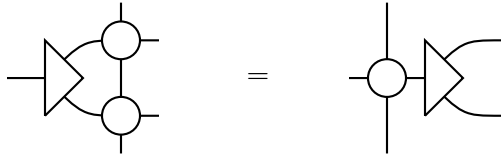
\begin{figure}[h!]
\centering
\begin{tikzpicture}[scale=.5]
\draw[thick] (-2,0) -- (0,0);
\draw[thick] (1,-2) -- (1,2);
\draw[thick] (2,1) -- (1,1);
\draw[thick] (2,-1) -- (1,-1);
\draw[thick] (-1,0) .. controls (0,1)  .. (1,1);
\draw[thick] (-1,0) .. controls (0,-1)  .. (1,-1);
\filldraw[color=black, fill=white, thick](-1,-1) -- (-1,1) -- (0,0) -- (-1,-1) ;
\filldraw[color=black, fill=white, thick](1,1) circle (.5);
\filldraw[color=black, fill=white, thick](1,-1) circle (.5);
\draw[thick] (7,0) -- (10,0);
\draw[thick] (8,2) -- (8,-2);
\draw[thick] (9,0) .. controls (10,1)  .. (11,1);
\draw[thick] (9,0) .. controls (10,-1)  .. (11,-1);
\filldraw[color=black, fill=white, thick](9,-1) -- (9,1) -- (10,0) -- (9,-1) ;
\filldraw[color=black, fill=white, thick](8,0) circle (.5);
\draw (4,0) node[anchor=west]{$=$};
\end{tikzpicture}
  \caption{\label{fig:zipper} Schematic diagram for the zipper lemma. Circles represent $6j$-symbols $U$ or $\overline U$, whereas the triangle 
  represents the tensors $\overline Y$ or $Y$.}
\end{figure}
\begin{lemma} (Zipper lemma, see fig. \ref{fig:zipper}.)
With the intertwiners from the spaces above, we have
\ben
\begin{split}
& \sum_{t_2,t_4,t_7}
U_A
\begin{pmatrix}
& t_{1} & \\
t_3 & & t_2 \\
& t_4 &
\end{pmatrix}
U_B
\begin{pmatrix}
& t_{4} & \\
t_5 & & t_7 \\
& t_6 &
\end{pmatrix}
Y_{A,B}^C
\begin{pmatrix}
& t_2 & \\
t_0; & & t_8 \\
& t_7 &
\end{pmatrix} \\
&\hspace{1cm} = \ \ 
\sum_{t_8'}
\overline Y_{A,B}^C
\begin{pmatrix}
& t_3 & \\
t_0; & & t_8' \\
& t_5 &
\end{pmatrix}
U_C
\begin{pmatrix}
& t_{1} & \\
t_8' & & t_8 \\
& t_6 &
\end{pmatrix}
\end{split}
\een
and
\ben
\begin{split}
& \sum_{t_4',t_3,t_5}
\overline U_A
\begin{pmatrix}
& t_{1}' & \\
t_3' & & t_3 \\
& t_4' &
\end{pmatrix}
\overline U_B
\begin{pmatrix}
& t_{4}' & \\
t_5' & & t_5 \\
& t_6' &
\end{pmatrix}
\overline Y_{A,B}^C
\begin{pmatrix}
& t_3 & \\
t_0; & & t_8' \\
& t_5 &
\end{pmatrix} \\
& \hspace{1cm} = \ \ 
\sum_{t_8''}
Y_{A,B}^C
\begin{pmatrix}
& t_3' & \\
t_0; & & t_8'' \\
& t_5' &
\end{pmatrix}
\overline U_C
\begin{pmatrix}
& t_{1}' & \\
t_8'' & & t_8' \\
& t_6' &
\end{pmatrix}
\end{split}
\een
where the sums are over an ONB of the respective space.
\end{lemma}
\begin{proof}
The zipper lemma is a version of the pentagon relation. However, we go through the details to show how this works explicitly once, 
because we will repeat similar manipulations without showing the complete calculations in several places later on.
We prove the first relation and begin with, noting that the term in $[\dots]$ already is a scalar, 
\ben
\begin{split}
\sum_{t_1,t_2,t_4,t_7}
& t_6^*(t_5^* \times 1_{\overline{\imath}})(1_B \times t_4)t_7t^*_7 B[t_4^* (t_3^* \times 1_{\overline{\imath}})(1_A \times t_1)t_2] B(t_2^* (1_A \times t_1^*)) \\
=& \ 
t_6^*(t_5^* \times 1_{\overline{\imath}})(1_B \times t_3^* \times 1_{\overline{\imath}})
\end{split}
\een
which follows from the ONB properties of the intertwiners and which can be rewritten as 
\ben
\sum_{t_1, t_2,t_4,t_7}
U_A
\begin{pmatrix}
& t_{1} & \\
t_3 & & t_2 \\
& t_4 &
\end{pmatrix}
U_B
\begin{pmatrix}
& t_{4} & \\
t_5 & & t_7 \\
& t_6 &
\end{pmatrix}
t_7^*(1_B \times t_2^*)(1_B \times 1_A \times t_1^*)
=
t_6^*(t_5^* \times 1_{\overline{\imath}})(1_B \times t_3^* \times 1_{\overline{\imath}}).
\een
Now we multiply with $t_0$ from the right on both sides. On the left, we 
use:
\ben
t_7^*(1_B \times t_2^*)(1_B \times 1_A \times t_1^*)(t_0 \times 1_1 \times 1_{\overline{\imath}}) =
\sum_{t_8} Y_{A,B}^C
\begin{pmatrix}
& t_2 & \\
t_0; & & t_8 \\
& t_7 &
\end{pmatrix}
t_8^* (1_C \times t_1^*).
\een
On the right side we use:
\ben
t_6^*(t_5^* \times 1_{\overline{\imath}})(1_B \times t_3^* \times 1_{\overline{\imath}})(t_0 \times 1_1 \times 1_{\overline{\imath}}) =
\sum_{t_8'} Y_{A,B}^C
\begin{pmatrix}
& t_5 & \\
t_0; & & t_8' \\
& t_3 &
\end{pmatrix}
t_6^* (t_8'{}^* \times 1_{\overline{\imath}}).
\een
On the term on the right we next use:
\ben
t_6^* (t_8'{}^* \times 1_{\overline{\imath}}) = \sum_{t_1,t_8}
U_C
\begin{pmatrix}
& t_{1} & \\
t_8' & & t_8 \\
& t_6 &
\end{pmatrix}
t_8^* (1_C \times t_1^*).
\een
Putting these identities together yields the result. The second identity is proven in a similar way.
\end{proof}

To finish the proof of (fusion), we consider the following orthogonality relations
\ben
\sum_C \sum_{t_8,t_0}
Y_{A,B}^C
\begin{pmatrix}
& t_2 & \\
t_0; & & t_8 \\
& t_7 &
\end{pmatrix}
\overline{
Y_{A,B}^C
\begin{pmatrix}
& t_2' & \\
t_0; & & t_8 \\
& t_7' &
\end{pmatrix}
} = \delta_{t_2,t_2'} \delta_{t_7,t_7'},
\een
and
\ben
\sum_{t_2,t_7}
Y_{A,B}^C
\begin{pmatrix}
& t_2 & \\
t_0; & & t_8 \\
& t_7 &
\end{pmatrix}
\overline{
Y_{A,B}^C
\begin{pmatrix}
& t_2 & \\
t_0; & & t_8' \\
& t_7 &
\end{pmatrix}
} = \delta_{t_8,t_8'},
\een
which follow from the ONB properties of the intertwiners. Similar relations hold for the barred zipper tensor. 
Then we look at a matrix element of the product $O^L_A O^L_B$, 
and insert the first orthogonality relation for the zipper tensor. The matrix element is 
labelled by closed paths $P=(t_1, t_2, \dots, t_{L}), P''=(t_1'', t_2', \dots, t_{L}'') \in {\rm Path}^{L}_c$
{\footnotesize
\ben
\begin{split}
(O^L_A O_B^L)^P_{P''} = & \  
\sum_{s_i, s_j', t_k'}  \sum_C \  \sum_{u,t_0} \ 
Y_{A,B}^C
\begin{pmatrix}
& s_1 & \\
t_0; & & u \\
& s_1' &
\end{pmatrix}
\overline{
Y_{A,B}^C
\begin{pmatrix}
& s_{L} & \\
t_0; & & u \\
& s'_{L} &
\end{pmatrix}
}
\\
& \cdot \ 
\prod_{i=0}^{L/2}
\overline U_A
\begin{pmatrix}
& t_{2i-1} & \\
s_{2i-1}  & & s_{2i} \\
& t_{2i-1}' &
\end{pmatrix}
\overline U_B
\begin{pmatrix}
& t_{2i-1}' & \\
s_{2i-1}' & & s_{2i}' \\
& t_{2i-1}'' &
\end{pmatrix} \\
& \cdot \
\prod_{j=0}^{k}
U_A
\begin{pmatrix}
& t_{2j} \ \ & \\
s_{2j} \ \ & & s_{2j+1} \\
& t_{2j}' \ \  &
\end{pmatrix}
U_B
\begin{pmatrix}
& t_{2j}' \ \ & \\
s_{2j}' \ \ & & \ s_{2j+1}' \\
& t_{2j}'' \ \ &
\end{pmatrix}
\end{split}
\een
}
On this expression, we apply the zipper lemma $L$ times, and
then we use the second orthogonality relation of the zipper tensor, giving an expression which no longer 
depends on $t_0$. 
The sum over $t_0 \in \Hom_\cM(AB,C)$ then yields the prefactor 
$N_{A,B}^C = \dim_\CC \Hom_\cM(AB,C)$.

(Conjugate): This follows from the property (conjugate) of thm. \ref{thm1}. Note that we must use the precise form of the quartic root prefactors.

(Projector): $P^L = (P^L)^\dagger$ follows from (conjugate). We have 
\ben
(P^L)^2 = \sum_{A,B,C \in _\cM X_\cM} N_{A,B}^C \frac{d_A d_B}{D_X^2} O_C^L
\een
using (fusion). Frobenius reciprocity gives $N_{A,B}^C = N_{\bar A,C}^B$ and hence
$\sum_B N_{A,B}^C d_B = d_{\bar A} d_C = d_A d_C$. So
$\sum_{A,B} N_{A,B}^C d_B  d_A= D_X d_C$ and $(P^L)^2 = P^L$ follows.
\end{proof}

\subsection{MPOs and double triangle algebra}

In this subsection, we define an extended class of MPOs which giving a representation of the double triangle algebra and of the 
fusion rings $_\cM X_\cM$ {\it and} $_\cN X_\cN$, and not just of $_\cM X_\cM$ as in the previous subsection. 
As in that subsection, we consider the finite sets $_\cN X_\cM, _\cM X_\cN, _\cM X_\cM, _\cN X_\cN$ 
of endomorphisms 
with the properties listed in sec. \ref{relative}, and paths 
$P = (t_1, \dots, t_L) \in {\rm Path}_{a',a}$ and 
$P' = (t_1', \dots, t_L') \in {\rm Path}_{b',b}.$

Fixing an even number $L$, we consider the linear operators on $\sV^L_{\rm open}$
defined by the following matrix elements, see fig. \ref{fig:6}:
\begin{figure}[h!]
\begin{center}
\hspace*{-1cm}
  \includegraphics[width=1.2\textwidth,]{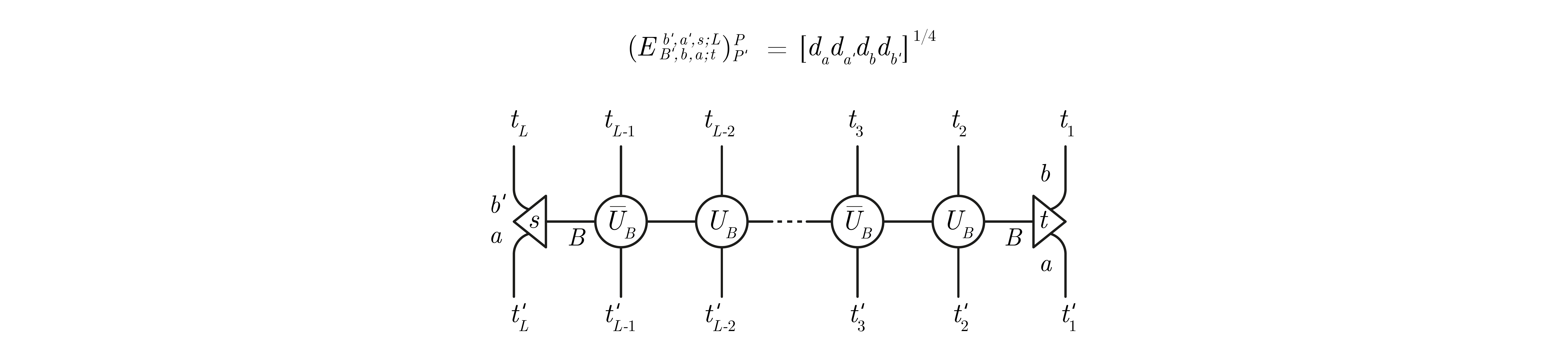}
\end{center}
  \caption{\label{fig:6} The matrix element $\langle t_L, \dots, t_1 | E^{b',a',s;L}_{B;b,a,t} |t_L', \dots, t_1' \rangle$ of the symmetry operator (MPO). Here `$P$'
  indicates the ``path'' of intertwiners $t_L, \dots, t_1$.}
\end{figure}
{\footnotesize
\ben
\begin{split}
& (E^{b',a',s;L}_{B;b,a,t})^P_{P'} := [d(a) d(b) d(a') d(b') ]^{1/4} 
\sum_{s_1, \dots, s_{L+1}}  \cdot \\
& \cdot \ 
Z_{a',\bar b'}^B
\begin{pmatrix}
& t_1& \\
s; & & s_1 \\
& t_1' &
\end{pmatrix}
\prod_{j=1}^k \left\{
\overline U_B
\begin{pmatrix}
& t_{2j-1} & \\
s_{2j-1} & & s_{2j} \\
& t_{2j-1}' &
\end{pmatrix}
 U_B
\begin{pmatrix}
& t_{2j} & \\
s_{2j} & & s_{2j+1} \\
& t_{2j}' &
\end{pmatrix}
\right\}
Z_{a,\bar b}^B
\begin{pmatrix}
&  & t_L\\
t; & s_{L+1} & \\
&  & t_L'
\end{pmatrix}
.\non
\end{split}
\een
}
Here the labels $E^{b',a',s;L}_{B;b,a,t}$ correspond to those on $e^{b',a',s}_{B;b,a,t}$ in the 
definition of the double triangle algebra, and the matrix 
elements are set to zero for all other paths whose 
beginning and endpoints do not match with $a,a',b,b'$ as above. 
$Z_{a,\bar b}^B$ is another type of $6j$-symbol defined 
as follows (see fig. \ref{fig:5}).

\begin{figure}[h!]
\begin{center}
\hspace*{-2cm}
  \includegraphics[width=1.3\textwidth,]{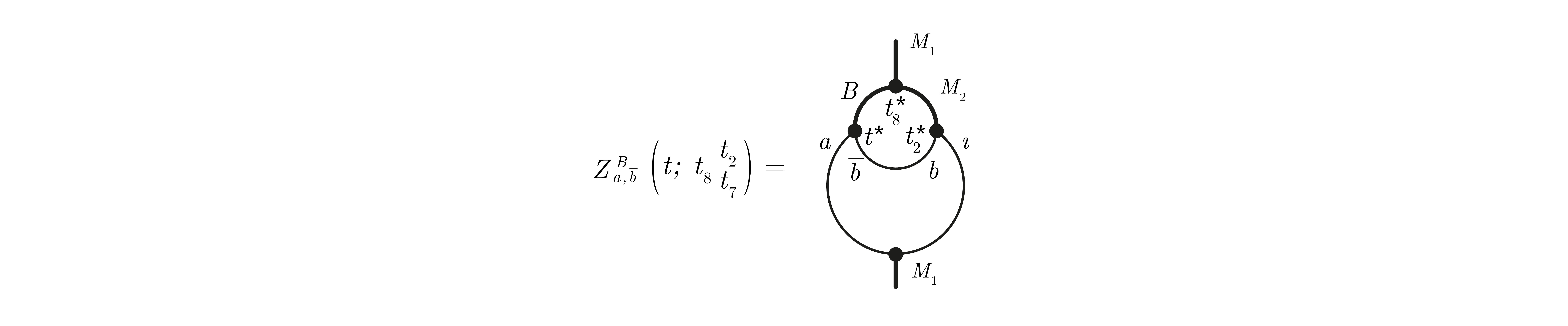}
 \end{center}
  \caption{\label{fig:5} The $6j$-symbol $Z_{a,\bar b}^B$.}
\end{figure}

$A_1, A_2, B \in \ _\cM X_\cM,  a,b,c \in \ _\cM X_\cN$ and 
$u_1 \in \Hom(a\overline{\imath},A_1), u_2 \in \Hom(B A_2,A_1), 
u_3 \in \Hom(b\overline{\imath},A_2), t \in \Hom(a\bar b,B)$. Then we define 
\ben
Z_{a,\bar b}^B
\begin{pmatrix}
&  & u_3\\
t; & u_2 & \\
&  & u_1
\end{pmatrix}
=
u_2^*
(t^* \times u_3^*)
(1_a \times r_b \times 1_{\bar \imath})
u_1
\een
which is up to a prefactor equal to the complex conjugate of the $6j$-symbol $U_B\begin{pmatrix}
   & u_3 & \\
t &  & u_2\\
& u_1 & 
\end{pmatrix}
$
as can be seen from Frobenius duality. The corresponding $Z$-type $6j$ symbol with the rearranged 
pattern of the indices is its complex conjugate (with the corresponding 
substitution of the intertwiners). 
The sum runs over a complete ONB of intertwiners in Hom-spaces 
compatible with the source and target spaces of $t_i, t_i'$ 
implicit in the definitions of the path spaces.
Our main observation in this section is the following:

\begin{theorem}
\label{thm:2}
For any even $L=2k$, the map $\pi^L: \lozenge \to \End(\sV^L_{\rm open}), e^{b',a',s}_{B;b,a,t} \mapsto E^{b',a',s;L}_{B;b,a,t}$ is a representation 
of $(\lozenge, \star)$. 
\end{theorem}

\begin{proof}
We let $B,B', B'',M_1,M_2,M_1',M_2' \in {}_\cM X_\cM$, $a,b,c,d \in {}_\cM X_\cN$, 
and consider the following intertwiners: 
{\footnotesize
\ben
\begin{split}
&t \in \Hom_\cM(a \bar b,B),\\
&t' \in \Hom_\cM(c \bar a, B'),\\
&t'' \in \Hom_\cM(B'B, B''),\\
&t''' \in \Hom_\cM(c\bar b, B''),\\
&s_1 \in \Hom_\cM(a \overline{\imath},M_1),\\
&s_2 \in \Hom_\cM(b \overline{\imath}, M_2),\\
&s_3 \in \Hom_\cM(B M_2, M_1),\\
&s_2' \in \Hom_\cM(c \overline{\imath}, M_1'), \\
&s_3' \in \Hom_\cM(B' M_1, M_1'), \\
&s_3'' \in \Hom_\cM(B'' M_2', M_1'), \\
\end{split}
\een
}
Taking a sum over an ONB of intertwiners $s_3$ gives
\ben
(t^* \times s_2^*)(1_a \times r_b \times 1_{\bar d})s_1 = \sum_{s_3} 
Z_{a,\bar b}^B
\begin{pmatrix}
&  & s_2\\
t; & s_3 & \\
&  & s_1
\end{pmatrix} \ s_3.
\een
Multiplication of two copies of this identity, using the ONB property, and standard topological moves give
\ben
(t'^* \times t^* \times s_2^*)(1_c \times r_a \times r_b \times 1_{\bar d})s_2' =
\sum_{s_1, s_3, s_3'}
Z_{a,\bar b}^B
\begin{pmatrix}
&  & s_2\\
t; & s_3 & \\
&  & s_1
\end{pmatrix}
Z_{c,\bar a}^{B'}
\begin{pmatrix}
&  & s_1\\
t'; & s_3' & \\
&  & s_2'
\end{pmatrix}
(1_{B'} \times s_3)s_3'.
\een
Employing the definition of the zipper tensor entails
\ben\label{prev}
\begin{split}
& \hspace{1cm}
\sum_{t''',s_3''}
Z_{c,\bar b}^{B''}
\begin{pmatrix}
&  & s_2\\
t'''; & s_3'' & \\
&  & s_2'
\end{pmatrix}
[t''^*(t'^* \times t^*)(1_c \times r_a \times 1_{\bar b})t'''] s_3''
\\
=& \ 
\sum_{s_1, s_3, s_3',s_3''}
Z_{a,\bar b}^B
\begin{pmatrix}
&  & s_2\\
t; & s_3 & \\
&  & s_1
\end{pmatrix}
Z_{c,\bar a}^{B'}
\begin{pmatrix}
&  & s_1\\
t'; & s_3' & \\
&  & s_2'
\end{pmatrix}
Y_{B,B'}^{B''}
\begin{pmatrix}
&  & s_3\\
t''; & s_3'' & \\
&  & s_3'
\end{pmatrix} s_3'',
\end{split}
\een
so we can cancel the sum over the ONB $s_3''$. The resulting identity is graphically represented in fig. \ref{fig:7}.

\begin{figure}[h!]
\begin{center}
\hspace*{-.8cm}
  \includegraphics[width=1.2\textwidth,]{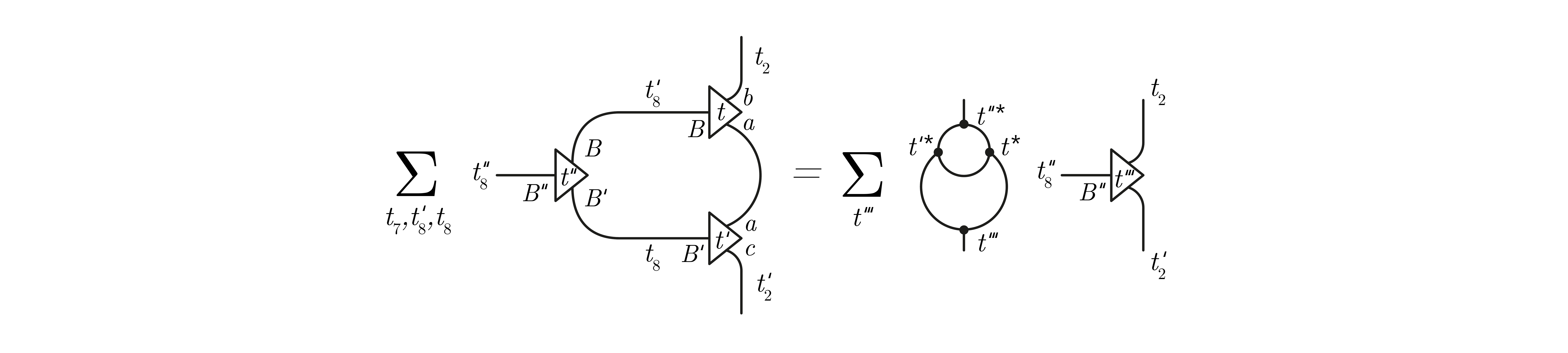}
\end{center}
  \caption{\label{fig:7} The identity following from eq. \eqref{prev}.}
\end{figure}
 
 Now we represent the product $E^{b',a',s;L}_{B;b,a,t} E^{a',c',s';L}_{B';a,c,t'}$ as the left side of fig. \ref{fig:8}. 

\begin{figure}
\begin{center}
\hspace*{-1cm}
  \includegraphics[width=1.1\textwidth,]{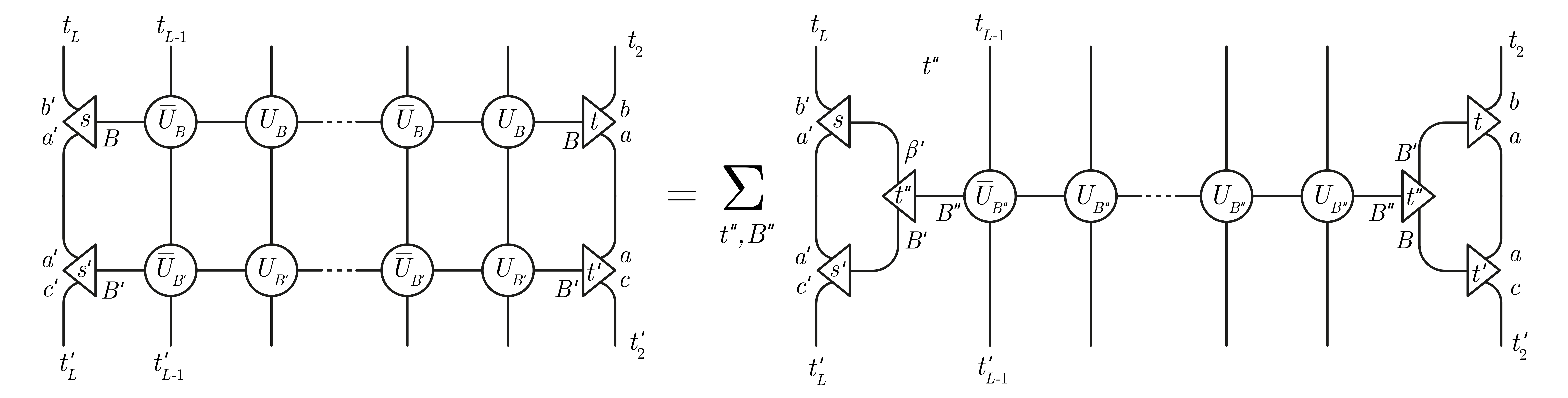}
 \end{center}
  \caption{\label{fig:8} The product $E^{b',a',s;L}_{B;b,a,t} E^{a',c',s';L}_{B';a,c,t'}$.}
\end{figure}

Applying the zipper lemma $L$ times we get the picture on the right side of fig. \ref{fig:8}. Then we apply \eqref{prev} (i.e. fig. \ref{fig:7}) 
at both ends of that picture, and apply the ONB property to the sum over $u$, we get 
\ben\label{eq:EE}
\begin{split}
&E^{b',a',s;L}_{B;b,a,t} E^{a',c',s';L}_{B';a,c,t'} \\
= & \ 
\sum_{s'',t''} (d_a d_{a'})^{1/2} 
[s''^*(1_{c'} \times r_{a'}^* \times 1_{\bar b'})(s' \times s)(t'^* \times t^*)(1_c \times r_a \times 1_{\bar b}) t''] E^{c', b', s''; L}_{B''; c,b,t''},
\end{split}
\een
as described graphically in fig. \ref{fig:9}. 

\begin{figure}[h!]
\begin{center}
\hspace*{-1.5cm}
  \includegraphics[width=1.2\textwidth,]{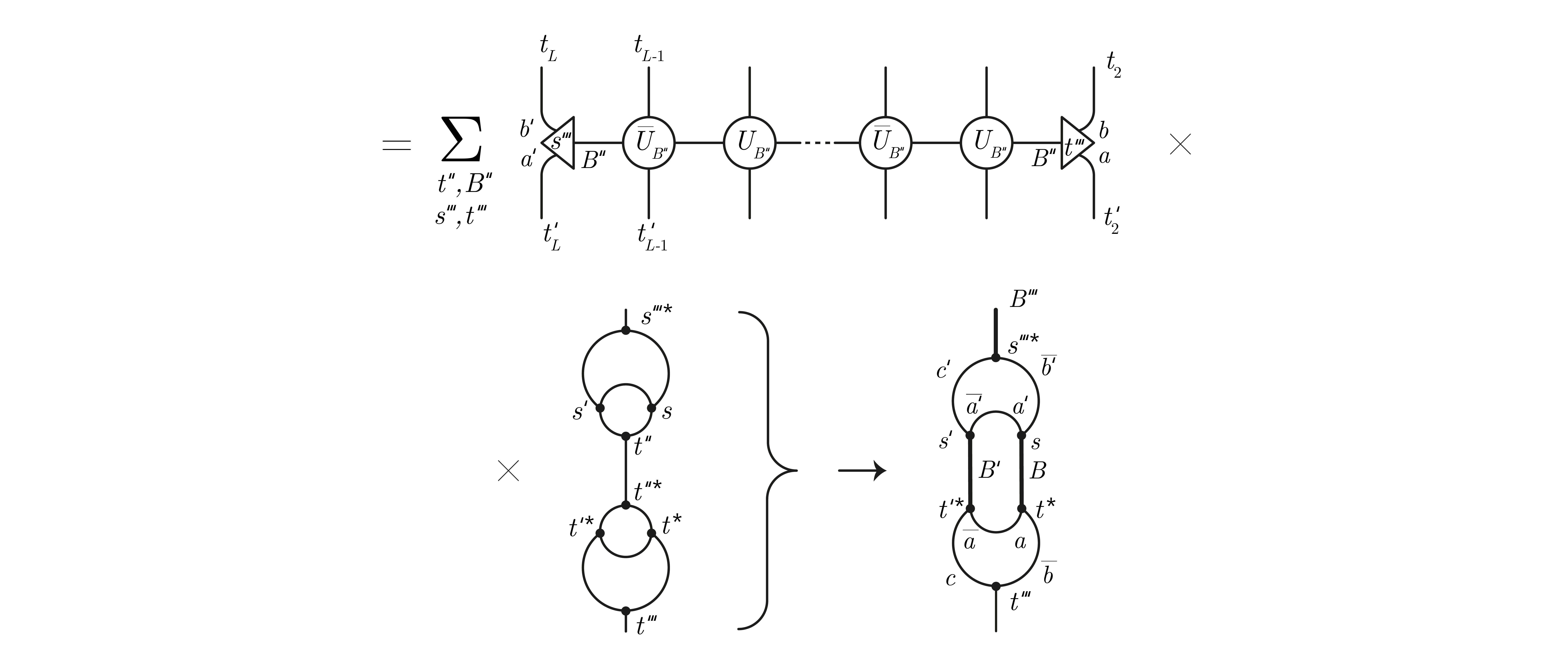}
  \end{center}
  \caption{\label{fig:9} The right side of fig. \ref{fig:8}.}
\end{figure}

This is precisely the multiplication law in the double triangle algebra ``when the indices match'', because the 
scalar $C_{s,t;s',t'}^{s'',t''} := (d_a d_{a'})^{1/2} [ \dots ]$ on the 
right side in \eqref{eq:EE} is 
the structure constant arising when projecting 
$e^{b',a',s}_{B;b,a,t} \star e^{a',c',s'}_{B';a,c,t'}$
onto the $e^{c', b', s''}_{B''; c,b,t''}$ component, see bottom panel in fig. \ref{fig:8} and fig. \ref{fig:9}. 
\begin{figure}[h!]
\begin{center}
\hspace*{-2.5cm}
  \includegraphics[width=1.3\textwidth,]{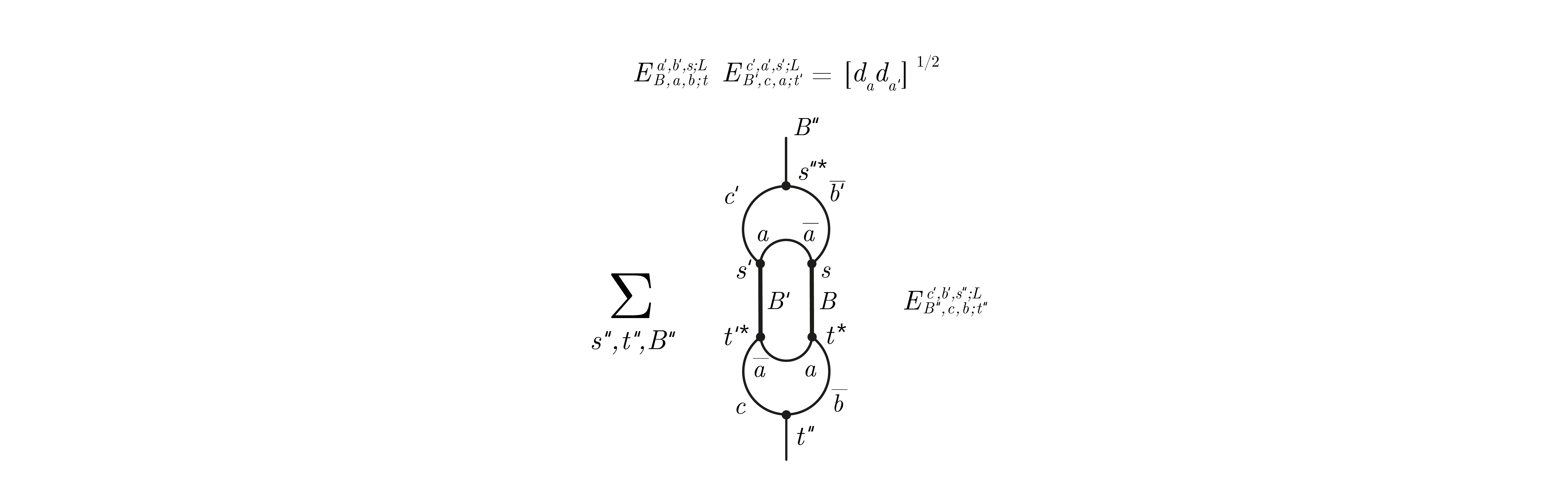}
  \end{center}
  \caption{\label{fig:10} The formula for the product $E^{b',a',s;L}_{B;b,a,t} E^{a',c',s';L}_{B;a,c,t'}$
  following from fig. \ref{fig:9}.}
  \end{figure}
When the indices do not match, we get zero 
simply because the intertwiner spaces do not match either. 
\end{proof}

Recall the definition of  $p^\mp_\mu, q_{\mu,\nu} \in \cZ_h$ [see \eqref{pdef}, \eqref{qdef}], where $\mu,\nu \in {}_\cN X_\cN$. 
The above representation $\pi^L$ gives linear operators
\ben
Q^L_{\mu,\nu} := \pi^L(q_{\mu,\nu})
\een
for any length $L$ of the spin chain. These operators leave 
the subspace $\sV^L \subset \sV^L_{\rm open}$ of periodic spin chain configurations invariant because for elements in 
$\cZ_h$, the source and target objects in the corresponding paths $P = (t_1, \dots, t_L) \in {\rm Path}_{a',a}$ and 
$P' = (t_1', \dots, t_L') \in {\rm Path}_{b',b}$ are the same, $a=a',b=b'$.
We have the following

\begin{corollary}
\label{cor:triangle}
The following holds for all $L = 2k \ge 0$:
\begin{enumerate}
\item For $B \in \ _\cM X_\cM$ we have 
\ben
\pi^L(e_B) = d_B \cdot O_B^L,
\een
where $O_B^L$ coincides with the MPOs defined earlier in \eqref{eq:Odef}.

\item For $\nu \in \ _\cN X_\cN$ we have
\ben
\pi^L(p^\pm_\nu) = d_\nu \cdot O_{\alpha_\nu^\pm}^L, 
\een
where $O_A^L$ coincides with the MPOs defined earlier in \eqref{eq:Odef}
and $A=\alpha^\pm_\nu$ are the alpha-induced endomorphisms of $\cM$.

\item There holds
\ben
O_{\alpha^\pm_\lambda} = \sum_{B \in \ _\cM X_\cM} 
\langle 
\alpha_\lambda^\pm, B
\rangle \, O_B
\een 
and we have the fusion
\ben
O_{\alpha_\mu^\pm}^L O_{\alpha_\nu^\pm}^L = \sum_{\lambda \in 
\ _\cN X_\cN} N_{\mu,\nu}^\lambda O_{\alpha_\lambda^\pm}^L,
\een
for all $\lambda,\mu,\nu \in \ _\cN X_\cN$. 

\item 
We have 
\ben
\label{eq:qmn}
Q_{\lambda,\mu}^L = \sum_A Y_{\lambda,\mu,A} \, O_A^L,
\een
with $Y_{\lambda,\mu,A}$ as in fig. \ref{fig:Y} and \eqref{eq:Ydef}.
$Q^L_{\mu,\nu} = 0$ iff $Z_{\mu,\nu} = \langle \alpha^+_\mu, \alpha^-_\nu\rangle = 0$, $Q^L_{\mu,\nu} = (Q^L_{\mu,\nu})^\dagger$, $Q_{\mu,\nu}^L Q^L_{\mu',\nu'} = \delta_{\mu,\mu'} \delta_{\nu,\nu'}
Q^L_{\mu,\nu}$, 
$[O_{\alpha^\pm_\lambda}^{L}, Q^L_{\mu,\nu} ] = 
[O_{B}^{L}, Q^L_{\mu,\nu} ] = 0$ for all $\mu, \nu, \lambda \in \ _\cN X_\cN, B \in \ _\cM X_\cM$.

\item $\sum_{\mu,\nu \in _\cN X_\cN} Q_{\mu,\nu}^L = I^L$, where $I^L$ is the identity operator (matrix elements $(I^L)^P_{P'} = \delta_{P,P'}$) 
on the spin chain.

\item $Q_{0,0}^L = P^L$, with $P^L$ the orthogonal projector given by thm. \ref{thm1}.

\item We have
\ben
\sum_{\mu, \lambda \in \, _\cN X_\cN} \frac{d_\lambda d_\mu}{D_X} O_{\alpha_{\lambda}^+}^L O_{\alpha_{\mu}^-}^L = P^L
\een
with $P^L$ the orthogonal projector given by thm. \ref{thm1}
and $D_X = \sum_{\mu \in \, _\cN X_\cN} d_\mu^2$.
\end{enumerate}
\end{corollary}

\begin{proof} 
1) The is similar to 2) but simpler and therefore omitted. Note that 
the claim is consistent with the fusion algebra \eqref{84} for $e_B$
in the double triangle algebra and theorem \ref{thm1}.

2) For $a,b \in {}_\cM X_\cN, 
M_i \in {}_\cM X_\cM, \lambda,\mu \in {}_\cN X_\cN,$ let $s \in \Hom_\cM(b\lambda, a), 
t_1 \in \Hom_\cM(\alpha^\pm_\lambda M_2,M_1), t_2 \in \Hom_\cM(b \overline{\imath}, M_2), t_3 \in \Hom_\cM(a \overline{\imath},M_1)$ be members of 
ONBs of intertwiners. Now
we define
\ben
^\pm Z_{a,\bar b}^\lambda
\begin{pmatrix}
&  & t_2\\
s; & t_1 & \\
&  & t_3
\end{pmatrix} := 
t_1^*
(\bar r_b^* \times 1_{\alpha_\lambda^\pm} \times 1_{M_2})(1_b \times \epsilon^\pm(\lambda, \bar b) \times t_2^*)
(s \times r_b \times 1_{\bar \imath})
t_3,
\een
see fig. \ref{fig:+Z}. We also define the $Z$ with the rearranged indices used below by complex conjugation of this expression. A simple BF 
move and the definition of the $6j$-symbols shows that 
\ben
^\pm Z_{a,\bar b}^\lambda
\begin{pmatrix}
&   t_2 & \\
s; & & t_1\\
&   t_3 &
\end{pmatrix}
= U_{\alpha^\pm_\lambda}
\begin{pmatrix}
&   t_2 & \\
s'; & & t_1\\
&   t_3 &
\end{pmatrix}
\een 
where $s'= \epsilon^\pm(b,\lambda)s$, which defines an intertwiner 
in $\Hom(\alpha^\pm_\lambda b,a)$. The claim then follows from the definition \eqref{eq:Odef} of the MPOs $O^L_{\alpha^\pm_\lambda}$ and the following lem. \ref{lem:4}.

3) These formulas follow from \eqref{pdef}, \eqref{86} combined with 1), 2)
and the fact that $\pi^L$ is a representation of the double triangle 
algebra. 

4) The relations for $Q^L_{\mu,\nu}$ except the formula for the adjoint follows immediately from the previous theorem and 
\eqref{qrel}. Applying $\pi^L$ to \eqref{YlmA} and using \eqref{YlmA} and item 1) of lem. \ref{Ylem}  gives \eqref{eq:qlm}. 
Then taking the adjoint and using $O_A^\dagger = O_A$ (see thm. \ref{thm1}) as well as item 2) of lem. \ref{Ylem} gives the claim.
The commutation relations follow 
from 1), 2), the fact that the triangle algebra elements 
$e_B, p^\pm_\mu$ are in the horizontal center, and the fact that 
the $q_{\mu,\nu}$ are the central projections.

5) By \cite{bockenhauer1999alpha}, thm. 6.8, we have $\sum_{\mu,\nu} q_{\mu,\nu} = e_0$, where $e_0$ is the projector $e_B$ of the double triangle algebra with $B=$
identity morphism. Applying $\pi^L$ gives 
$\sum_{\mu,\nu} Q_{\mu,\nu}^L = \pi^L(e_0)$.  Therefore, $\pi^L(e_0)$ is equal to $I^L$ since $e_0$ is the unit of the double triangle algebra. (One can also see using the explicit definition of $\pi^L(e_0)$ in terms of 
$6j$-symbols.)

6) Evaluating the definition of $q_{\mu,\nu}$ for $\mu=\nu=0=$ identity morphism in the triangle algebra 
\eqref{qdef} shows that it is equal to $\sum_{B \in \ _\cM X_\cM} e_B=f_0$, which is the ``horizontal unit'' in the double triangle algebra. Applying $\pi^L$ and using the definition of $P^L$ in theorem \ref{thm1} gives the result.

7)  This follows from 1), 2), 3) together with \cite{bockenhauer1999alpha}, thm. 5.10.  
\end{proof}

\begin{lemma}\label{lem:4}
The double triangle elements $p^\pm_\lambda$ have the MPO representations (where  $P,P' \in {\rm Path}^L_c$ and $L=2k$):
{\footnotesize
\ben
\begin{split}
& \hspace{3.5cm} \pi^L(p_\lambda^\pm)^P_{P'} :=  d_\lambda  
\sum_{a,b} \
\sum_{s,s_1, \dots, s_{L+1}}  \cdot \\
& \cdot \ 
{}^\pm Z_{a,\bar b}^\lambda
\begin{pmatrix}
& t_1 & \\
s; & & s_1 \\
& t_1' &
\end{pmatrix}
\prod_{j=1}^k \left\{
\overline U_{\alpha^\pm_\lambda}
\begin{pmatrix}
& t_{2j-1} & \\
s_{2j-1} & & s_{2j} \\
& t_{2j-1}' &
\end{pmatrix}
 U_{\alpha^\pm_\lambda}
\begin{pmatrix}
& t_{2j} & \\
s_{2j} & & s_{2j+1} \\
& t_{2j}' &
\end{pmatrix}
\right\}
{}^\pm Z_{a,\bar b}^\lambda
\begin{pmatrix}
&  & t_L\\
s; & s_{L+1} & \\
&  & t_L'
\end{pmatrix}
.\non
\end{split}
\een
}
\end{lemma}
\begin{proof}
We consider the ``$+$'' case and denote by $\{ v \}$ an ONB of partial isometries of $\Hom_\cM(\alpha^+_\lambda,B)$, 
where $B \in {}_\cM X_\cM, \lambda \in {}_\cN X_\cN$. Thus, $v^* v = 1$ and $\sum_{v} vv^* = 1$ in an obvious notation. 
By definition, this ONB has $\langle B, \alpha_\lambda^+ \rangle$ elements. For $t \in \Hom_\cM(a\bar b,B)$, $a,b \in {}_\cN X_\cM$, 
we write
\ben
s_{v,t,B} := \left[ \frac{d_a}{d_B} \right]^{1/2} \ \epsilon^+(\lambda, b) (v \times 1_b)(t^* \times 1_b)(1_a \times r_b), 
\een
with $r_b$ as usual a solution to the conjugacy relations. It follows that as $t,v$ run through ONBs of intertwiners, 
$s_{v,t,B}$ runs through an ONB of intertwiners in $\Hom_\cM(b\lambda,a)$. Letting $v,v'$ be two isometries from our 
ONB of $\Hom_\cM(\alpha^+_\lambda,B)$ for fixed $B$, the definitions give
\ben\label{rep1}
\left[ \frac{d_a}{d_B} \right]^{-1/2} \ {}^+ Z_{a,\bar b}^\lambda
\begin{pmatrix}
&  & t_0\\
s_{v,t,B}; & u_{v'} & \\
&  & t_0'
\end{pmatrix}
=
\delta_{v,v'} \ Z_{a,\bar b}^B
\begin{pmatrix}
&  & t_0\\
t; & u & \\
&  & t_0'
\end{pmatrix}
\een
where $u \in \Hom_\cM(B M_3,M_1), t_0 \in \Hom_\cM(b \overline{\imath}, M_2), t_0' \in \Hom_\cM(a \overline{\imath},M_1)$ members of 
ONBs of intertwiners and $M_i \in {}_\cM X_\cM$. See fig. \ref{fig:+Z} for the graphical illustration of these formulas. 
\begin{figure}[h!]
\begin{center}
\hspace*{-2.1cm}
  \includegraphics[width=1.2\textwidth,]{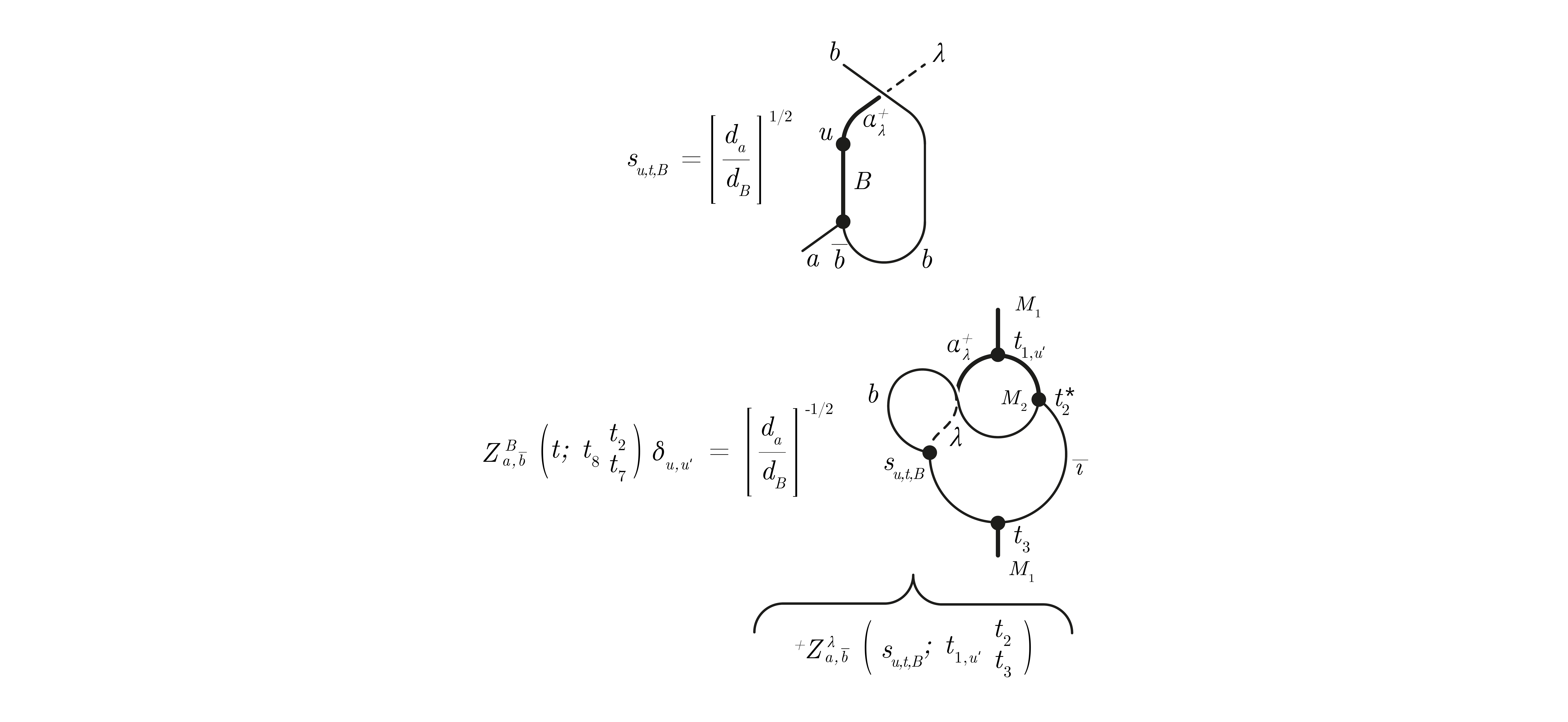}
  \end{center}
  \caption{\label{fig:+Z} Upper panel: the intertwiner $s_{u,t,B}$. Lower panel: 
  Relationship between $^\pm Z_{a,\bar b}^\lambda$ and $Z^B_{a,\bar b}$.}
\end{figure}
Here, $u_{v} := (v \times 1_{M_3})u$, which runs through an ONB of intertwiners in $\Hom_{\cM}(\alpha^+_\lambda M_3, M_1)$
as $v,u$ run through ONBs. We also observe that by construction
\ben\label{rep2}
U_{\alpha^+_\lambda}
\begin{pmatrix}
& t_{j} & \\
u_{j,v_{j}} & & u_{j+1,v_{j+1}} \\
& t_{j}' &
\end{pmatrix}
=
\delta_{v_j,v_{j+1}} \ U_{B}
\begin{pmatrix}
& t_{j} & \\
u_{j} & & u_{j+1} \\
& t_{j}' &
\end{pmatrix}
\een
for partial isometries $v_j$ from our 
ONB of $\Hom_\cM(\alpha^+_\lambda,B)$ for fixed $B$. A similar relation holds for the conjugate $6j$-symbol. 
We make the replacements \eqref{rep1}, \eqref{rep2} in each factor the on the right side of the following expression using an ONB of 
partial isometries $v_j \in \Hom_\cM(\alpha^+_\lambda,B)$ for each $j$
in the defining relation for $\pi^L$ multiplied by a Kronecker delta: 
{\footnotesize
\ben
\begin{split}
& \delta_{v,v'} \pi^L (e_{B; t', b',a'}^{t,b,a})^P_{P'} = \, 
[d(a) d(b) d(a') d(b') ]^{1/4} 
\sum_{u_1,\dots,u_{L+1}} \sum_{v_1,\dots,v_{L+1}} \delta_{v,v_1} \delta_{v_1,v_2} \cdots 
\delta_{v_{L},v_{L+1}} \delta_{v_{L+1},v'} \cdot \\
& \, 
Z_{a',\bar b'}^B
\begin{pmatrix}
& t_1 & \\
t'; & & u_{1} \\
& t_1' &
\end{pmatrix} 
\prod_{j=1}^{k} \left\{
\overline U_{B}
\begin{pmatrix}
& t_{2j-1} & \\
u_{2j-1} & & u_{2j} \\
& t_{2j-1}' &
\end{pmatrix}
 U_{B}
\begin{pmatrix}
& t_{2j} & \\
u_{2j} & & u_{2j+1} \\
& t_{2j}' &
\end{pmatrix}
\right\} \, 
Z_{a,\bar b}^B
\begin{pmatrix}
&  & t_{L}\\
t; & u_{L+1} & \\
&  & t_{L}'
\end{pmatrix}
\non
\end{split}
\een
}
and note that the delta-function implies that we can insert an additional summation over the $v_j$ in the resulting expression, so that we can 
subsequently use \eqref{rep2} and \eqref{rep1} 
on the terms in the product. 
Then setting $s:=s_{v,t,B}, s'=s_{v',t',B}$ resp. $s_j=u_{j,v_{j}}$ we obtain an equivalent sum over an ONB of intertwiners in 
$\Hom_\cM(b\lambda,a)$ resp $\Hom_{\cM}(\alpha^+_\lambda M_j, M_j')$.
In the resulting equality we set $v=v', a=a',b=b',s=s'$. The left side of the equality no longer depends on $v$, so if we perform a 
summation over $v \in \Hom_\cM(\alpha^+_\lambda,B)$ for fixed $B$, we obtain a factor of $\langle B, \alpha_\lambda^+ \rangle$ on 
that side. Then we divide by $d_B$ and take the sum over $B$ and then $a,b$. This gives among other things a summation over $s,s_j$ on the right side. 
Using the defining equation $e_B=\sum_{t,a,b} e_{B; t, b,a}^{t,b,a}$
results in
{\footnotesize
\ben
\begin{split}
& \hspace{1.5cm} \sum_B d_B^{-1} \langle B, \alpha_\lambda^+ \rangle \pi^{L}(e_B)^P_{P'} = 
\sum_{a,b} 
\sum_{s,s_1, \dots, s_{L+1}}  \cdot \\
& \cdot \ 
{}^+ Z_{a,\bar b}^\lambda
\begin{pmatrix}
& t_1 & \\
s; & & s_1 \\
& t_1' &
\end{pmatrix}
\prod_{j=1}^k \left\{
\overline U_{\alpha^+_\lambda}
\begin{pmatrix}
& t_{2j-1} & \\
s_{2j-1} & & s_{2j} \\
& t_{2j-1}' &
\end{pmatrix}
 U_{\alpha^+_\lambda}
\begin{pmatrix}
& t_{2j} & \\
s_{2j} & & s_{2j-1} \\
& t_{2j}' &
\end{pmatrix}
\right\}
{}^+ Z_{a,\bar b}^\lambda
\begin{pmatrix}
&  & t_L\\
s; & s_{L+1} & \\
&  & t_L'
\end{pmatrix}
.\non
\end{split}
\een
}
This is equivalent to the claimed statement using the definition of 
$\pi^L$ and of $p^\pm_\lambda$ \eqref{pdef}. 
\end{proof}

\subsection{Chain Hamiltonians, local operators}
\label{sec:local}

Consider a vector $|t_1, \dots, t_L\rangle \in \sV^L_{\rm open}$. 
According to the description of such vectors as sequences of compatible intertwiners, this may be interpreted as a specific intertwiner $\in \Hom(B (\i \bar \i)^k, C)$, where $L=2k$. On $\Hom(B (\i \bar \i)^k, C)$ the Jones projections $B(e_0), B(e_{-1}), \dots, B(e_{-(L-1)})$ act by left multiplication, so we have an action of the algebra $\cT\cL_{[0,L-1]}$ generated by these projections on $\sV^L_{\rm open}$. We call $E^L_x, x=1, \dots, L$ the operator induced by the Jones projection $e_{-(x-1)}$ on 
$\sV^L_{\rm open}$. Its matrix elements are
\ben
\label{exdef}
\langle t_1, \dots, t_L| E^L_x |t_1', \dots, t_L'\rangle
=
\begin{cases}
[(t_{x} \times 1_\i)t_{x+1}]^* e_0(t_{x}' \times 1_\i)t_{x+1}' & 
\text{$x$ even,}\\
[(t_x \times 1_{\bar \i})t_{x+1}]^* e_{-1}(t_x' \times 1_{\bar \i})t_{x+1}' &
\text{$x$ odd,}
\end{cases}
\een
noting that the expressions are scalars. Using the expressions $e_0 = d^{-1} \bar r \bar r^*, e_{-1} = d^{-1} \bar \i(rr^*)$ one easily 
finds the numerical value of the matrix elements to be 
\ben
\label{exdef1}
\begin{split}
&\langle t_1, \dots, t_L| E^L_x |t_1', \dots, t_L'\rangle \\
= &d^{-1} (t_{x+1}',\tilde t_x')(\tilde t_{x+1}, t_x)
\begin{cases}
\left( 
\frac{d(M_x') d(M_x)}{d(a_x') d(a_x)} 
\right)^{1/2} 
\delta_{a_x, a_{x+1}} \delta_{a_x',a_{x+1}'} 
& 
\text{$x$ even,}\\
\left( 
\frac{d(a_x') d(a_x)}{d(M_x') d(M_x)} 
\right)^{1/2} 
\delta_{M_x, M_{x+1}} \delta_{M_x',M_{x+1}'} 
&
\text{$x$ odd,}
\end{cases}
\end{split}
\een
where the intertwiners are supposed to be from the spaces $t_x \in \Hom(a_x \bar \i, M_x), t_{x+1} \in \Hom(M_x \i, a_{x+1})$ (and similarly $t_{x}',t_{x+1}'$) when $x$ is even 
and from $t_x \in \Hom(M_x  \i, a_x), t_{x+1} \in \Hom(a_x \bar \i, M_{x+1})$ (and similarly $t_{x}',t_{x+1}'$) when $x$ is odd. The inner products 
$(\tilde t_{x+1}, t_x)$ or $(t_{x+1}',\tilde t_x')$ are as in \eqref{scalts} and ``tilde'' means the Frobenius dual of an intertwiner, see \eqref{frobenius}.

By construction, the 
$E^L_x$ are projections. We denote by
\ben
\cT\cL_{[x_1,x_2]}^L := \text{algebra generated by $E_x^L, x_1<x<x_2-1$,}
\een
the algebra of local operators acting on the sites $x_1, x_1+1, \dots, x_2$. 
By construction, these algebras are isomorphic to Temperly-Lieb-Jones algebras with $x_2-x_1-1$ generators and parameter $d$. We have the following Lemma. 

\begin{lemma}
\label{lem:5}
Any element $A \in \cT\cL_{[2,L-1]}^L$ commutes with any MPO in $\pi^L(\lozenge)$, i.e. with the image of the double triangle algebra 
under the representation $\pi^L$ on $\sV^L_{\rm open}$.
\end{lemma}

\begin{proof}
We may write the projectors $E_x^L$ in terms of $6j$-symbols as in \cite{trebst2008short}. Then we can see that we can move 
$E_x^L$ at any point through the chain of $U_B, \bar U_B$'s as in fig.
\ref{fig:6} using a ``vertical'' version of the zipper lemma, which is proven in exactly the same way as the zipper lemma itself. 

With a certain amount of tedium, one can also show this directly without a graphical notation which we do here since we have not demonstrated the vertical zipper lemma. 
For definiteness, take $x$ even (the other case is similar).
At the level of matrix elements, the proof boils down to the statement that 
\ben
\label{commutatorUU}
\begin{split}
&\sum_{t,t'_x,t'_{x+1}}
U_B \left(
\begin{matrix}
& t'_x & \\
s_1 & & t \\
& t_x & 
\end{matrix}
\right)
\bar U_B \left(
\begin{matrix}
& t'_{x+1} & \\
t & & s_2 \\
& t_{x+1} & 
\end{matrix}
\right)
\langle \dots t''_x, t''_{x+1} \dots | E^L_x | \dots t'_x, t'_{x+1} \dots
\rangle\\
=
&\sum_{t,t'_x,t'_{x+1}}
U_B \left(
\begin{matrix}
& t''_x & \\
s_1 & & t \\
& t'_x & 
\end{matrix}
\right)
\bar U_B \left(
\begin{matrix}
& t''_{x+1} & \\
t & & s_2 \\
& t'_{x+1} & 
\end{matrix}
\right)
\langle \dots t'_x, t'_{x+1} \dots | E^L_x | \dots t_x, t_{x+1} \dots
\rangle
\end{split}
\een
where the intertwiners are supposed to be from the spaces 
$t_x \in \Hom(a_x \bar \i, M_x), t_{x+1} \in \Hom(M_x \i, a_{x+1}), 
t_x' \in \Hom(a_x' \bar \i, M_x'), t_{x+1}' \in \Hom(M_x' \i, a_{x+1}')
t_x \in \Hom(a_x'' \bar \i, M_x''), t_{x+1}'' \in \Hom(M_x'' \i, a_{x+1}'')$, 
as well as $s_1 \in \Hom(Ba_x',a_x), s_2 \in \Hom(Ba'_{x+1}, a_{x+1}), t \in 
\Hom(BM_x',M_x)$. Only two $6j$-symbols are involved because $E_x^L$ only 
acts non-trivially on the sites $x,x+1$. We now evaluate the left side of 
\eqref{commutatorUU} making use of the intertwiner calculus, the definitions of the 
$6j$-symbols, of \eqref{exdef}, and of $e_0 = d^{-1} \bar r \bar r^*$:
\ben
\begin{split}
=& d^{-1} (\tilde t_{x+1}'',t_x'')
\left( 
\frac{d(M_x'')}{d(a_x'')} 
\right)^{1/2} 
\delta_{a_x'',a_{x+1}''} 
\sum_{t,t'_x,t'_{x+1}} \\
&t^{\prime *}_x(1_{\bar \i} \times t^{\prime *}_{x+1})(\bar r \times 1_{a'_{x+1}}) \ t_x^*(s_1^* \times 1_{\bar \i})(1_B \times t_x')t \ t_{x+1}^* (t^* \times 1_\i)(1_B \times t'_{x+1}) s_2\\
=& d^{-1} (\tilde t_{x+1}'',t_x'')
\left( 
\frac{d(M_x'')}{d(a_x'')} 
\right)^{1/2} 
\delta_{a_x'',a_{x+1}''} 
\sum_{t'_x,t'_{x+1}} \delta_{a_x',a_{x+1}'}  \\
&t^{\prime *}_x(1_{\bar \i} \times t^{\prime *}_{x+1})(\bar r \times 1_{a'_{x}}) \ 
t_{x+1}^*(t_x^* \times 1_{\i})(s_1^* \times 1_{\bar \i \i})
(1_B \times t_x' \times 1_\i)(1_B \times t'_{x+1}) s_2\\
=& d^{-1} (\tilde t_{x+1}'',t_x'')
\left( 
\frac{d(M_x'')}{d(a_x'')} 
\right)^{1/2} 
\delta_{a_x'',a_{x+1}''} 
\sum_{t'_x} \delta_{a_x',a_{x+1}'} \\
&t^{\prime *}_x(1_{\bar \i} \times t^{\prime *}_{x+1})
(s_1^* \times 1_{\bar \i \i})
(1_B \times t_x' \times 1_\i)(1_B \times t_x^{\prime *} \times 1_\i)(s_2 \times \bar r)\\
=& d^{-1} (\tilde t_{x+1}'',t_x'') (\tilde t_{x+1},t_x)
\left( 
\frac{d(M_x'')d(M_x)}{d(a_x'')d(a_x)} 
\right)^{1/2} 
\delta_{a_x'',a_{x+1}''} \delta_{a_x,a_{x+1}} \delta_{s_1,s_2}
\end{split}
\een
We proceed in a similar manner evaluating the right side of 
\eqref{commutatorUU},
\ben
\begin{split}
=& d^{-1} (\tilde t_{x+1},t_x)
\left( 
\frac{d(M_x)}{d(a_x)} 
\right)^{1/2} 
\delta_{a_x,a_{x+1}} 
\sum_{t,t'_x,t'_{x+1}} \\
&
(1_{a_x'} \times \bar r^*)(t_x' \times 1_\i)t_{x+1}' \
t^{\prime *}_x 
(s_1^* \times 1_{\i})
(1_B \times t_x^{\prime\prime})t \
t^{\prime *}_{x+1}(t^* \times 1_\i)(1_B \times t_{x+1}'')s_2\\
\\
=& d^{-1} (\tilde t_{x+1},t_x)
\left( 
\frac{d(M_x)}{d(a_x)} 
\right)^{1/2} 
\delta_{a_x,a_{x+1}} 
\sum_{t'_x,t'_{x+1}} \delta_{a_x',a_{x+1}'}  \\
&t^{\prime *}_{x+1}(t^{\prime *}_x \times 1_\i)(s_1^* \times 1_{\bar \i \i}) (1_B \times t_x'' \times 1_\i)(1_B \times t_{x+1}'')s_2 \
(1_{a_x'} \times \bar r^*)(t_x' \times 1_\i)t'_{x+1}\\
=& d^{-1} (\tilde t_{x+1},t_x)
\left( 
\frac{d(M_x)}{d(a_x)} 
\right)^{1/2} 
\delta_{a_x,a_{x+1}} 
(s_1^* \times \bar r^*)(1_B \times t_x'' \times 1_\i)(1_B \times t_{x+1}'')s_2\\
=& d^{-1} (\tilde t_{x+1}'',t_x'') (\tilde t_{x+1},t_x)
\left( 
\frac{d(M_x'')d(M_x)}{d(a_x'')d(a_x)} 
\right)^{1/2} 
\delta_{a_x'',a_{x+1}''} \delta_{a_x,a_{x+1}} \delta_{s_1,s_2}
\end{split}
\een
which is the same as before and thus concludes the proof.
\end{proof}

We mention that for a closed spin chain (i.e. a state in $\sV^L$), we may also construct a Temperly-Lieb projection $E_x^L$ for $x=L$ 
which we think of as involving the sites $1$ and $L$ at least in certain cases\footnote{Alexander Stottmeister, private communication.}. In the literature, the resulting algebra is called the ``annular Temperly-Lieb algebra'' \cite{jones2006hilbert}. 
Assume that the Jones index $d^2 = [\cM:\cN]<4$. If $\omega:=i e^{i \pi / (2 (k+2))}, k=1,2,3 \dots$, then 
$d = -(\omega^2 + \omega^{-2})$ are precisely the possible values of the  square root of the Jones index below $2$ \eqref{jonesq}. 
It is standard to show that $B_x^L := \omega I^L + d\omega^{-1} E_x^L, x=1, \dots, L-1$ give a unitary representation of the braid 
group on $L$ strands. Now define the projection (see fig. \ref{fig:ATL})
\ben
E_L^L =
\left( \prod_{x=L-1}^{1} (B_x^L)^{-1} \right)
E_{L-1}^L
\left( \prod_{x=1}^{L-1} B_x^L \right)
\een
on the closed spin chain. Then $E^L_L$ is a projection and we have 
the additional relations $E_1^L E_L^L E_1^L = d^{-1} E_1^L, E_{L-1}^L E_L^L E_{L-1}^L = d^{-1} E_{L-1}^L, 
E_L^L E_1^L E_L^L = d^{-1} E^L_L, [E^L_x, E^L_L] = 0 \, (|x-L|>1)$ of the annular Temperly-Lieb algebra (using the opposite braiding $(B_x^L)^{-1}$ gives another representation.) 

\begin{figure}
\centering
\begin{tikzpicture}[scale=.3]
\draw[thick] (2,-2) -- (2,10);
\draw[thick] (4,-2) -- (4,10);
\draw[thick] (8,-2) -- (8,10);
\draw[thick] (10,-2) -- (10,10);
\draw[ line width=.3cm, color=white] (0,-2) .. controls (2,1) and (8,2) .. (9,3);
\draw[thick] (0,-2) .. controls (2,1) and (8,2) .. (9,3);
\draw[ line width=.3cm, color=white] (9,5) .. controls (8,6) and (2,7) .. (0,10);
\draw[thick] (9,5) .. controls (8,6) and (2,7) .. (0,10);
\filldraw[color=white] (8.5,3) -- (10.5,3) -- (10.5,5) -- (8.5,5) -- (8.5,3);
\draw[thick] (8.5,3) -- (10.5,3) -- (10.5,5) -- (8.5,5) -- (8.5,3);
\draw (8.5,4) node[anchor=west]{{\tiny $E^L_{L-1}$}};
\draw (-1,4) node[anchor=east]{$E^L_{L}:$};
\draw (7.5,4) node[anchor=east]{$\cdots$};
\end{tikzpicture}
  \caption{\label{fig:ATL} Wire diagram for the MPO $E^L_L$. The crossings represent $B_x^L := \omega I^L + d\omega^{-1} E_x^L$.
  }
\end{figure}
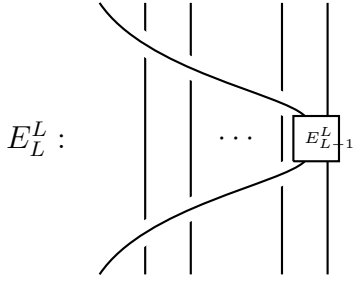

For the open (closed) spin chain, we define the Hamiltonian as 
\ben
H := J \sum_{x=1}^{L-1} E_x^L \quad (+J E^L_L),
\een
where $J$ is a coupling constant. By the previous lemma and cor. \ref{cor:triangle}, we have 
$[O_A^L, E^L_x] = [Q_{\lambda,\mu}^L, E^L_x] = 0$ for all $x \in [1,L]$ for the closed chain or all $x \in [1,L-1]$ for the open chain.
In either case we get, noting that $Q_{\lambda,\mu}^L$ leave the Hilbert space $\sV^L$ of the closed chain invariant:

\begin{corollary}
The local densities $E_x^L$ of the Hamiltonian  leave the ``sectors'' $\sV_{\lambda,\mu}^L:=Q^L_{\lambda,\mu} \sV^L \subset \sV^L$ invariant.
The same holds for the open chain.
\end{corollary}

Because the $Q^L_{\lambda,\mu}$ are non-zero mutually orthogonal projections whenever $Z_{\lambda,\mu} \neq 0$ (cor. \ref{cor:triangle}), 
it follows that we have a corresponding direct sum decomposition of the closed chain Hilbert space 
$\sV^L$ into the sectors $\sV_{\lambda,\mu}^L$, each of which is invariant under 
time evolution and local densities of the Hamiltonian. This is analogous to the case of a rational $1+1$-dimensional CFT, where 
the sectors of the total Hilbert space splits into sectors labelled by pairs of simple objects $(\lambda,\mu)$ -- which are Virasoro-irreps in that case and correspond to the 
decomposition of the partition function into $\sum_{\lambda,\mu} Z_{\lambda,\mu} \chi_\lambda(q) \overline{\chi_\mu(q)}$ -- 
each of which is invariant under the densities of the Hamiltonian (Virasoro algebra). 

\section{Defects}

\subsection{Defect algebra}
In this section, we come to the main results of this paper which 
concern defects. A defect is roughly a kind of ``boundary condition'' 
for a bi-partite closed spin-chain. In our approach, 
each defect corresponds to an appropriately defined orthogonal summand of the Hilbert space for the full bipartite chain, 
and is invariant under any local operator as defined in sec. \ref{sec:local} which is not acting on the boundary points separating the sub-chains. 
As we will see, our ``boundary conditions'' are very similar to the classification of transparent boundaries in 1+1 CFTs, see sec. \ref{sec:transparent}, 
which is in fact the prime motivation for our construction. The subspaces corresponding to the defects are the eigenspaces of an abelian algebra, 
$\cD^{L_1,L_2}$, generated by certain operators $\Psi^{L_1,L_2}_{\nu,\rho; w_1, w_2}$ soldering the two subchains of lengths $L_1, L_2$ together.\footnote{Throughout, 
$L$ and $L_1,L_2$ are even numbers as before, due to the fact that we alternate between $\imath$ and $\bar \imath$
in the construction of the chain Hilbert space $\sV^L$. of the subchain.}  

Since the defects ``sit'' where the two sub-chains meet at some time zero snapshot of the chain, we think of them as ``vertical'' and they would remain in a fixed position if the chain 
was evolved with a Hamiltonian $H^{L_1} + H^{L_2}$ which is a sum of a Hamiltonian made from local operators of subchain 1 and 2 which do not touch the endpoints of the subchains.
The operators $O_A^L$ associated with a single chain (see thm. \ref{thm1}) are also often discussed in the literature in connection with defects. We 
tend to think of them as ``horizontal''. The algebra of vertical and horizontal defects turns out to be isomorphic to the fusion algebra ${}_\cM X_{\cM}$ with structure constants $N_{A,B}^C$
in the case of diagonal theories, (see sec. \ref{sec:specialcases}) although they of course correspond to entirely different representations of that algebra. In the general case, the
algebras seem to be non-isomorphic leaving the need for further investigations. Defects from the point of view of PEPs are discussed in app. \ref{PEPs}, but also here the 
precise connection needs clarification.  

After these preliminaries, we now come the definition of the MPOs $\Psi^{L_1,L_2}_{\nu,\rho; w_1, w_2}$ and their structural properties.
The operators $\Psi^{L_1,L_2}_{\nu,\rho; w_1, w_2}$  
act on the tensor product of two open chains, i.e. the 
Hilbert space 
\ben
\sV^{L_1,L_2}:=\sV^{L_1}_{\rm open} \otimes \sV^{L_2}_{\rm open}. 
\een
However, they 
couple the two tensor copies together at the end-points 
in each of the chains and therefore connect the two systems.
These operators generate the algebra $\cD^{L_1,L_2}$ and 
will be labelled by a pair $\nu, \rho \in \, _{\cN} X_\cN$ and a pair $w_1
\in \Hom(\alpha_\nu^-, \alpha_\rho^+), w_2 \in \Hom(\alpha_\rho^+, \alpha_\nu^-)$, 
so that both $w_1, w_2$ may be thought of as running between $1, \dots, Z_{\rho,\nu}$. 
The operators will only be defined for pairs $(\nu,\rho)$ such that $Z_{\rho,\nu} \neq 0$, so 
there are 
\ben
\sum_{\nu,\rho \in \, _\cN X_\cN} Z_{\nu,\rho}^2 = | \, _\cM X_\cM \, |
\een
such operators, using \cite{bockenhauer1999alpha}, cor. 6.10 in the equality sign. 
This suggests that the joint eigenspaces of the 
$\Psi^{L_1,L_2}_{\nu,\rho; w_1, w_2}$ and the central projections of
$\cD^{L_1,L_2}$ are labelled by the objects $A$ of $\, _\cM X_\cM$, 
as we will confirm below.

To start, we consider the projection operators 
$Q^{L_1}_{\mu_1,\lambda_1}$ and $Q^{L_2}_{\mu_2,\lambda_2}$ 
and according to our cor. 1, item 5), we can 
decompose the Hilbert space $\sV^{L_1,L_2}$ of the full bipartite chain 
into an orthogonal direct sum with summands given by 
\ben
\sV_{\mu_1, \lambda_1; \mu_2, \lambda_2}^{L_1,L_2} 
:= Q^{L_1}_{\mu_1,\lambda_1} \sV^{L_1}_{\rm open} 
\otimes Q^{L_2}_{\mu_1,\lambda_2} \sV^{L_2}_{\rm open} . 
\een
The operators $\Psi^{L_1,L_2}_{\nu,\rho; w_1, w_2}$ 
will be engineered to 
have the following properties: 
\begin{enumerate}
    \item (Identity) For $(\nu,\rho)=(id,id)$ 
    the identity endomorphisms, we have 
    \ben
    \Psi^{L_1,L_2}_{id,id;1,1} = I^{L_1} \otimes I^{L_2}, 
    \een
    where $1$ is the only unitary intertwiner in $\Hom(\alpha^+_{id},\alpha^-_{id})=\Hom(id,id) = \bC 1$ and $I^{L_i}$ is the identity operator on the closed
    chain $\sV^{L_i}$ (viewed as a projection operator on the open chain Hilbert space).
    \item (Conjugate) We have
    \ben
    (\Psi^{L_1,L_2}_{\nu,\rho; w_1,w_2})^\dagger 
    = \Psi^{L_1,L_2}_{\nu,\rho; w_2^*, w_1^*}
    .
    \een
    \item (Commutativity) 
    We have
    \ben
    \label{Com1}
    [\Psi^{L_1,L_2}_{\nu,\rho; w_1,w_2}, 
    \Psi^{L_1,L_2}_{\nu',\rho'; w_1', w_2'}] = 0.
    \een
    If $A \in \cT\cL^{L_1}_{[2,L_1-1]} \otimes \cT\cL^{L_2}_{[2,L_2-1]}$ 
    is any local operator on the doubled chain Hilbert space 
    $\sV^{L_1,L_2}=\sV^{L_1}_{\rm open} \otimes \sV^{L_1}_{\rm open}$ acting 
    on any sites of chain 1 or chain 2 except for the endpoints $\{1,L_i\}$ of either chain, then
    \ben
     \label{Com2}
    [\Psi^{L_1,L_2}_{\nu,\rho;w_1,w_2}, A] = 0.
    \een
    \item (OPE) We have 
    \ben
    \label{OPE}
    \Psi^{L_1,L_2}_{\nu,\rho;w_1,w_2} 
    \Psi^{L_1,L_2}_{\nu',\rho';w_1',w_2'} = 
    \sum_{\nu'',\rho'',w_1'',w_2''} c^{\nu'',\rho'',w_1'',w_2''}_{\nu',\rho',w_1',w_2'; \nu,\rho;w_1,w_2}
    \Psi^{L_1,L_2}_{\nu'',\rho'';w_1'',w_2''}
    \een
    where only such terms occur in the sum for which $\nu'' \subset \nu \nu', \rho'' \subset \rho \rho'$, i.e. which are compatible with fusion and. The numerical 
    coefficients are defined as
    \ben
    \label{cdef}
    c^{\nu'',\rho'',w_1'',w_2''}_{\nu',\rho',w_1',w_2'; \nu,\rho;w_1,w_2}
    := \sum_{e,f} \eta^{w_1''}_{e,f;w_1, w_1'}
    \eta^{w_2''}_{f^*,e^*;w_2, w_2'}
    \een
    where $e,f$ run over ONBs of $\Hom(\nu \nu',\nu'')$ respectively 
    $\Hom(\rho \rho',\rho'')$, and where $\eta^{w''}_{e,f;w, w'}$ are 
    Rehren's structure constants \eqref{rehren}. 
    
    \item (Fusion)  
    \ben
    \Psi^{L_1,L_2}_{\nu,\rho;w_1,w_2} 
    \sV^{L_1,L_2}_{\lambda_1,\mu_1;\lambda_2,\mu_2} \subset
    \bigoplus_{\lambda_1',\mu_1';\lambda_2',\mu_2'} 
    \sV^{L_1,L_2}_{\lambda_1',\mu_1';\lambda_2',\mu_2'} ,
    \een
    where only such terms occur in the sum which have $\mu_i' \subset \nu \mu_i, \lambda_i' \subset \rho \lambda_i$, i.e. which are compatible with fusion. 
\end{enumerate}

We now define the operators $\Psi^{L_1,L_2}_{\nu,\rho; w_1, w_2}$. We 
denote by $e_i$ respectively $f_i$ ONBs of 
\bena
\label{eifidef}
e_1 &\in \Hom(\mu_1 \nu, \mu_1')\\
f_1 &\in \Hom(\rho \lambda_1, \lambda_1')\\
e_2 &\in \Hom(\mu_2 \rho, \mu_2')\\
f_2 &\in \Hom(\nu \lambda_2, \lambda_2'),
\eena
where here and in the rest of this section, $i=1,2$ refers to the different subsystems of the bipartite periodic chain. Then we define MPOs $X_{e_1,f_2; s,t}^{s',t'}$ and $X_{e_2,f_1; s,t}^{s',t'}$
acting on two adjacent sites in the chain as in fig. \ref{fig:220413_Fig-1_TE}, where the down and right pointing triangles denote $6j$-symbols 
defined by figs. \ref{fig:220413_Fig-2_TE} and \ref{fig:220413_Fig-3_TE}.

\begin{figure}
\begin{center}
  \includegraphics[width=0.9\textwidth,]{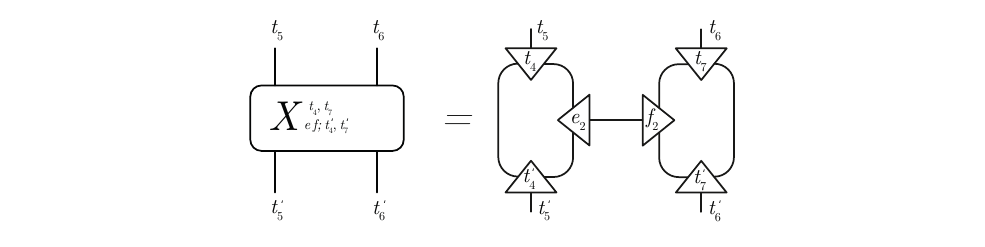}
  \end{center}
  \caption{  \label{fig:220413_Fig-1_TE} The definition of $X_{e_1,f_2; s,t}^{s',t'}$.}
\end{figure}

\begin{figure}
\begin{center}
  \includegraphics[width=0.9\textwidth,]{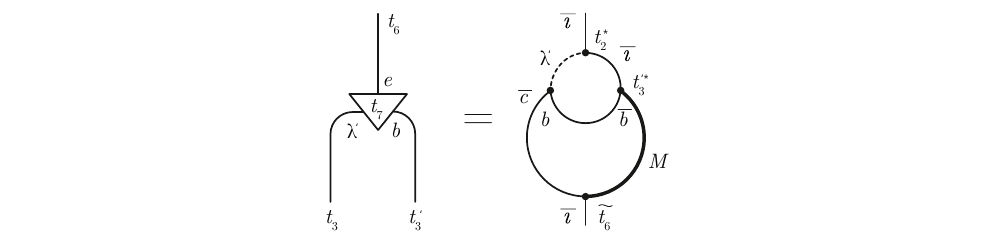}
  \end{center}
  \caption{  \label{fig:220413_Fig-2_TE} Downwards-pointing triangles.}
\end{figure}

\begin{figure}
\begin{center}
  \includegraphics[width=0.9\textwidth,]{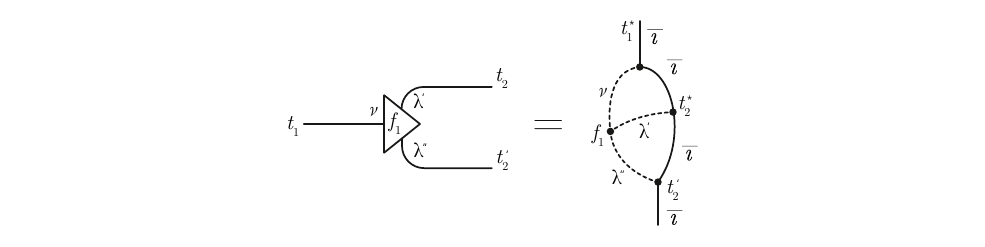}
  \end{center}
  \caption{  \label{fig:220413_Fig-3_TE} Leftwards-pointing triangles.}
\end{figure}

These figures also implicitly define the intertwiner spaces associated with the indices of these MPOs. The up- and left-pointing triangles are defined as their complex conjugates.
Next, we let $a,b,c,d \in \, _\cM X_\cN$, $\mu, \lambda \in \, _\cN X_\cN$, $t \in \Hom(\lambda c, b), 
s \in \Hom(c\mu, d), v \in \Hom(\alpha^-_\mu, \alpha^+_\lambda)$ and put 
\ben
\label{xbdef}
\bar x_{t,s,v,a} := (d_\mu d_\lambda)^{1/4} \left( \frac{d_a}{d_c D_X} \right)^{1/2} \, 
(s^* \times 1_a)
(1_c \times \epsilon^-(\mu, \bar a)^*)
(1_{c\bar a} \times v)
(1_c \times \epsilon^+(\lambda, \bar a))
(t \times 1_{\bar a}) \quad \in \lozenge, 
\een
and then 
\ben
\label{Phi}
\Phi_{t,s,v}^L := \sum_{a \in \, _\cN X_\cN} \pi^L(\bar x_{t,s,v,a}). 
\een
This definition is illustrated in the upper box of fig. \ref{fig:220413_Fig-4_TE}
which also shows in the lower box a graphical illustration of the adjoint. Note that 
the intertwiner $x_{t,s,v,a} \in \lozenge$ is depicted following the 
conventions by \cite{bockenhauer1999alpha}, i.e. its graphical representation as a wire diagram is 
rotated by 90 degrees relative to our normal conventions for pictures. 
The formula for the adjoint is 
\ben
\label{bPhi}
(\Phi_{t,s,v}^L)^\dagger := \sum_{a \in \, _\cN X_\cN} \pi^L(x_{t,s,v,a}), 
\een
wherein
\ben
\label{xdef}
x_{t,s,v,a} := (d_\mu d_\lambda)^{1/4} \left( \frac{d_a}{d_c D_X} \right)^{1/2} \, 
(1_a \times \bar s^*)
(\epsilon^-(a, \bar \mu)^* \times 1_c)
(\bar v \times 1_{a \bar c})
(\epsilon^+(a, \bar \lambda) \times 1_{\bar c})
(1_a \times \bar t) \quad \in \lozenge.
\een

\begin{figure}[h!]
\begin{center}
  \includegraphics[width=0.9\textwidth,]{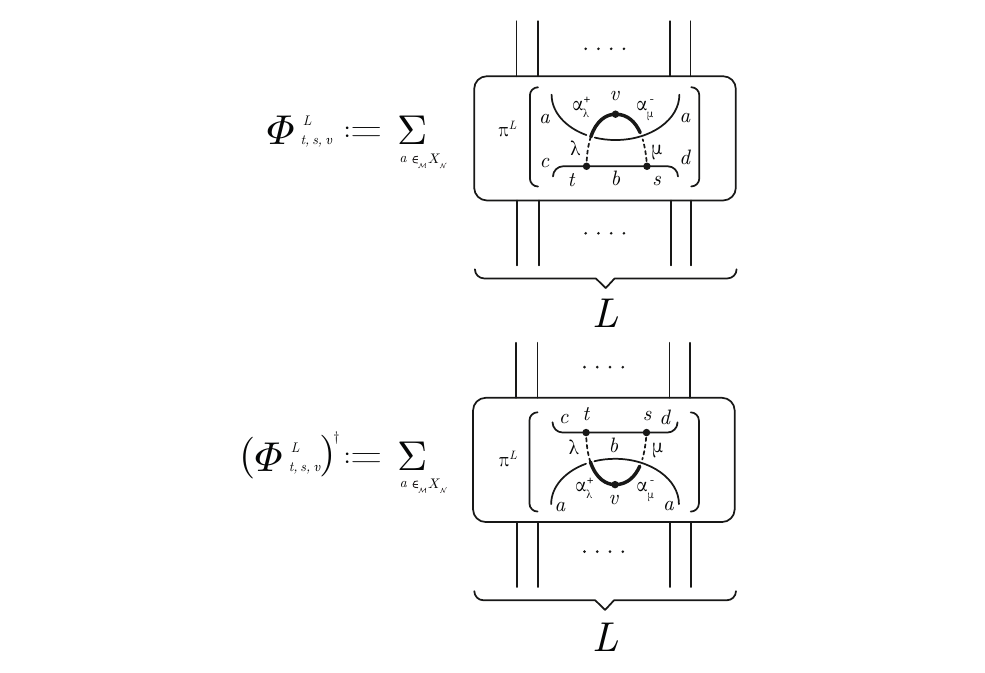}
  \end{center}
  \caption{  \label{fig:220413_Fig-4_TE} The definition of $\Phi_{t,s,v}^L$.}
\end{figure}

Next, we set
\ben
\begin{split}
\zeta_{e_1, f_1; w_1 v_1}^{v_1'} &:= 
\left[ \frac{d(\mu_1) d(\nu)}{d(\mu_1') d(\theta)} \right]^{1/2}
E[
e^*_1 \epsilon^-(\mu_1,\nu)^* (w_1 \times v_1) f_1(v_1')^*
]\\
\zeta_{e_2, f_2; w_2 v_2}^{v_2'} &:= 
\left[\frac{d(\mu_2) d(\rho)}{d(\mu_2') d(\theta)} \right]^{1/2}
E[
e^*_2 \epsilon^-(\mu_2,\rho)^* (w_2 \times v_2) f_2(v_2')^*
]
\label{Rehrenz}
\end{split}
\een
with $v_i \in \Hom(\alpha^-_{\mu_i}, \alpha^+_{\lambda_i}), v_i \in \Hom(\alpha^-_{\mu_i'}, \alpha^+_{\lambda_i'})$ and $w_i$ as above, 
see fig. \ref{fig:220413_Fig-5_TE}. Note the similarity to Rehren's 
structure constants \eqref{rehren} $\eta_{e_2, f_2; w_2, v_2}^{v_2'}$
which differ only by the braiding operators. 

\begin{figure}[h!]
\begin{center}
  \includegraphics[width=0.9\textwidth,]{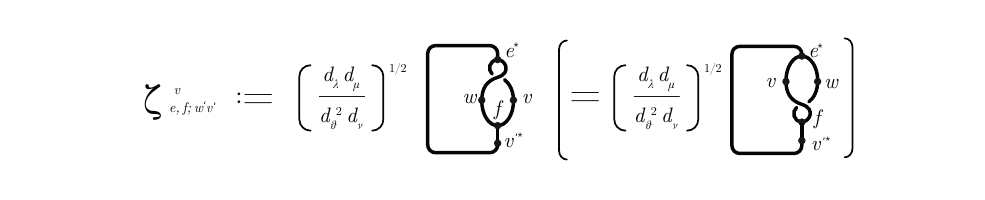}
  \end{center}
  \caption{  \label{fig:220413_Fig-5_TE} Wire diagram for the $\zeta$-structure constants \eqref{Rehrenz}.}
\end{figure}

\begin{definition}\label{def:psi}
We define the following MPO's on $\sV^{L_1,L_2}$:
\ben
\begin{split}
\label{Psidef}
   & \Psi^{L_1,L_2}_{\nu,\rho; w_1, w_2} := 
   \sum_{t_i, t_i', s_i, s_i', e_i, f_i, v_i, v_i'} 
   \zeta_{e_1, f_1; w_1, v_1}^{v_1'}
   \zeta_{e_2, f_2; w_2, v_2}^{v_2'} \cdot \\
   &(\Phi^{L_1}_{t_1,s_1,v_1} \otimes \Phi^{L_2}_{t_2,s_2,v_2})
   (X^{s_2,t_1}_{e_2, f_1;s_2',t_1'} \otimes I^{L_1-2} \otimes
   X^{s_1,t_2}_{e_1, f_2;s_1',t_2'} \otimes I^{L_2-2})
   (\Phi^{L_1}_{t_1',s_1',v_1'} \otimes \Phi^{L_2}_{t_2',s_2',v_2'})^\dagger, 
\end{split}
\een
This definition is illustrated in fig. \ref{fig:220413_Fig-6_TE}.

\begin{figure}[h!]
\begin{center}
  \includegraphics[width=0.9\textwidth,]{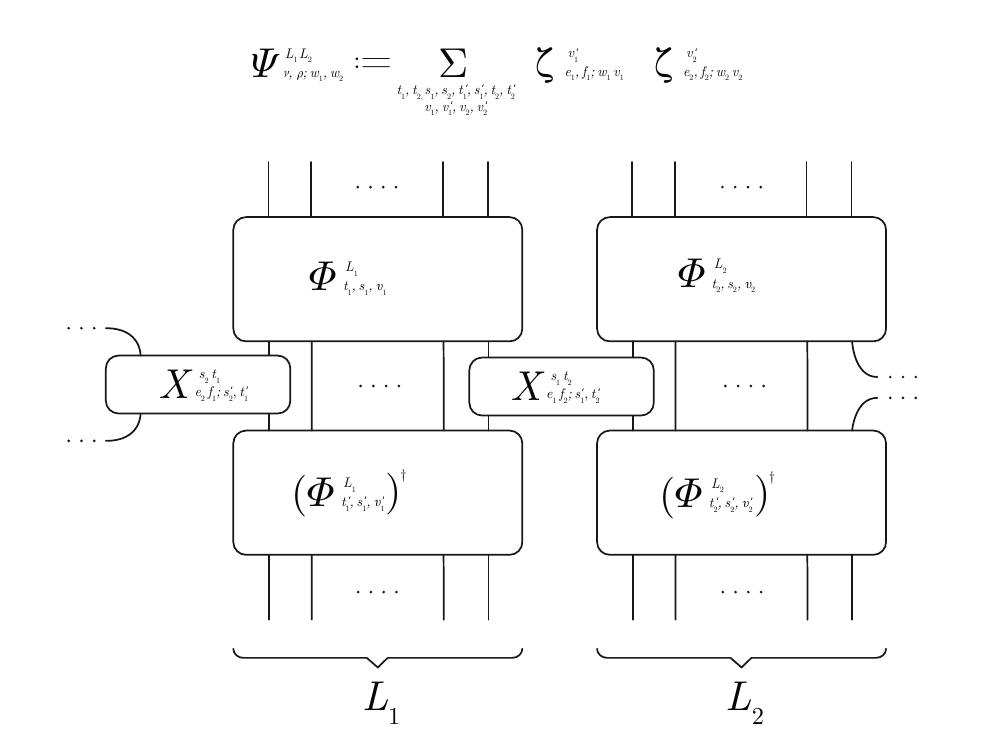}
  \end{center}
  \caption{  \label{fig:220413_Fig-6_TE} Definition of the MPOs $\Psi^{L_1,L_2}_{\nu,\rho; w_1, w_2}$.}
\end{figure}

The sums run over ONBs of intertwiners $e_i, f_i$ as in \eqref{eifidef} 
for fixed $\nu, \rho$ and $w_i$ and over ONBs
\begin{equation}
\begin{split}
\label{tisidef}
&t_i \in \Hom(\lambda_i \bar b_i, \bar c_i), \quad 
s_i \in \Hom(\mu_i \bar b_i, \bar d_i), \\ 
&t_i' \in \Hom(\lambda_i' \bar b_i', \bar c_i'), \quad 
s_i' \in \Hom(\mu_i' \bar b_i', \bar d_i').
\end{split}  
\end{equation}
\end{definition}

\begin{definition}
We define the {\em defect algebra} $\cD^{L_1,L_2}$ to be the $\dagger-$ algebra generated by the 
operators $\Psi^{L_1,L_2}_{\nu,\rho; w_1, w_2}$. 
\end{definition}

The following is one of the main results of this paper and 
clarifies the nature of the defect algebra. 

\begin{theorem}
\label{thm:3}
The MPOs $\Psi_{\nu,\rho; w_1, w_2}^{L_1,L_2}$ as in def. \ref{def:psi}
fulfill properties 1)--5). 
\end{theorem}

\begin{proof}
1) In this case $f_i = e_i = 1$ and consequently 
$X^{s_2,t_1}_{e_2, f_1;s_2',t_1'} 
= \delta^{s_2}_{s_2'} \delta^{t_1}_{t_1'}$ 
and
$X^{s_1,t_2}_{e_1, f_2;s_1',t_2'} 
= \delta^{s_1}_{s_1'} \delta^{t_2}_{t_2'}$, 
as well as 
$\zeta_{e_1, f_1; 1, v_1}^{v_1'} = \delta^{v_1}_{v_1'}$ 
and
$\zeta_{e_2, f_2; 1, v_2}^{v_2'} = \delta^{v_2}_{v_2'}$.
Therefore, we obtain 
\ben
\Psi^{L_1,L_2}_{id,id; 1, 1} = 
   \sum_{t_i, s_i, v_i} 
   (\Phi^{L_1}_{t_1,s_1,v_1} \otimes \Phi^{L_2}_{t_2,s_2,v_2})
   (\Phi^{L_1}_{t_1,s_1,v_1} \otimes \Phi^{L_2}_{t_2,s_2,v_2})^\dagger.
\een
Using now the representations \eqref{Phi}, \eqref{bPhi} of $\Phi_{t,s,v}^L, (\Phi_{t,s,v}^L)^\dagger$ 
in terms of elements of the double triangle algebra $\lozenge$, see fig. \ref{fig:220413_Fig-4_TE}, 
the representation property of $\pi^L$ (thm. \ref{thm:2}), 
the relations in $\lozenge$ and the expression for $q_{\mu,\nu}$ in the double 
triangle algebra, fig. \ref{fig:220413_Fig-19_TE} and eq. \eqref{qdef}, the expressions \eqref{xdef}, \eqref{xbdef},
and the product $\star$ in the double triangle algebra, 
we see that
\ben
\label{phiphi}
\begin{split}
\sum_{v_i, t_i, s_i} \Phi^{L_i}_{t_i,s_i,v_i} (\Phi^{L_1}_{t_i,s_i,v_i})^\dagger
&=\sum_{v_i, t_i, s_i,a,b} \pi^{L_i}(\bar x_{t_i,s_i,v_i,a} \star x_{t_i,s_i,v_i,b}) \\
&=\sum_{\mu_i,\lambda_i} \pi^{L_i}(q_{\mu_i,\lambda_i})\\
&= \sum_{\mu_i,\lambda_i} Q^{L_i}_{\mu_i,\lambda_i}, 
\end{split}
\een
where the sum is over ONBs of intertwiners $t_i, s_i$ as in \eqref{tisidef}
for fixed $\mu_i, \lambda_i$ and then (implicitly) over $\mu_i, \lambda_i$. The statement now follows form 
cor. \ref{cor:triangle}, item 5), because $\sum_{\mu_i,\lambda_i} Q^{L_i}_{\mu_i,\lambda_i} = I^{L_i}$
acts as the identity on the closed chain, and as a projector onto states satisfying periodic boundary conditions 
on the open chain.

\medskip 
2) This is a straightforward consequence of the definition.

\medskip 
3) The first statement \eqref{Com1} follows from 4) as follows. In 
\eqref{cdef}, we can replace the summation over $e,f$ by a summation 
over $\epsilon^-(\nu,\nu')e',\epsilon^+(\rho,\rho')f'$ where 
$e',f'$ run over ONBs of $\Hom(\nu' \nu,\nu'')$ respectively 
    $\Hom(\rho' \rho,\rho'')$, because the braiding operators are unitary so this corresponds to a unitary change of bases in these Hom-spaces. 
    Then we can use the functoriality 2) property of the braiding operators 
    in the $\alpha$-induction construction inside Rehren's structure constants
    \eqref{rehren} to cancel the braiding operators and thereby see that the structure constants 
    are symmetric under an exchange of $(\nu,\rho,w_1,w_2) \leftrightarrow
    (\nu',\rho',w_1',w_2')$, i.e. the operator algebra \eqref{OPE} is abelian. 
    
    The second statement \eqref{Com2} follows because the local operators 
    $A$ as in the statement commute with the MPOs $\Phi^{L_i}_{t_i,s_i,v_i}$ by lem. \ref{lem:5}. 
    Since the sites on which $A$ act are distinct from the sites on which the operators 
    $X$ in \eqref{Psidef} act, $A$ therefore commutes with 
    \eqref{Psidef}. 
    
    \medskip 
4) This part of the proof is the most involved. It is carried out largely in terms of pictures where several computations involving $6j$-symbols work similarly as in the proof of thm. \ref{thm:2} and are therefore not spelled out to the last detail. 
First we write out the product in \eqref{OPE}. 
Focussing on the right system $L_2$ and using \eqref{phiphi} we obtain 
fig. \ref{fig:220413_Fig-7_TE} using the representation property of $\pi^{L_2}$, the relations in the double triangle algebra $\lozenge$,
and \cite{bockenhauer1999alpha}, lem. 6.2 (where the normalization factor on the right side of fig. 60 is $\delta_{v_2', v_2''}$) to obtain the middle box. 

\begin{figure}[h!]
\begin{center}
  \includegraphics[width=0.9\textwidth,]{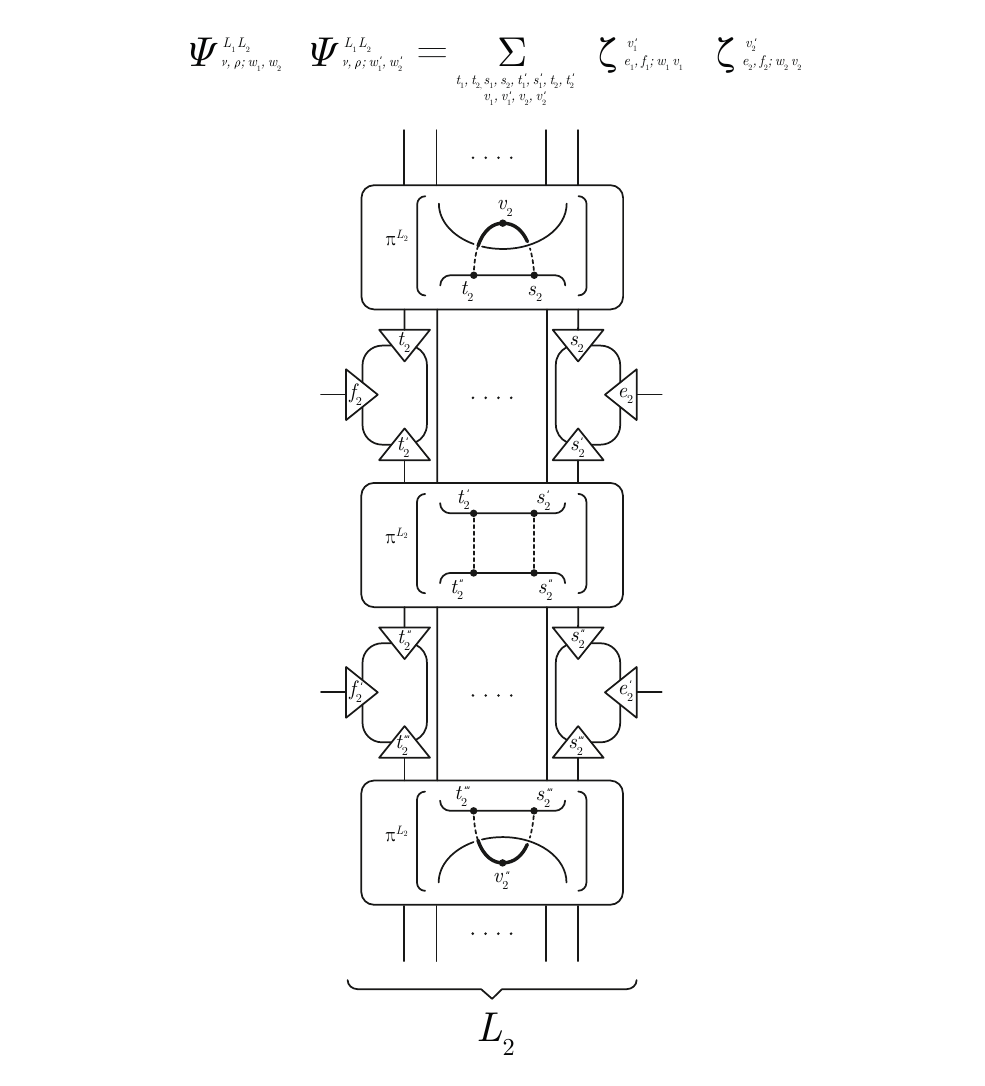}
  \end{center}
  \caption{  \label{fig:220413_Fig-7_TE} Right half of the graphical representation of the left side of
  \eqref{OPE}.}
\end{figure}

Zooming in onto the middle part of this figure,
we first insert a thick line and a summation over $t,A$ as in the 
first step in fig. \ref{fig:220413_Fig-8_TE}. 

\begin{figure}[h!]
\begin{center}
  \includegraphics[width=0.9\textwidth,]{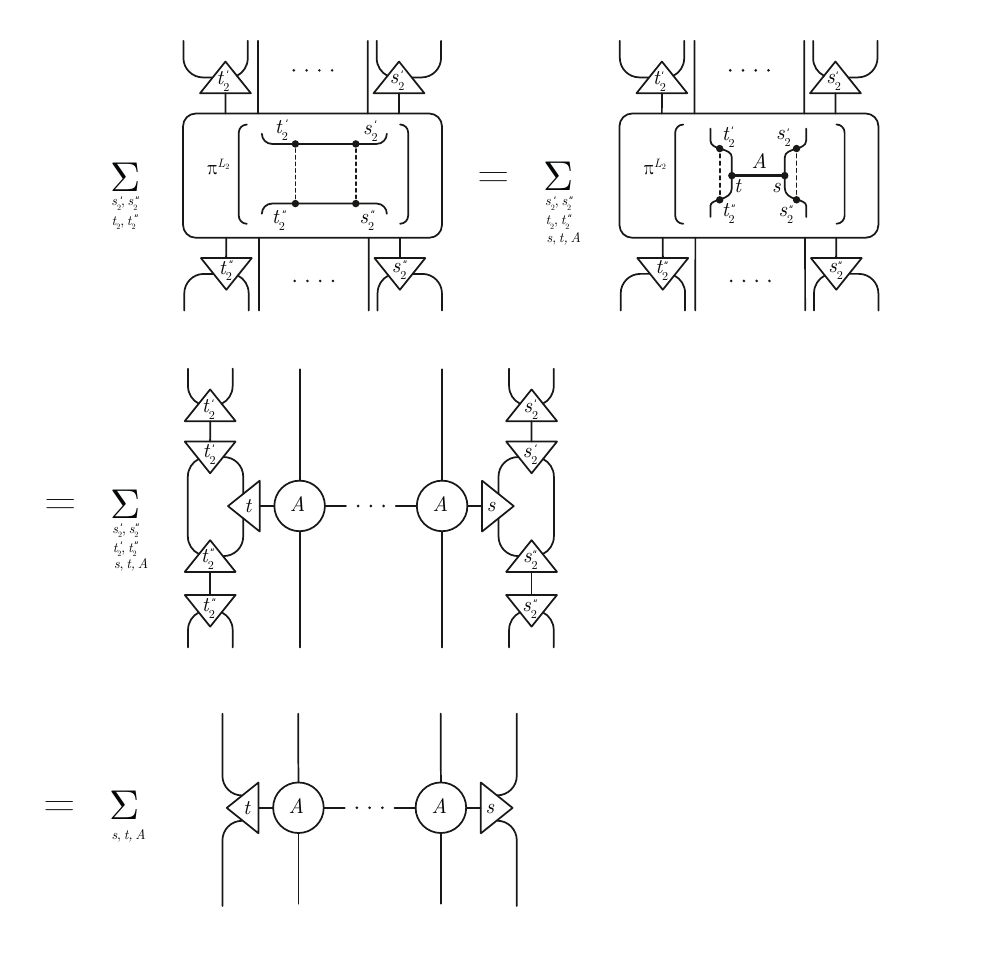}
  \end{center}
  \caption{  \label{fig:220413_Fig-8_TE} Manipulation of the middle part of fig. \ref{fig:220413_Fig-7_TE}.}
\end{figure}

Then we write out the 
definition of our representation $\pi^{L_2}$ of $\lozenge$, corresponding 
to the second step in fig. \ref{fig:220413_Fig-8_TE}. The summations over
$t_2, t_2', s_2, s_2'$ can now be carried out yielding the final 
diagram in fig. \ref{fig:220413_Fig-8_TE}. Inserting the identity of 
fig. \ref{fig:220413_Fig-8_TE} into fig. \ref{fig:220413_Fig-7_TE}, we obtain  fig. \ref{fig:220413_Fig-9_TE}, using also the relation for $P^{L_2}$
in cor. 1, item 6). 

\begin{figure}[h!]
\begin{center}
  \includegraphics[width=0.9\textwidth,]{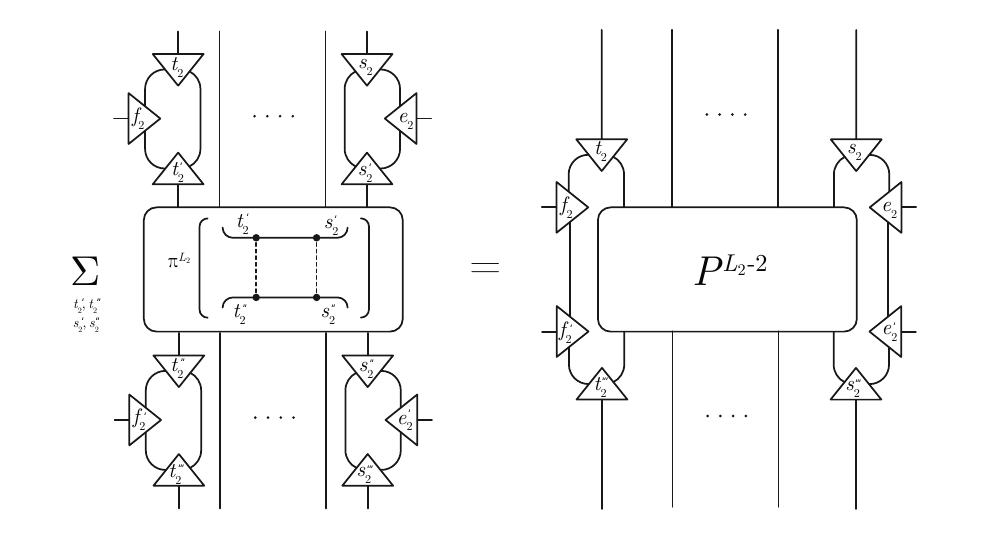}
  \end{center}
  \caption{  \label{fig:220413_Fig-9_TE} Result of manipulations on fig. \ref{fig:220413_Fig-8_TE}.}
\end{figure}

Considering the identity of fig. \ref{fig:220413_Fig-9_TE} applied to both subchains $L_1$ and $L_2$, we obtain fig. \ref{fig:220413_Fig-10_TE}. 

\begin{figure}[h!]
\begin{center}
  \includegraphics[width=0.9\textwidth,]{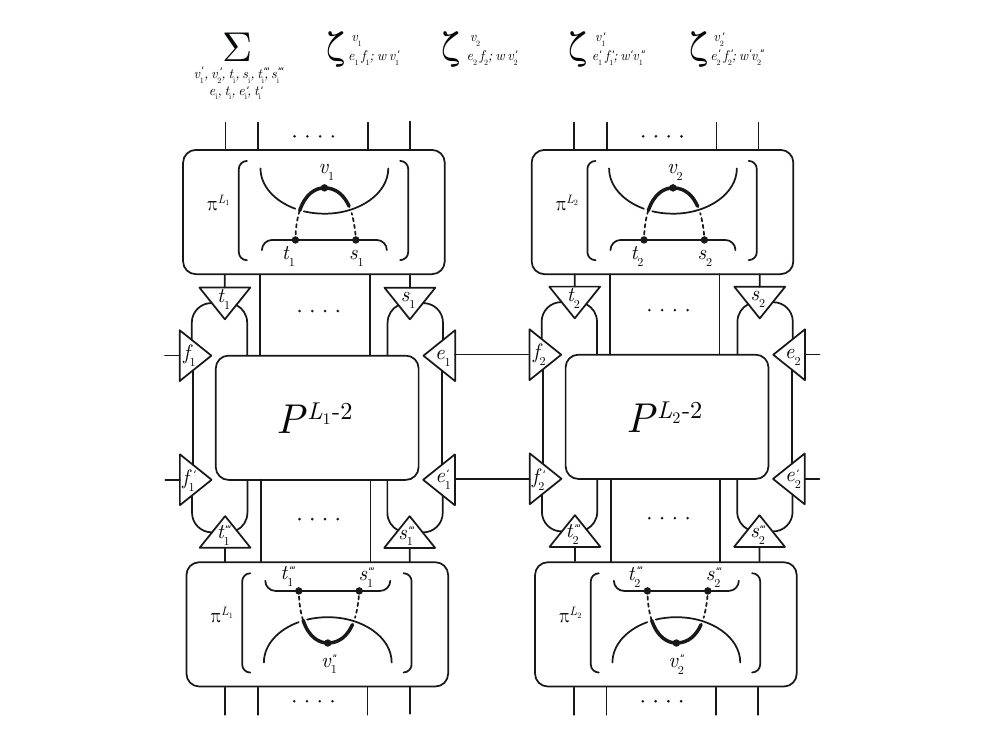}
  \end{center}
  \caption{  \label{fig:220413_Fig-10_TE} Result of manipulations on fig. \ref{fig:220413_Fig-9_TE}.}
\end{figure}

We next focus on the middle parts of this figure involving the 
MPOs $P^{L_1-2}$ and $P^{L_2-2}$. 

\begin{figure}[h!]
\begin{center}
  \includegraphics[width=0.9\textwidth,]{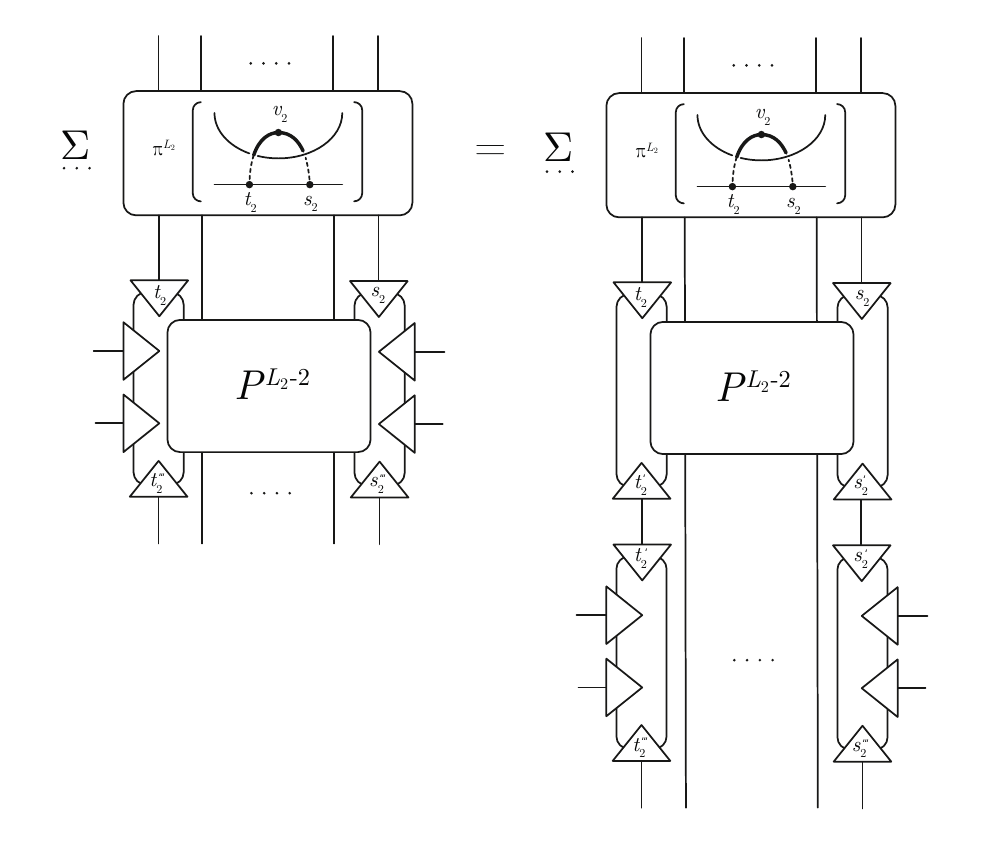}
  \end{center}
  \caption{  \label{fig:220413_Fig-4a_TE} }
\end{figure}

\begin{figure}[h!]
\begin{center}
  \includegraphics[width=0.9\textwidth,]{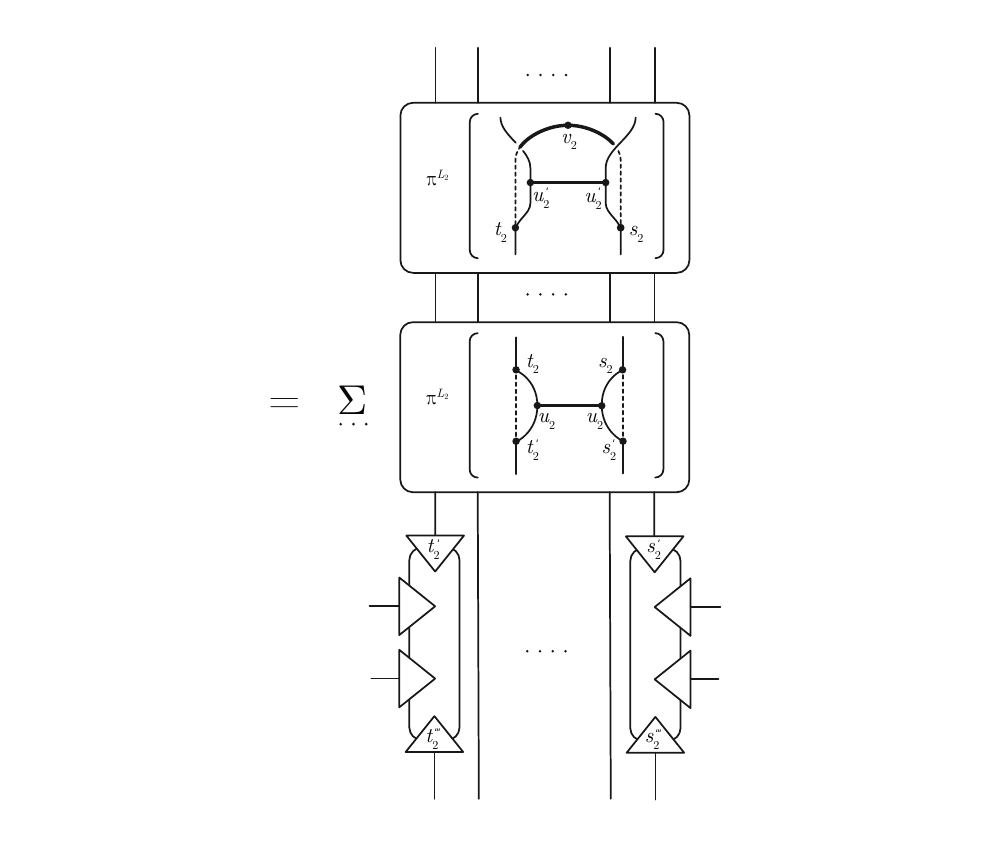}
  \end{center}
  \caption{  \label{fig:220413_Fig-4b_TE} }
\end{figure}

\begin{figure}[h!]
\begin{center}
  \includegraphics[width=0.9\textwidth,]{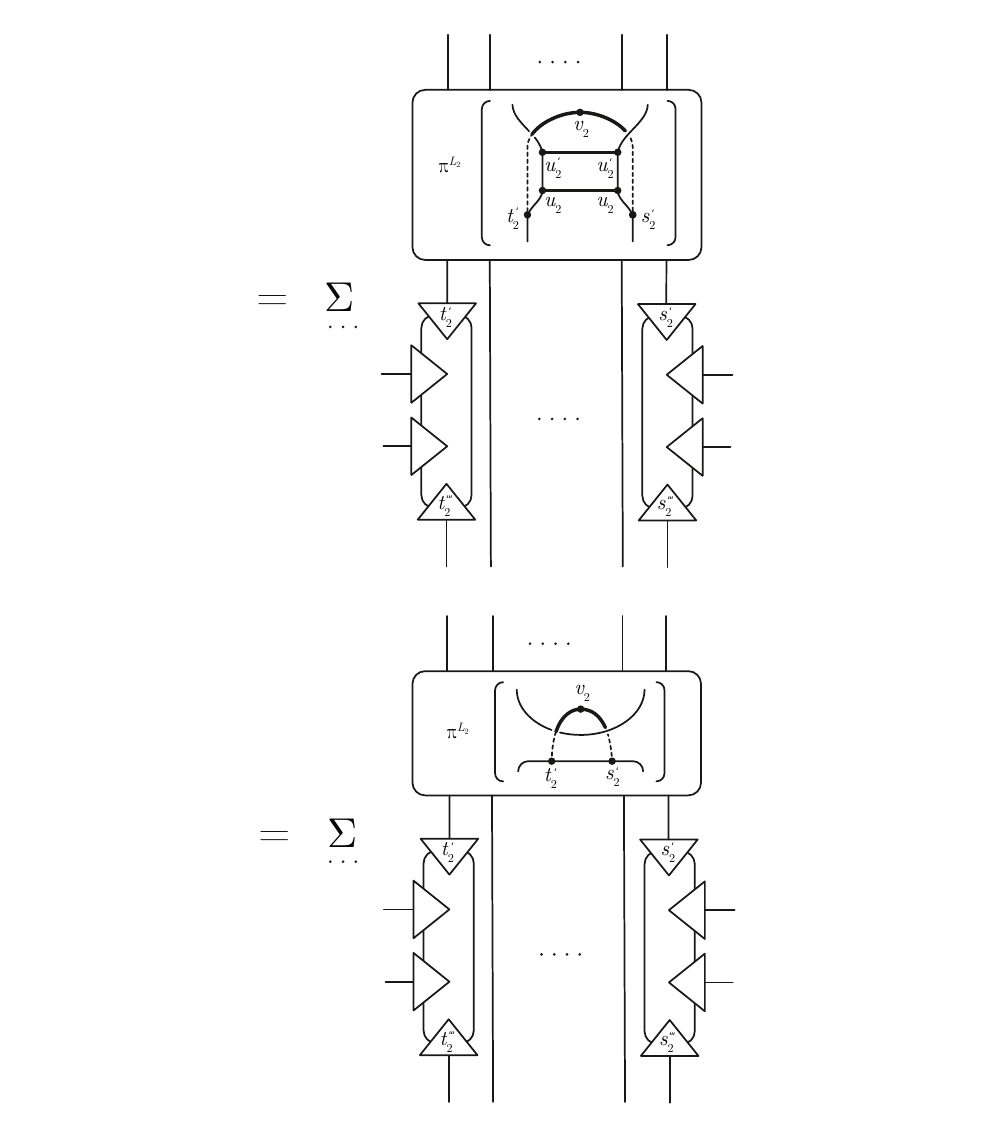}
  \end{center}
  \caption{  \label{fig:220413_Fig-4c_TE}}
\end{figure}

In fig. \ref{fig:220413_Fig-4a_TE}, we first 
introduce summations over $t_2, t_2', s_2, s_2'$ to go to the second panel. Then we insert the 
expression of $P^{L_2-2}$ as an MPO, cor. 1, item 6), and perform simple manipulations with $6j$-symbols very similar to the 
proof of thm. \ref{thm1} to go to fig. \ref{fig:220413_Fig-4b_TE}. To 
arrive at fig. \ref{fig:220413_Fig-4c_TE}, we use the representation property of $\pi^{L_2}$, thm. \ref{thm1}. To go to the second panel, we have written out the representation $\pi^{L_2}$ and performed simple manipulations 
with $6j$-symbols very similar to the 
proof of thm. \ref{thm1}. Using the representation property of $\pi^{L_2}$, this allows us after the steps shown in fig. \ref{fig:220413_Fig-4c_TE} 
to eliminate the projectors 
$P^{L_i-2}$ in fig. \ref{fig:220413_Fig-10_TE} in terms of identity operators, and the figure now starts to look like structurally similar to the diagram 
for a single $\Psi$-operator, as in fig. \ref{fig:220413_Fig-6_TE}.

With this replacement understood in fig. \ref{fig:220413_Fig-6_TE},
we now zoom in onto one of the parts involving the intertwiners $e_i, f_i, e_i', f_i'$, wherein 
\bena
\label{eifidef1}
e_1' &\in \Hom(\mu_1' \nu', \mu_1'')\\
f_1' &\in \Hom(\rho' \lambda_1', \lambda_1'')\\
e_2' &\in \Hom(\mu_2' \rho', \mu_2'')\\
f_2' &\in \Hom(\nu' \lambda_2', \lambda_2'').
\eena
This is shown in figs. \ref{fig:220413_Fig-11a_TE}, \ref{fig:220413_Fig-11b_TE}. 

\begin{figure}[h!]
\begin{center}
  \includegraphics[width=0.9\textwidth,]{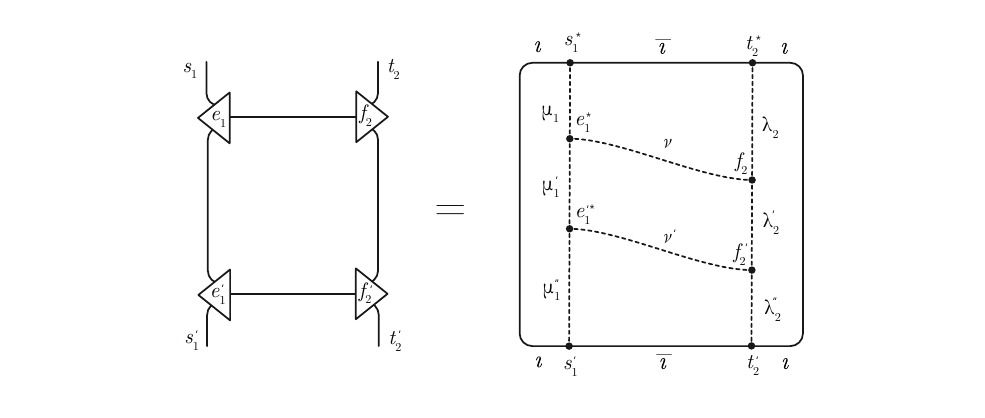}
  \end{center}
  \caption{  \label{fig:220413_Fig-11a_TE}}
\end{figure}

\begin{figure}[h!]
\begin{center}
  \includegraphics[width=0.9\textwidth,]{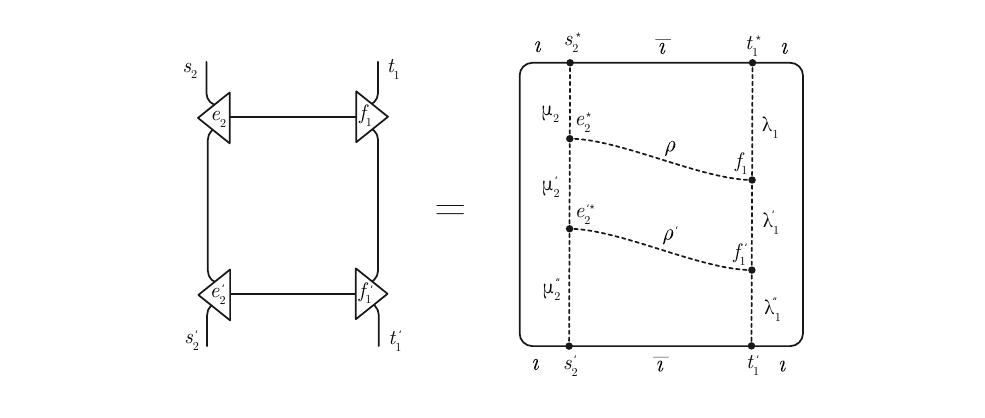}
  \end{center}
  \caption{  \label{fig:220413_Fig-11b_TE}}
\end{figure}

The summations over 
ONBs of appropriate intertwiners implicit in the non-open wires 
yield the right panels using the properties of the $6j$-symbols. 
Next, we zoom in onto the $\zeta$-factors (see 
fig. \ref{fig:220413_Fig-5_TE}) in fig. \ref{fig:220413_Fig-10_TE},
which are displayed as wire-diagrams in fig. \ref{fig:220413_Fig-12_TE}. 

\begin{figure}[h!]
\begin{center}
  \includegraphics[width=0.9\textwidth,]{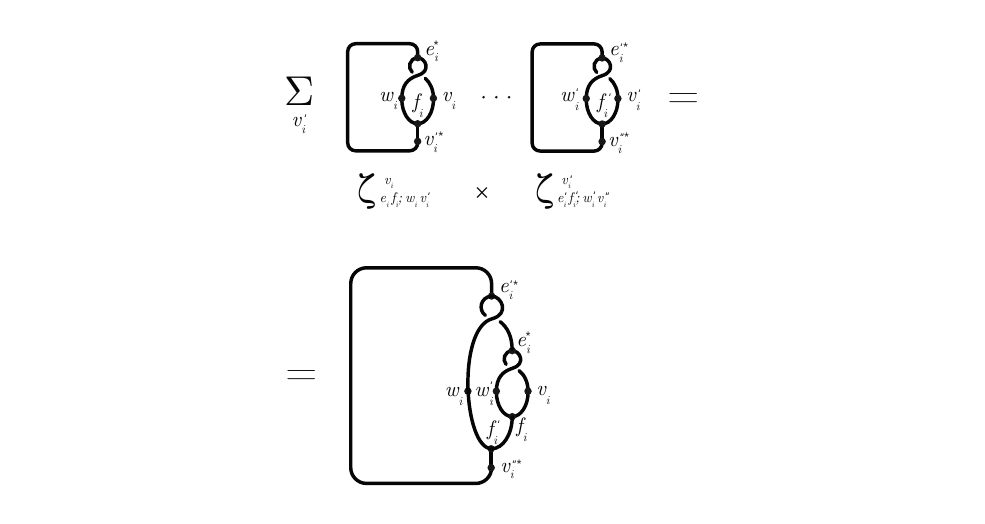}
  \end{center}
  \caption{  \label{fig:220413_Fig-12_TE}}
\end{figure}

We can carry out the summation over $v_i'$ (for a similar argument, 
see the proof of thm. 1.4 in \cite{rehren2000canonical}), leading to the second panel. 
We next observe that both in fig. \ref{fig:220413_Fig-11a_TE} and fig. \ref{fig:220413_Fig-11b_TE}, fig. \ref{fig:220413_Fig-1_TE} as well as in 
fig. \ref{fig:220413_Fig-12_TE}, we have summations over ONBs $(1_{\rho'} \times f_1)f_1'$, $(1_{\nu'} \times f_2)f_2'$, $e_1^{\prime*}(e_1^{*} \times 1_{\nu'})$, $e_2^{\prime *}(e_2^{*} \times 1_{\rho'})$ respectively their adjoints. By unitarity of the $6j$-symbols we can switch to ONBs of the form 
$(f_1 \times 1_{\lambda_1'})f_1'$, $(f_2 \times 1_{\lambda_2'})f_2'$, $e_1^{\prime *}(1_{\mu_1'}\times e_1^{*})$, 
$e_2^{\prime *}(1_{\mu_2'}\times e_2^{*})$. The effect of this change 
on fig. \ref{fig:220413_Fig-12_TE} is shown in fig. \ref{fig:220413_Fig-13_TE}, using also a sequence of braiding-fusion moves.

\begin{figure}[h!]
\begin{center}
  \includegraphics[width=0.9\textwidth,]{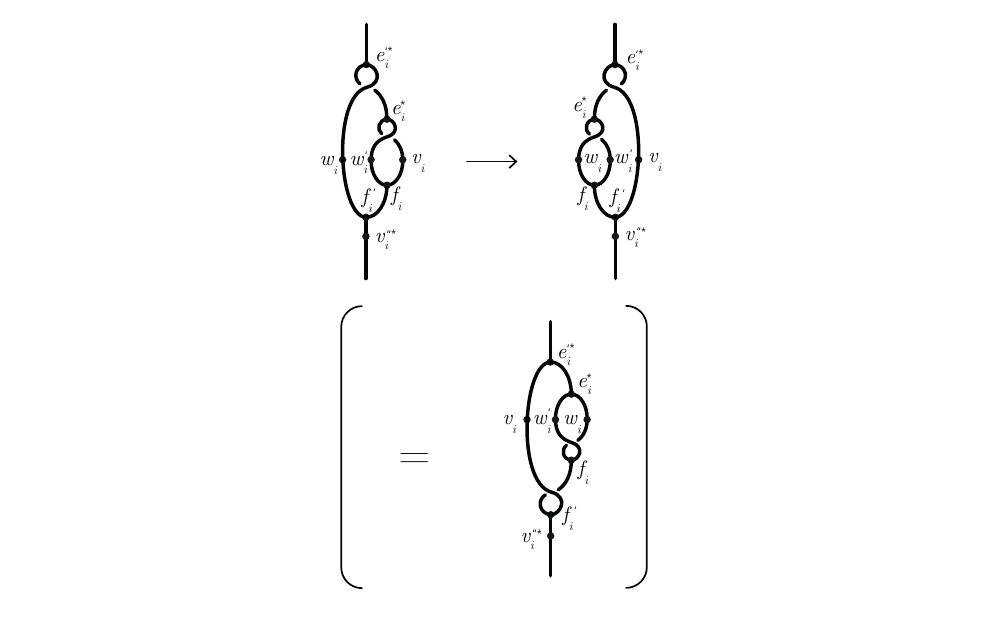}
  \end{center}
  \caption{  \label{fig:220413_Fig-13_TE}}
\end{figure}

Likewise, the effect of 
these changes on fig. \ref{fig:220413_Fig-11a_TE} and fig. \ref{fig:220413_Fig-11b_TE} is shown in the second 
panel of fig. \ref{fig:220413_Fig-14a_TE} which also includes the 
result of fig. \ref{fig:220413_Fig-13_TE} (with a similar relation for $1 \leftrightarrow 2$).

\begin{figure}[h!]
\begin{center}
  \includegraphics[width=0.9\textwidth,]{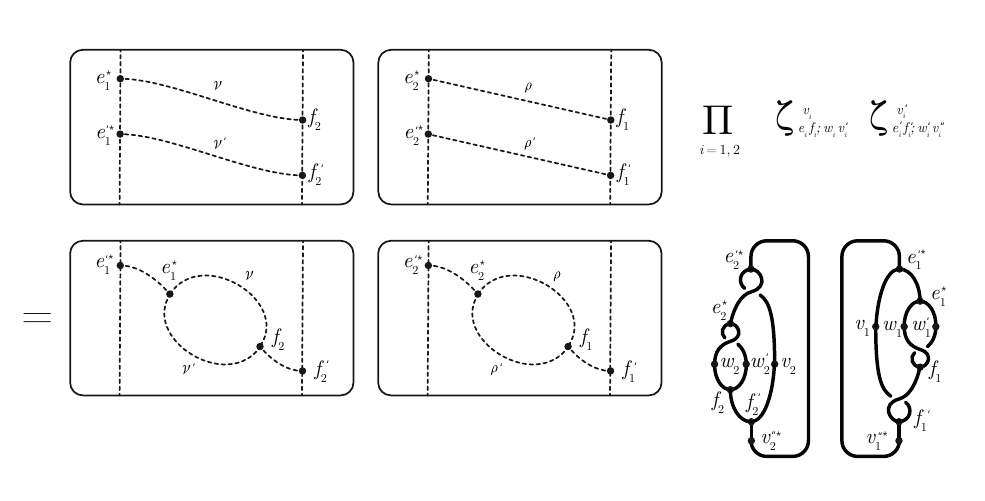}
  \end{center}
  \caption{  \label{fig:220413_Fig-14a_TE}}
\end{figure}

 In the second panel in fig. \ref{fig:220413_Fig-14a_TE}, 
we obtain from the dashed bubbles on the left side Kronecker delta's
$\delta_{e_1,f_2} \delta_{e_2,f_1}$. Using these delta's and taking the 
rightmost wire diagrams apart via a new summation over $w_1'',w_2''$
yields the second panel in fig. \ref{fig:220413_Fig-14b_TE}. 

\begin{figure}[h!]
\begin{center}
  \includegraphics[width=0.9\textwidth,]{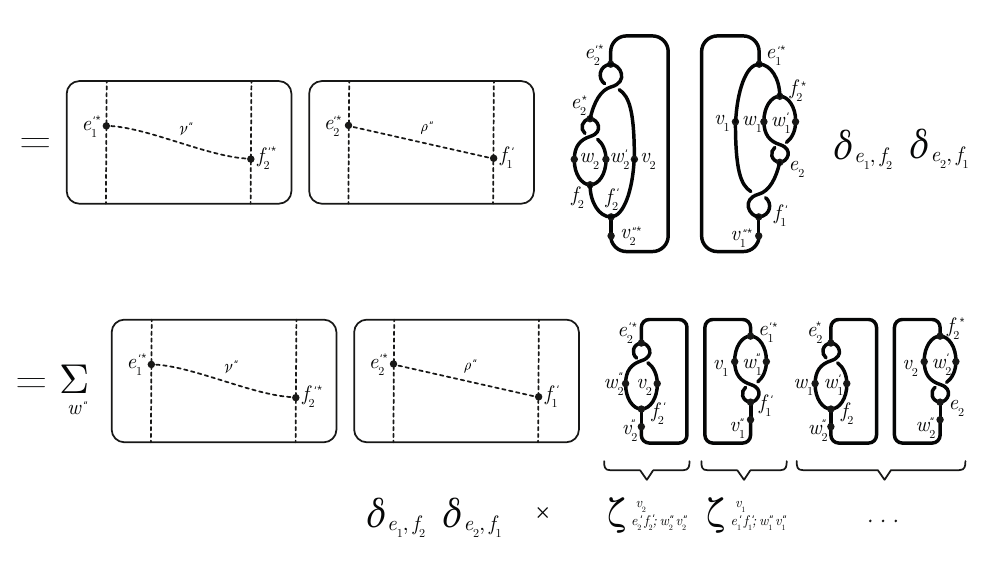}
  \end{center}
  \caption{  \label{fig:220413_Fig-14b_TE}}
\end{figure}

We now reverse the steps fig. \ref{fig:220413_Fig-11a_TE} and
fig. \ref{fig:220413_Fig-11b_TE} on the two leftmost
square diagrams in the lower panel of fig. \ref{fig:220413_Fig-14b_TE}.
Finally, we insert this into fig. \ref{fig:220413_Fig-10_TE}, remembering that we are allowed to eliminate the projectors 
$P^{L_i-2}$ in fig. \ref{fig:220413_Fig-10_TE} in terms of identity operators.
This yields the diagram for the operator sum on the right side of \eqref{OPE} with the structure constants 
\ben
    \label{cdef1}
    c^{\nu'',\rho'',w_1'',w_2''}_{\nu',\rho',w_1',w_2'; \nu,\rho;w_1,w_2}
    := \sum_{e,f} \zeta^{w_1''}_{e,f;w_1, w_1'}
    \zeta^{w_2''}_{f^*,e^*;w_2, w_2'}.
    \een
But the $\zeta$-structure constants \eqref{Rehrenz} differ from the $\eta$-structure constants \eqref{rehren} only by a braiding operator which is unitary and hence gives a unitary 
change of ONB in the respective intertwiner spaces. So we can replace the $\zeta$-structure constants with the $\eta$-structure
constants.

    \medskip 
5) Consider the operator 
\ben
\label{op}
(Q_{\mu_1, \lambda_1}^{L_1} 
\otimes Q_{\mu_2, \lambda_2}^{L_2})
\Psi^{L_1,L_2}_{\nu,\rho; w_1, w_2}
(Q_{\mu_1', \lambda_1'}^{L_1} 
\otimes Q_{\mu_2', \lambda_2'}^{L_2}).
\een
We recall the representation $Q^{L}_{\mu,\nu} = \pi^L(q_{\mu,\nu})$ and the graphical expression for $q_{\mu,\nu}$, in the double 
triangle algebra, fig. \ref{fig:220413_Fig-19_TE}. We insert this into
fig. \ref{fig:220413_Fig-7_TE}, use the representation property of $\pi^{L_i}$, the relations in the double triangle algebra $\lozenge$,
and \cite{bockenhauer1999alpha}, lem. 6.2. Then we see that in the sum over intertwiners $e_i, f_i$ as in \eqref{eifidef} only such terms survive such that the values of 
$\mu_i, \lambda_i, \mu_i', \lambda_i'$ are equal to 
the values prescribed by the projections $Q$ in \eqref{op}. This 
is equivalent to the claimed fusion property in 5).
\end{proof}

\subsection{Classification of defects}

According to thm. \ref{thm:3}, 4) the defect algebra $\cD^{L_1,L_2}$
associated with a bi-partite closed spin chain with subsystems of lengths $L_1, L_2$ is abelian and its structure constants are independent 
of $L_1, L_2$. This means that the defect algebra is a representation of a 
universal object. In fact, it is precisely isomorphic to the algebra classifying transparent boundaries in 1+1 dimensional CFTs 
\cite{bischoff2015tensor,bischoff2016phase}, see sec. \ref{sec:transparent}:

\begin{theorem}
The defect algebra $\cD^{L_1,L_2}$
associated with a bi-partite spin chain with subsystems of lengths $L_1, L_2$ is a representation of $(\cB^+)' \cap \cB^+$, where 
$\cB^+$ is the extension of $\cN \otimes \cN^{\rm opp}$ defined by 
the braided product Q-system $Z_1[X] \times^+ Z_2[X]$ with $Z_i[X]=(\theta_{{\rm R} i}, x_{{\rm R}i}, w_{{\rm R}i}), i=1,2$ two copies of the full center Q-system of $_\cN X_\cN$. 
\end{theorem}

\begin{proof}
This statement is implicit in the proof of \cite{bischoff2016phase}, prop. 4.19. 
First we recall the generators of $(\cB^+)' \cap \cB^+$. By the general relation between Q-systems and extensions, $\cB^+$ is pointwise equal to $NV^+$, with $N \in \cN \otimes \cN^{\rm opp}$ and with $V^\pm$ obeying the relations \eqref{qsystem1} for $X^+$ as in \eqref{qsystem2} for our two isomorphic copies of the full center Q-system. To be precise, 
denoting $l=(\lambda,\mu,v), \bar l = (\bar \lambda, \bar \mu, \bar v)$ and by $T_l \in \cN \otimes \cN^{\rm opp}$ a set of isometries obeying the Cuntz algebra 
relations (with $\mu,\lambda \in{} _{\cN} X_\cN, v \in \Hom(\alpha^+_\lambda, \alpha_\mu^-)$), 
we may set $\theta_{{\rm R}1} = \sum_l T_l (\lambda \otimes \mu ^{\rm opp}) T_l^*$ and $\theta_{{\rm R}2} = \sum_l T_{\bar l} (\bar \lambda \otimes \bar \mu ^{\rm opp}) T_{\bar l}^*$. 
By  \cite{bischoff2016phase}, lem. 4.5, the center $(\cB^+)' \cap \cB^+$ is generated as a vector space by $P^+V^+$, where
$P^+ \in \Hom(id, \theta_{{\rm R}1} \theta_{{\rm R}2}) \cap \cN \otimes \cN^{\rm opp}$. A basis of such $P^+$ is given by 
\ben
P^+_{\lambda, \mu; v_1, v_2} := (\bar r_\lambda^* \otimes (\bar r_\mu^{\rm opp})^*) T_{\mu, \lambda, v_1}^* \theta_1(T_{\bar \mu, \bar \lambda, \bar v_2}^*),
\een
with $\mu,\lambda \in {}_{\cN} X_\cN$ and  $v_i$ an ONB of $\Hom(\alpha^+_\lambda, \alpha_\mu^-)$.
So, a set of generators of the center $(\cB^+)' \cap \cB^+$ is
\ben
B_{\lambda, \mu, v_1, v_2} = P_{\lambda, \mu; v_1,v_2} V^+.
\een
It is now straightforward
to use the Q-system calculus, \eqref{qsystem1}
and \eqref{qsystem2}, and the definition of the full 
center Q-system $Z[X]$. Then it is found that 
\ben
\label{BOPE}
B_{\lambda, \mu, v_1, v_2} B_{\lambda', \mu', v_1', v_2'}
= \sum_{\lambda'',\mu'',v_1'',v_2''} d^{\lambda'',\mu'',v_1'',v_2''}_{\lambda',\mu',v_1',v_2'; \lambda,\mu;v_1,v_2} 
B_{\lambda'', \mu'', v_1'', v_2''}
\een
for the structure constants $d$ given by 
\bena
&d^{\lambda'',\mu'',v_1'',v_2''}_{\lambda',\mu',v_1',v_2'; \lambda,\mu;v_1,v_2} = \sum_{e_1,e_2,f_1,f_2} \eta^{v_1''}_{e_1,f_1;v_1 v_1'} \eta^{\bar v_2''}_{e_2,f_2;\bar v_2 \bar v_2'}
[d(\mu'') d(\lambda'')]^{-1} \\
& [(1_\mu \times \epsilon^+(\mu',\bar\mu) \times 1_{\bar \mu'})
(f_1 \times f_2)r_{\mu''}]^* (1_\lambda \times \epsilon^-(\lambda',\bar\lambda) \times 1_{\bar \lambda'})
(e_1 \times e_2)r_{\lambda''} 
 \non
\eena
 where 
$e_1, e_2, f_1, f_2$ run over ONBs of $\Hom(\lambda\lambda',\lambda''),
\Hom(\bar \lambda \bar \lambda',\bar \lambda''),
\Hom(\mu\mu',\mu''), 
\Hom(\bar \mu \bar \mu', \bar \mu'')$, respectively. The second line is 
$d(\mu'') d(\lambda'')$ if $e_2=\epsilon^-(\bar \lambda,\bar \lambda')\bar e_1$
and $f_2=\epsilon^+(\bar \mu, \bar \mu')\bar f_1$ and zero otherwise, by a set of  braiding-fusion moves and by using the definition of the conjugate intertwiners. 
Thus, the structure constants are (writing out the $\eta$-structure 
constants)
{\footnotesize
\ben
=\sum_{e,f}
\left[ \frac{d(\mu) d(\mu'') d(\lambda) d(\lambda'')}{d(\mu') d(\lambda') d(\theta)^2} \right]^{1/2}\,
E[\bar f^*\epsilon^-(\bar \mu, \bar \mu')^*(\bar v_2 \times \bar v_2') \epsilon^+(\bar \lambda,\bar \lambda') \bar e (\bar v_2'')^*]
E[f^*(v_1 \times v_1') e (v_1'')^*].
\een
}
By writing out the definition of the conjugate endomorphisms, using the functoriality condition 2) of alpha induction, and using the 
conjugacy relations, we arrive at, 
\ben
E[\bar f^*\epsilon^-(\bar \mu, \bar \mu')^*(\bar v_2 \times \bar v_2') \epsilon^+(\bar \lambda,\bar \lambda') \bar e (\bar v_2'')^*]
=
E[e^*((v_2)^* \times (v_2')^*) f v_2''].
\een
A comparison of \eqref{BOPE} with \eqref{OPE} then shows that the map
\ben
\label{Brep}
B_{\lambda,\mu,v_1,v_2} \mapsto \Psi^{L_1,L_2}_{\lambda,\mu, v_1, v_2^*}
\een
is an algebra representation. 
\end{proof}

By \cite{bischoff2015tensor}, thm. 4.44 we have a description of 
$(\cB^+)' \cap \cB^+$ in terms of the objects $A \in \, _\cM X_\cM$;
more precisely for each $A$ one can construct in a canonical way 
a minimal projection $p_A \in (\cB^+)' \cap \cB^+$. The 
element $D_A^{L_1,L_2} \in \cD^{L_1,L_2}$ corresponding to $p_A$ under the above representation \eqref{Brep}
is a projection operator acting on $\sV^{L_1}_{\rm open} \otimes \sV^{L_2}_{\rm open}$ which by definition must be of the form 
\bena
\label{dadef}
D_A^{L_1,L_2} = \sum_{\nu,\rho,w_1,w_2} \xi_{A;\nu,\rho,w_1,w_2}^{}
\Psi_{\nu,\rho;w_1,w_2}^{L_1,L_2}
\eena
for certain complex numbers $\xi_{A;\nu,\rho,w_1,w_2}$. These complex 
numbers are in principle determined by \cite{bischoff2015tensor}, thm. 4.44 and are universal, i.e. 
independent of $L_1,L_2$ or the dynamics on the spin-chain. In general, 
it seems nontrivial to determine them from data that are directly accessible, 
so we will leave this problem to a future work. If we have the 
projections $D_A^{L_1,L_2}$, we can write
\bena
\sV^{L_1}_{\rm open} \otimes \sV^{L_2}_{\rm open} = 
\bigoplus_{A \in \, _\cM X_\cM} \sV^{L_1,L_2}_A, 
\quad 
\sV^{L_1,L_2}_A := D_A^{L_1,L_2} (\sV^{L_1}_{\rm open} \otimes \sV^{L_2}_{\rm open}). 
\eena
The subspaces $\sV^{L_1,L_2}_A$ determine a specific boundary condition $A$ fusing the two chains $L_1, L_2$ together. By thm. \ref{thm:3}, 3) 
these subspaces are invariant under the action of any local operator which does not act on the endpoint of the two chains.

\subsection{Special cases: abelian fusion in $\cM$}
\label{sec:specialcases}

There is an interesting special case, still covering many examples, where 
the determination of the coefficients $\xi_{A;\nu,\rho,w_1,w_2}$ in \eqref{dadef} is possible, 
namely when all $Z_{\nu,\rho} \in \{0,1\}$. Then, since the horizontal center $\cZ_h$ is isomorphic to the 
fusion ring $\, _\cM X_\cM$, and since the direct summands of $\cZ_h$
are matrix algebras of size $Z_{\nu,\rho}$, it follows that 
$\, _\cM X_\cM$ is abelian, i.e. $N_{A,B}^C = N_{B,A}^C$ in particular ($\, _\cN X_\cN$ by assumption is always abelian because we assume that it is braided).
In such a situation, the pairs $(\nu,\rho)$ such that $Z_{\nu, \rho}=1$ are one-to-one correspondence with the simple objects $A \in \, _\cM X_\cM$. 
In fact, since we have $Z_{\nu,\rho} = \sum_B \langle \alpha^+_\nu,B\rangle \langle \alpha^-_\rho,B\rangle$, 
with  $\langle \alpha^\pm_\lambda,B\rangle = \dim \Hom(\alpha^\pm_\lambda,B)$, we may assign 
to $(\nu,\rho)$ that unique $B$ corresponding to the only summand which is $=1$.

In particular, we may diagonalize the 
fusion coefficients $N_{A,B}^C$. Furthermore, since $Z_{\nu,\rho} = \dim \Hom(\alpha^-_\rho, \alpha^+_\nu)  \in \{0,1\}$, 
we have no need for the degeneracy indices $w_1,w_2$ so that we may write
write $\Psi_{\nu,\rho}^{L_1,L_2}$ for the MPOs generating the defect algebra.
By \cite{bischoff2016phase}, sec. 5.4, the coefficients \eqref{dadef} are proportional to the 
inverse of the matrices diagonalizing $N_{A,B}^C$. By lem. \ref{Ylem}, 1) and 3), the latter 
are in turn given by the coefficients $Y_{\lambda, \mu, A}$ given in fig. \ref{fig:Y} . Combining these results gives 
\ben
D_A^{L_1,L_2} = \frac{d_A}{ D_X^{2} d^4} \sum_{\nu,\rho} d_\nu d_\rho \, Y_{\nu,\rho,A}
\Psi_{\nu,\rho}^{L_1,L_2}.
\een
At any rate, this makes the determination of the $D_A^{L_1,L_2}$ a 
feasible problem, because $N_{A,B}^C$ is known from the start. 

\medskip

In the case of a diagonal fusion category ${}_\cN X_\cN$, 
\ben
Z_{\nu,\rho} = \delta_{\nu,\rho}, 
\een
we have $|{}_\cN X_\cN|=|{}_\cM X_\cM|$ and 
$\delta_{\nu,\rho}= \sum_B \langle \alpha^+_\nu,B\rangle \langle \alpha^-_\rho,B\rangle$ meaning that 
$(\langle \alpha^\pm_\lambda,A\rangle)_{\lambda \in {}_\cN X_\cN, A \in {}_\cM X_\cM}$ are permutation matrices 
which are inverses of each other, setting up a bijection between the simple objects in 
${}_\cN X_\cN$ and ${}_\cM X_\cM$. Due to the homomorphism 
property of $\alpha$-induction, either one of these isomorphisms transforms the $N_{\mu,\nu}^\sigma$ fusion 
coefficients to the $N_{A,B}^C$ fusion coefficients. 

The defect operators $\Psi_{\nu,\nu}^{L_1,L_2}$ 
may be labelled by $A \in {}_\cM X_\cM$ under this isomorphism, so 
we may write $\Psi_{A}^{L_1,L_2}$. As a special case of 
\eqref{OPE}, their OPE is (after a suitable renormalization of the generators)
\ben
\Psi_{A}^{L_1,L_2} \Psi_{B}^{L_1,L_2} = \sum_C N_{A,B}^C \Psi_{C}^{L_1,L_2}, 
\een
so is an isomorphic copy of the abelian fusion product in ${}_\cM X_\cM$. Thus, in this special case, 
the defect algebra $\cD^{L_1,L_2}$ is an isomorphic copy of this fusion ring. By contrast to the MPOs
$O_A^L$, which also yield an algebra isomorphic to the fusion ring of ${}_\cM X_\cM$ by \eqref{Ofusion}, thm. \ref{thm1}, the MPOs 
$\Psi_{A}^{L_1,L_2}$ act on a bipartite chain. Therefore, by contrast to the former, we think of them as corresponding to 
{\it vertical}, instead of {\it horizontal}, defects.

Note that for non-diagonal theories, the fusion ring of ${}_\cM X_\cM$ can 
be non-abelian, while $\cD^{L_1,L_2}$ is always abelian. Thus, in general, the horizontal defect algebra (the $O_A^L$s)
is not isomorphic to our vertical defect algebra (the $\Psi_{\nu,\rho;w_1,w_2}^{L_1,L_2}$). 

\section{Conclusions}

In this paper we have explored some connections between subfactor theory, braided unitary fusion categories, CFTs in $1+1$ dimensions, 
defects, and anyonic spin chains. Future work should address the following points:

\begin{itemize}
\item We have not addressed at all the analytic question of scaling limits of anyonic spin chains. To what extent will the 
close analogy between the defect algebra in on the spin-chain and of the CFTs -- observed here at a purely algebraic level
for finite chains -- persist in such a scaling limit?

\item Our defect operators on the spin-chains have, after all, a fairly complicated structure. What simplifications can be obtained in 
simple cases such as the anyonic chain built upon the Ising fusion category?

\item If one already knows that the fusion categories used as an input into the anyonic chain construction came from a CFT, such as e.g. 
arising from the positive energy highest weight representations of a Virasoro algebra in the case the minimal models: Can the vacuum state of the 
CFT be used to obtain an approximation of the ground state for the anyonic chain?
\end{itemize}

\vspace{1cm}

{\bf Acknowledgements:} I thank Henning Rehren and Yasu Kawahigashi for their educational emails related to subfactors and CFTs, and 
Laurens Lootens, Alex Stottmeister and Frank Verstraete for explanations regarding anyonic spin chains. 
I am grateful to the Max-Planck Society for supporting the collaboration between MPI-MiS and Leipzig U., grant Proj.~Bez.\ M.FE.A.MATN0003, and to 
Thomas Endler from MPI-MiS, Leipzig, for help with figures. 

\appendix

 \section{Inclusions associated with finite groups}\label{appB}
 
Let $\cN \subset \cM$ be an irreducible inclusion of factors such that $\gamma(\cM)' \cap \cM = \cN_1' \cap \cM$ is abelian  
with dimension $n=[\cM:\cN]$. \cite{kosaki1989characterization} has shown that 
$\cN = \cM^\alpha$, where $g \mapsto \alpha_g$ is an action of a finite group $G$ of order $n$ 
by outer automorphisms acting on $\cM$. This situation is in principle very well-understood, see e.g. \cite{longo1994duality} who uses
the Cuntz-algebra picture. Here we work out the situation totally explicitily from the Q-system perspective as an illustration of this 
concept for the reader without a background in subfactor theory.
 
A. By the general theory since $\cM=\cN v$ point-wise there must exist a unique $u_g \in \cN$ such that
\ben
\alpha_g(v) = u_g v, \quad u_g \in \cN.
\een
Using $\alpha_g\alpha_h(v)=\alpha_{gh}(v)$ and $\alpha_g(v)^*=\alpha_g(v^*)$ or $\alpha_g(v^2)=[\alpha_g(v)]^2$ respectively, we find
\ben
\quad u_g u_h = u_{gh}, \quad u_g^* r = \theta(u_g) r, \quad xu_g\theta(u_g)= u_gx, 
\een
with $r = xw \in \Hom(\theta^2,id)$.
Using $\alpha_g(n)=n, n \in \cN$, we also find $[u_g, \theta(\cN)]=0$ so $u_g \in \Hom(\theta, \theta) \cap \cN$. 
Also, $\theta = j_\cN j_\cM |_{\cN} =  j_{\cN_1} j_\cN |_{\cN}$, so $\theta(\cN)= \cN_2$ and hence $u_g \in \cN_2' \cap \cN$. 

B. We can rewrite the decomposition $\theta \cong \oplus_i {\rho}_i$ where a given sector may appear multiple times, 
as $\theta \cong \oplus_\pi d_\pi \cdot {\rho}_\pi$, where $d_\pi$ is the multiplicity of the irreducible sector ${\rho}_\pi$. The 
corresponding intertwiners are $w^a_\pi, a=1, \dots, d_\pi$, and 
\ben
\theta(n)=\sum_\pi \sum_{a=1}^{d_\pi} w^a_\pi {\overline{\rho}}_\pi^{}(n) w^{a*}_\pi. 
\een
It follows from the intertwiner calculus that 
\ben
D_\pi(g)_a^b = w^{a *}_\pi u_g w^{b}_\pi \in \cN \cap \Hom(\bar \rho_\pi, \bar \rho_\pi) = \CC 1, 
\een 
can be identified with complex numbers, and the $d_\pi \times d_\pi$ matrices $D_\pi(g)_a^b$ form an irreducible 
representation of $G$ which transforms the fields $\psi^a_\pi := w_\pi^{a*} v$ as in 
\ben
\alpha_g(\psi^b_\pi) = \sum_{a=1}^{d_\pi} D_\pi(g)_a^b \psi^a_\pi. 
\een
One can equip the linear space $\sK_\pi = {\rm span}_\CC\{ \psi^a_\pi : a = 1, \dots, d_\pi\}$ with a scalar product
$(\phi_1, \phi_2)_\pi 1 = \phi_1^* \phi_2^{}$ because the latter operator is an intertwiner in $\Hom(id,id)$ which is at the same time in 
$\cM$ hence in $\cN' \cap \cM = \CC 1$. Then it follows that the numerical matrices $D_\pi(g)_a^b$ give a unitary operator $D_\pi(g)$ on 
$\sK_\pi$ which implements $\alpha_g(\phi)=D_\pi(g)\phi$ (unitarity because $\phi_1^* \phi_2^{} = \alpha_g(\phi_1)^* \alpha_g(\phi_2)$ from 
$\alpha_g(1)=1$.) Thus, $u_g$ itself 
\ben
\label{ugrep}
u_g = \sum_\pi \sum_{a,b=1}^{d_\pi} D_\pi(g)_a^b w^{a}_\pi w^{b*}_\pi
\een
is unitary. Thus $\pi$ corresponds to the unitary irreducible representations of $G$, with characters
\ben
\chi_\pi(g)= \sum_{a=1}^{d_\pi} D_\pi(g)_a^a.
\een 

C. Now let 
\ben
u_g' = j_\cN(u_g) \in \cN'. 
\een
The unitaries $u_g' \in \cN'$ also form
a representation of $G$. 
For the minimal expectation $E': \cN' \to \cM'$ we calculate
\ben
E'(u_g') = \frac{1}{d} w^{\prime *} \gamma'(u_g') w' = \frac{1}{d} j_\cM(v^* u_g v)
\een
using $w'=j_\cM(v)$. Since $u_g \in \Hom(\theta,\theta), v \in \Hom(\theta, id)$, we must have $v^* u_g v \in \cN' \cap \cM=\CC1$. 
To compute this number, we apply again $E$, since $E(v^* u_g v)=v^* u_g v$,
\ben
\begin{split}
E(v^* u_g v) =& \sum_\pi  \sum_{a,b=1}^{d_\pi} D_\pi(g)_a^b E(v^* w^{a}_\pi w^{b*}_\pi v)\\
=&  \sum_\pi  \sum_{a,b=1}^{d_\pi} D_\pi(g)_a^b \frac{\dim({ \bar \rho}_\pi)}{d} \delta^a_b \\
=&  \sum_\pi  \chi_\pi(g)  \frac{d_\pi}{d} = d \cdot \delta_{1,g}.
\end{split}
\een
Thus, $E'(u_g')=\delta_{1,g}$. 

D. Let $e \equiv e_{\cM'} = [\cM' \eta] =d^{-1} \cdot vv^* \in \cM \cap \cN_1'$ be the Jones projection for the extension $\cN_1 \subset \cN$. (Since 
$j_\cN(e_{\cN_1})=e_{\cN_1}$ and $j_\cN(\cN_1)=\cM'$ we have $e_{\cN_1}=e_{\cM'}$.)
Then $E(e) = d^{-2} 1$. On the other hand
\ben
\label{prev}
\begin{split}
1 =& d^2 \cdot E(e) = d^2 \cdot j_\cN(E(e)) =  d \cdot j_\cN(E(vv^*))\\
=&\frac{d}{|G|} j_\cN \left( \sum_g \alpha_g(vv^*) \right) \\
=& \frac{d}{|G|} \sum_g j_\cN(u_g vv^* u_g^*)\\ 
=&  \frac{d^2}{|G|} \sum_g u_g' j_{\cN}(e) u_g^{\prime *}\\
=& \sum_g u_g' e u_g^{\prime *}
\end{split}
\een
since $|G|=d^2$
On the other hand, for any $n' \in \cN$, 
\ben
n' e = 1 n' e =  \sum_g u_g' e u_g^{\prime*} n' e =  \sum_g u_g' E'(u_g^{\prime*} n') e
\een
in view of $E'(n')e=en'e$, and so 
\ben
\label{nrep}
n' = \sum_g u_g' m_g' , \quad m'_g= E'(u_g^{\prime*} n') \in \cM'. 
\een 
By $E'(u_g')=\delta_{1,g}$ this representation $n' = \sum_g  u_g' m_g'$ is unique. 

E. By construction $\alpha_g(e) \in \cM$ and since $\alpha_g(\cN_1)=\cN_1$ pointwise, also $\alpha_g(e) \in \cN_1'$. 
Now $j_\cN$ fixes $\cM \cap \cN_1'$, so $j_\cN(\alpha_g(e)) \in \cM \cap \cN_1$. Since $e=d^{-1} \cdot vv^*$ and 
since we have $\alpha_g(v)=u_gv, u_g'=j_\cN(u_g)$, and $j_\cN(e)=e$, we have $u_g' e u_g^{\prime *} \equiv e_g \in \cM \cap \cN_1'$.  
By the Jones basic construction, $\cN_1' = \cN' \vee \{e\}, \cM = \cN \vee \{e\}$, so in particular $u_g' \cM u_g^{\prime *} = \cM$
and hence also $u_g' \cM' u_g^{\prime *} = \cM'$.
In conclusion, we have $\cN'=\cM' \vee \{ u'_g \}$, and the properties of the $u_g'$ and of $E'$ in D show that this 
is crossed product of $\cM'$ by an action of $G$. We may thus say that
\ben
\cN = \cM^\alpha, \quad \cN' = \cM' \rtimes_{\alpha'} G, 
\een
where $\alpha'_g(m')=u_g' m' u_g^{\prime *}$ is an action of $G$ on $\cM'$.

F. By \eqref{prev} and E, $1=\sum_g e_g$. Recall that $\cM \cap \cN_1'$ is abelian of dimension $|G|$. 
Thus, the $e_g$ are commuting projections. They are either pairwise distinct or $e_g = e$ for all $g$ which is not possible so 
the $e_g$ are the minimal projections for $\cM \cap \cN_1'$. 

G. $\alpha'$ is also an action on $\cM$, and by $u_g'\in \cN'$, we have $\cN \subset \cM^{\alpha'} \subset \cM$. 
However since $\cN' \cap \cM$, this action is outer and hence $[\cN:\cM]=|G|=[\cM^{\alpha'} : \cM]$, thus 
$\cM^{\alpha'}=\cN$. Then, we may repeat the same procedures A--F with potentially new $\tilde u_g' \in \cN'$. 
By D, F $\tilde u_g' e \tilde u_g^{\prime *} = e_g$ and thus $u_g^{\prime *}\tilde u_g' \equiv w_g' \in \{e\}'$ so $w_g' \in \cM'$. 
Thus, the above statements remain true with obvious modifications if we replace $\alpha$ by $\alpha'$

We can summarize the situation as follows:

\begin{proposition}
Let $\cN \subset \cM$ be an irreducible inclusion of factors such that $\gamma(\cM)' \cap \cM$ is abelian  
with dimension $n=[\cM:\cN_1]$. Then there exist a finite group $G$ of order $n$ and $\{u_g':g\in G\} \subset \cN'$
such that
\begin{enumerate}
\item $g \mapsto u_g'$ is a unitary representation of $G$.

\item $\cN'=\cM' \vee \{u_g' : g \in G\}$ and each $n' \in \cN'$ is representable uniquely 
as $n'=\sum_g u_g' m_g'$ with $m_g' \in \cM'$.
\item $m \mapsto u_g' m u_g^{\prime *}$ is an outer action of $G$ transforming $\cM$ to itself. Its fixed points 
are exactly the elements from $\cN$.
\item $\cM = \cN \vee \{\psi^a_\pi : \pi \in \hat G, a=1,\dots,d_\pi\}$ and each $m \in \cM$
is uniquely representable as $m=\sum_{\pi,a} n^a_\pi \psi^a_\pi$ with $n^a_\pi \in \cN$.
$\psi^a_\pi$ are scaled isometries with $(\psi_\pi^a)^* \psi^a_\pi = (d_\pi/d)1$.
\item For any $g \in G$:
\ben
u'_g \psi^b_\pi u_g^{\prime *} = \sum_{a=1}^{d_\pi} D_\pi(g)_a^b \psi^a_\pi. 
\een
where $D_\pi$ is an irreducible representation of dimension $d_\pi$.
\item The canonical expectations $E:\cM \to \cN, E':\cN' \to \cM'$ are given respectively by $E'(u'_{g})=\delta_{1,g}1, E(\psi^a_\pi)=\sqrt{d}\ \delta_{\pi,0}1$.
\end{enumerate}
\end{proposition}

\section{PEPs and double triangle algebra}
\label{PEPs}
We represent an MPO $O_A^L$ as an elementary cell as in fig. \ref{fig:MPO2} for some fixed chosen, e.g. $L=8$. 
Considering a number of such 
a elementary building block, we build a network as indicated in fig. \ref{fig:PEP1} by concatenating the outer legs of each cell, where for our choice $L=8$ we have an octagonal lattice.Such a network
can be considered as a tensor with a very large number of indices given by the labels carried by the uncontracted inner lines of each cell.  
For each elementary cell, we have an 
index $A \in {}_\cM X_\cM$. Each index $A$ associated with a cell is now summed, weighted by the coefficient $d_{A}$, thereby effectively 
inserting a projector $P^L = Q_{id,id}^{L}$ as in fig. \ref{fig:MPO1} into each cell. Such a structure, proposed in \cite{bultinck2017anyons} is called a ``PEPS'' (projected entangled pair-) state. 
In principle, we could have taken the pair $(\mu,\nu)$ to be different for each cell, but we will 
see from the next ``pull-through'' lemma that in such a case the PEPS is zero. 
\begin{figure}[h!]
\centering
\begin{tikzpicture}[scale=.6]
%
\draw (-7,1) node[anchor=east]{$\dots$};
\draw (-7,-1) node[anchor=east]{$\dots$};
\draw (9,1) node[anchor=west]{$\dots$};
\draw (9,-1) node[anchor=west]{$\dots$};
\draw (-7,9) node[anchor=east]{$\dots$};
\draw (-7,7) node[anchor=east]{$\dots$};
\draw (9,9) node[anchor=west]{$\dots$};
\draw (9,7) node[anchor=west]{$\dots$};

\draw (-2,-4.5) node[anchor=north]{$\vdots$};
\draw (-4,-4.5) node[anchor=north]{$\vdots$};
\draw (6,-4.5) node[anchor=north]{$\vdots$};
\draw (4,-4.5) node[anchor=north]{$\vdots$};

\draw[thick]  (0.16,0) arc (0:359:3.16);

\draw[thick] (-1,1) -- (1,1);
\filldraw[color=black, fill=white, thick](0,1) circle (.6);
\draw (0,1) node[rotate=90]{$U_A$};

\draw[thick] (-1,-1) -- (1,-1);
\filldraw[color=black, fill=white, thick](0,-1) circle (.6);
\draw (0,-1) node[rotate=90]{$\bar U_A$};

\draw[thick] (-2,-4) -- (-2,-2);
\filldraw[color=black, fill=white, thick](-2,-3) circle (.6);
\draw (-2,-3) node[rotate=0]{$U_A$};

\draw[thick] (-4,-4) -- (-4,-2);
\filldraw[color=black, fill=white, thick](-4,-3) circle (.6);
\draw (-4,-3) node[rotate=0]{$\bar U_A$};

\draw[thick] (-7,-1) -- (-5,-1);
\filldraw[color=black, fill=white, thick](-6,-1) circle (.6);
\draw (-6,-1) node[rotate=270]{$U_A$};

\draw[thick] (-7,1) -- (-5,1);
\filldraw[color=black, fill=white, thick](-6,1) circle (.6);
\draw (-6,1) node[rotate=270]{$\bar U_A$};

\draw[thick] (-4,4) -- (-4,2);
\filldraw[color=black, fill=white, thick](-4,3) circle (.6);
\draw (-4,3) node[rotate=180]{$U_A$};

\draw[thick] (-2,4) -- (-2,2);
\filldraw[color=black, fill=white, thick](-2,3) circle (.6);
\draw (-2,3) node[rotate=180]{$\bar U_A$};
%
\draw[thick, shift={(8,0)}]  (0.16,0) arc (0:359:3.16);

\draw[thick, shift={(8,0)}] (-1,1) -- (1,1);
\filldraw[color=black, fill=white, thick, shift={(8,0)}](0,1) circle (.6);
\draw (8,1) node[rotate=90]{$U_B$};

\draw[thick, shift={(8,0)}] (-1,-1) -- (1,-1);
\filldraw[color=black, fill=white, thick, shift={(8,0)}](0,-1) circle (.6);
\draw (8,-1) node[rotate=90]{$\bar U_B$};

\draw[thick, shift={(8,0)}] (-2,-4) -- (-2,-2);
\filldraw[color=black, fill=white, thick, shift={(8,0)}](-2,-3) circle (.6);
\draw (6,-3) node[rotate=0]{$U_B$};

\draw[thick, shift={(8,0)}] (-4,-4) -- (-4,-2);
\filldraw[color=black, fill=white, thick, shift={(8,0)}](-4,-3) circle (.6);
\draw (4,-3) node[rotate=0]{$\bar U_B$};

\draw[thick, shift={(8,0)}] (-7,-1) -- (-5,-1);
\filldraw[color=black, fill=white, thick, shift={(8,0)}](-6,-1) circle (.6);
\draw (2,-1) node[rotate=270]{$U_B$};

\draw[thick, shift={(8,0)}] (-7,1) -- (-5,1);
\filldraw[color=black, fill=white, thick, shift={(8,0)}](-6,1) circle (.6);
\draw (2,1) node[rotate=270]{$\bar U_B$};

\draw[thick, shift={(8,0)}] (-4,4) -- (-4,2);
\filldraw[color=black, fill=white, thick, shift={(8,0)}](-4,3) circle (.6);
\draw (4,3) node[rotate=180]{$U_B$};

\draw[thick, shift={(8,0)}] (-2,4) -- (-2,2);
\filldraw[color=black, fill=white, thick, shift={(8,0)}](-2,3) circle (.6);
\draw (6,3) node[rotate=180]{$\bar U_B$};

%
\draw[thick, shift={(8,8)}]  (0.16,0) arc (0:359:3.16);

\draw[thick, shift={(8,8)}] (-1,1) -- (1,1);
\filldraw[color=black, fill=white, thick, shift={(8,8)}](0,1) circle (.6);

\draw[thick, shift={(8,8)}] (-1,-1) -- (1,-1);
\filldraw[color=black, fill=white, thick, shift={(8,8)}](0,-1) circle (.6);

\draw[thick, shift={(8,8)}] (-2,-4) -- (-2,-2);
\filldraw[color=black, fill=white, thick, shift={(8,8)}](-2,-3) circle (.6);

\draw[thick, shift={(8,8)}] (-4,-4) -- (-4,-2);
\filldraw[color=black, fill=white, thick, shift={(8,8)}](-4,-3) circle (.6);

\draw[thick, shift={(8,8)}] (-7,-1) -- (-5,-1);
\filldraw[color=black, fill=white, thick, shift={(8,8)}](-6,-1) circle (.6);

\draw[thick, shift={(8,8)}] (-7,1) -- (-5,1);
\filldraw[color=black, fill=white, thick, shift={(8,8)}](-6,1) circle (.6);

\draw[thick, shift={(8,8)}] (-4,4) -- (-4,2);
\filldraw[color=black, fill=white, thick, shift={(8,8)}](-4,3) circle (.6);

\draw[thick, shift={(8,8)}] (-2,4) -- (-2,2);
\filldraw[color=black, fill=white, thick, shift={(8,8)}](-2,3) circle (.6);

%
\draw[thick, shift={(0,8)}]  (0.16,0) arc (0:359:3.16);

\draw[thick, shift={(0,8)}] (-1,1) -- (1,1);
\filldraw[color=black, fill=white, thick, shift={(0,8)}](0,1) circle (.6);

\draw[thick, shift={(0,8)}] (-1,-1) -- (1,-1);
\filldraw[color=black, fill=white, thick, shift={(0,8)}](0,-1) circle (.6);

\draw[thick, shift={(0,8)}] (-2,-4) -- (-2,-2);
\filldraw[color=black, fill=white, thick, shift={(0,8)}](-2,-3) circle (.6);

\draw[thick, shift={(0,8)}] (-4,-4) -- (-4,-2);
\filldraw[color=black, fill=white, thick, shift={(0,8)}](-4,-3) circle (.6);

\draw[thick, shift={(0,8)}] (-7,-1) -- (-5,-1);
\filldraw[color=black, fill=white, thick, shift={(0,8)}](-6,-1) circle (.6);

\draw[thick, shift={(0,8)}] (-7,1) -- (-5,1);
\filldraw[color=black, fill=white, thick, shift={(0,8)}](-6,1) circle (.6);

\draw[thick, shift={(0,8)}] (-4,4) -- (-4,2);
\filldraw[color=black, fill=white, thick, shift={(0,8)}](-4,3) circle (.6);

\draw[thick, shift={(0,8)}] (-2,4) -- (-2,2);
\filldraw[color=black, fill=white, thick, shift={(0,8)}](-2,3) circle (.6);

\end{tikzpicture}
  \caption{\label{fig:PEP1} Schematic diagram for the PEPS.}
\end{figure}
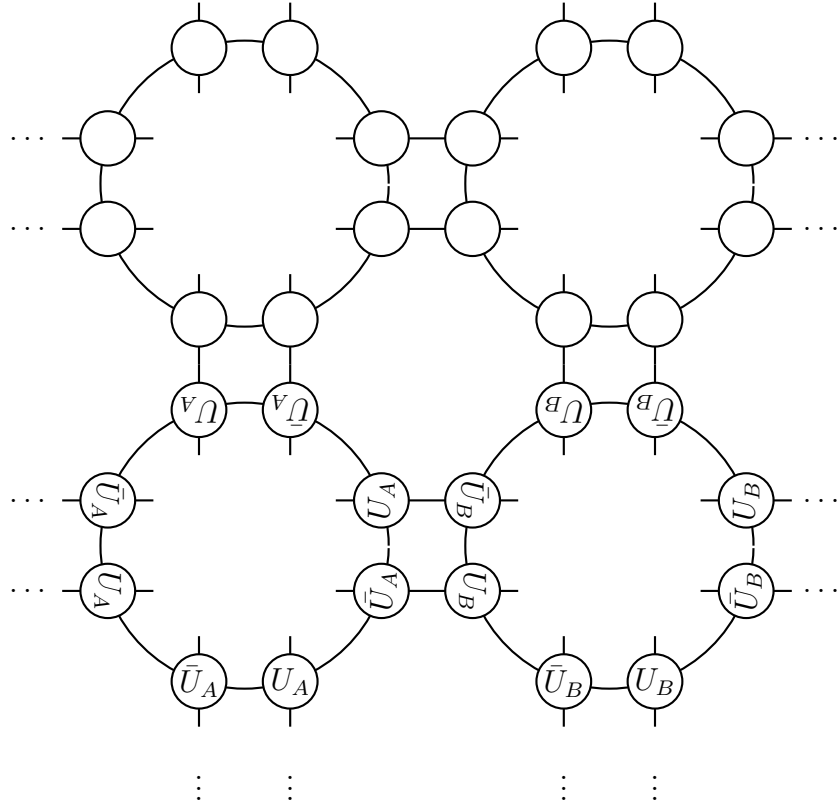
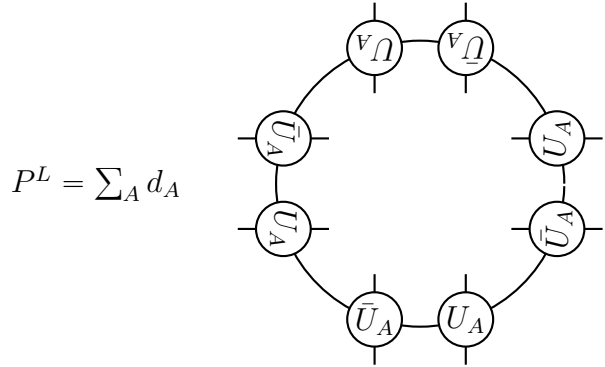
\begin{figure}[h!]
\centering
\begin{tikzpicture}[scale=.6]
\draw (-8,0) node[anchor=east]{$P^{L} = \sum_A d_{A}$};

\draw[thick]  (0.16,0) arc (0:359:3.16);

\draw[thick] (-1,1) -- (1,1);
\filldraw[color=black, fill=white, thick](0,1) circle (.6);
\draw (0,1) node[rotate=90]{$U_A$};

\draw[thick] (-1,-1) -- (1,-1);
\filldraw[color=black, fill=white, thick](0,-1) circle (.6);
\draw (0,-1) node[rotate=90]{$\bar U_A$};

\draw[thick] (-2,-4) -- (-2,-2);
\filldraw[color=black, fill=white, thick](-2,-3) circle (.6);
\draw (-2,-3) node[rotate=0]{$U_A$};

\draw[thick] (-4,-4) -- (-4,-2);
\filldraw[color=black, fill=white, thick](-4,-3) circle (.6);
\draw (-4,-3) node[rotate=0]{$\bar U_A$};

\draw[thick] (-7,-1) -- (-5,-1);
\filldraw[color=black, fill=white, thick](-6,-1) circle (.6);
\draw (-6,-1) node[rotate=270]{$U_A$};

\draw[thick] (-7,1) -- (-5,1);
\filldraw[color=black, fill=white, thick](-6,1) circle (.6);
\draw (-6,1) node[rotate=270]{$\bar U_A$};

\draw[thick] (-4,4) -- (-4,2);
\filldraw[color=black, fill=white, thick](-4,3) circle (.6);
\draw (-4,3) node[rotate=180]{$U_A$};

\draw[thick] (-2,4) -- (-2,2);
\filldraw[color=black, fill=white, thick](-2,3) circle (.6);
\draw (-2,3) node[rotate=180]{$\bar U_A$};
\end{tikzpicture}
  \caption{\label{fig:MPO1} Schematic diagram for the insertion of the MPO $P^{L}$ into a cell. Here $L=8$.}
\end{figure}

We can also insert a string corresponding to a representer of the double triangle algebra such as 
$O_A$ into the PEPS as shown in fig. \ref{fig:10}. The precise location of the 
string is irrelevant, again by the ``pull-through'' property of \cite{bultinck2017anyons}:

\begin{proposition}
We have the pull-through identity in fig. \ref{fig:10}. 
\end{proposition}
\begin{figure}[h!]
\begin{center}
  \includegraphics[width=1.1\textwidth,]{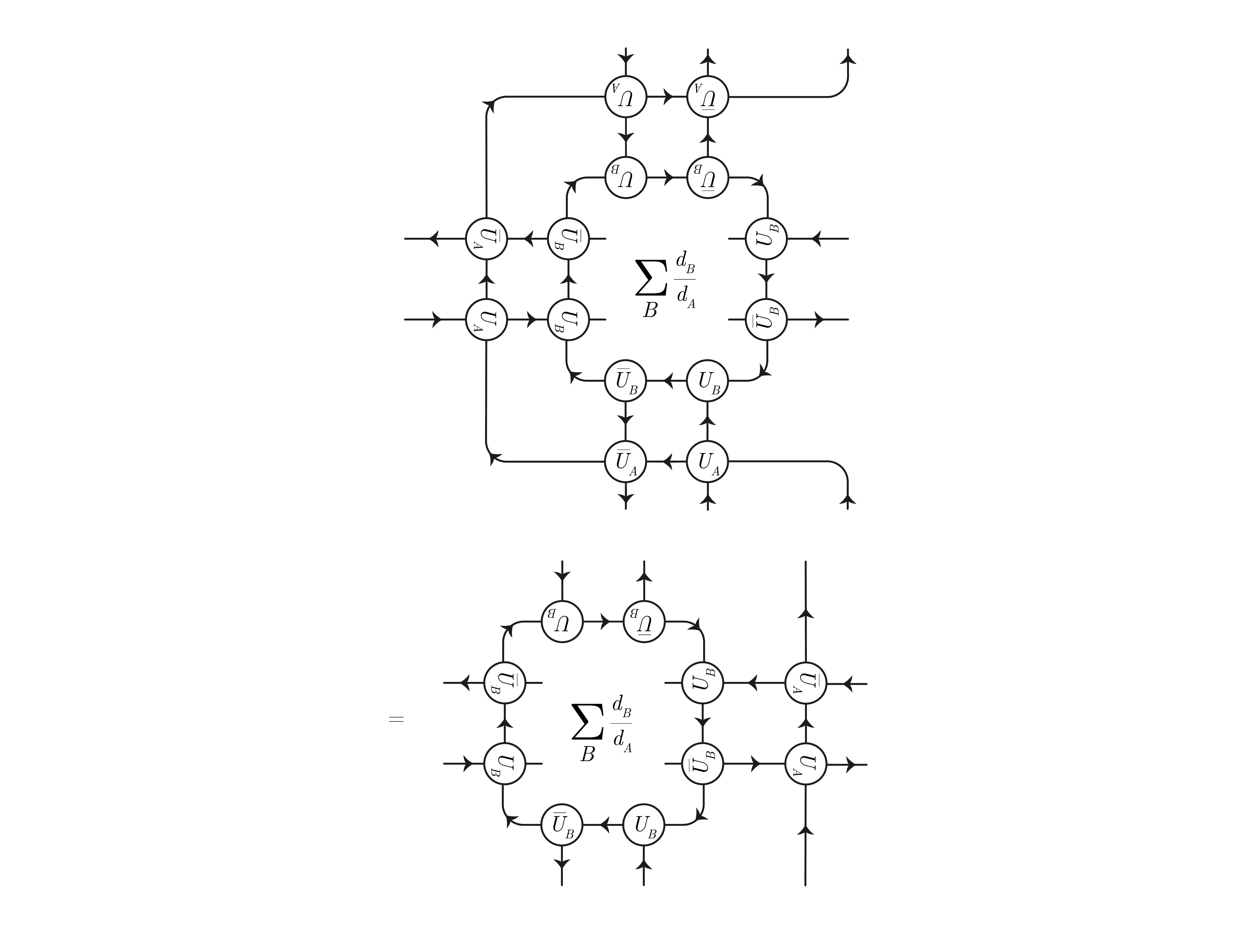}
  \end{center}
  \caption{\label{fig:10} The pull-through property described in \cite{bultinck2017anyons}.}
\end{figure}
\begin{proof}
A graphical proof is given by  \cite{bultinck2017anyons}, but it is not so clear to us how their assumptions precisely correspond to our setting. 
So we argue instead as follows. First, the zipper lemma implies that the 
pull through property is equivalent to the equality shown graphically in fig. \ref{fig:11}. 
\begin{figure}[h!]
\begin{center}
  \includegraphics[width=1.1\textwidth,]{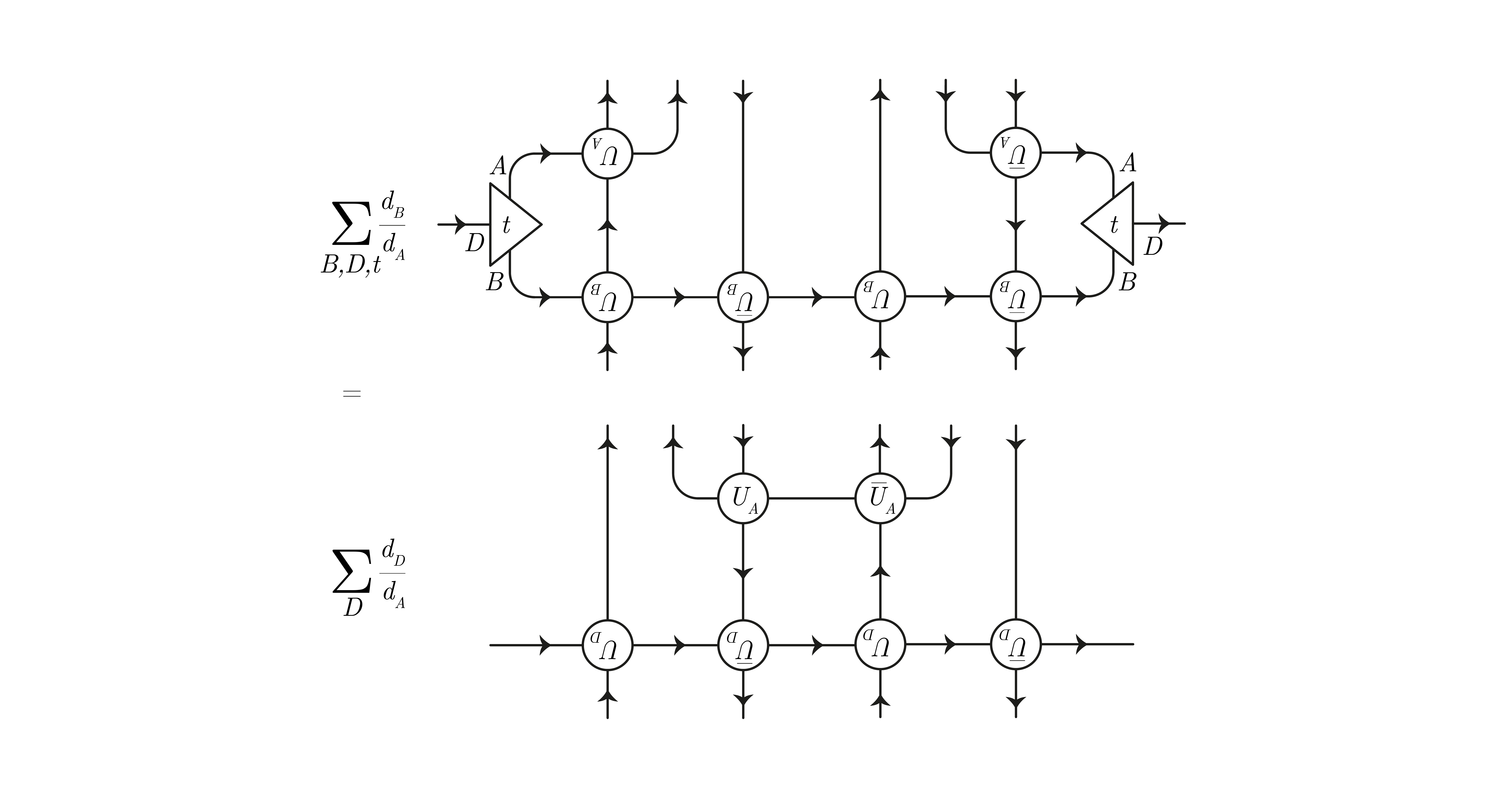}
  \end{center}
  \caption{\label{fig:11} Equivalent form of the pull through property.}
\end{figure}
Here and in the following, we put arrows on the cells to see more clearly that the tensor contractions implied by the 
connected lines are well-defined, i.e. over indices in the correct intertwiner spaces.
Then we contract suitable legs with $U_B$ and use the 
duality and unitarity of the $6j$-symbols as in lemma \ref{lem:6j}. This turns the desired equality into fig. \ref{fig:12}. 
\begin{figure}[h!]
\begin{center}
  \includegraphics[width=1.1\textwidth,]{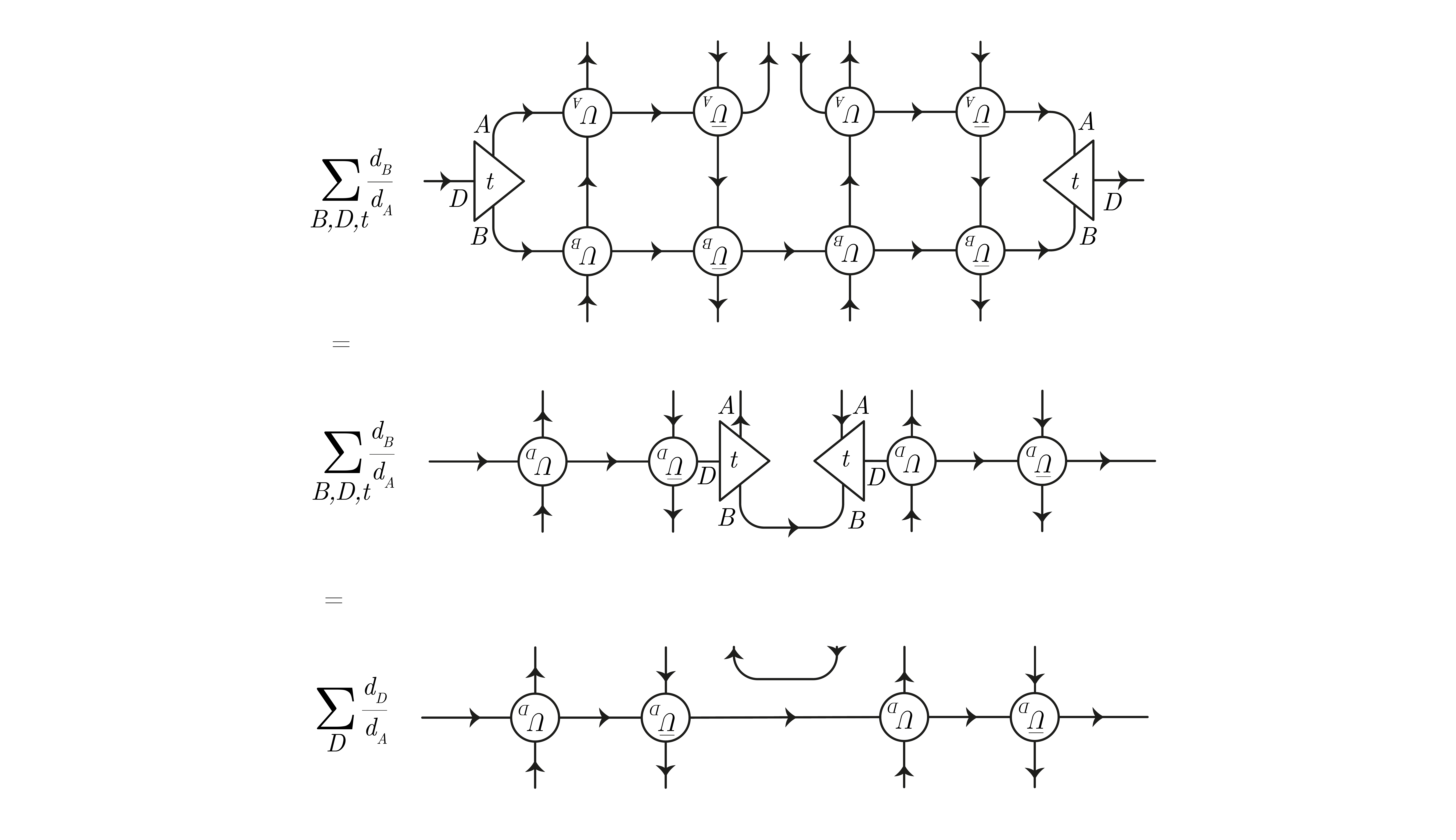}
  \end{center}
  \caption{\label{fig:12} Upper panel: Equivalent form of the upper panel in fig. \ref{fig:11} after contracting into suitable legs of a $6j$-symbol.
  Middle panel: Obtained from upper panel by zipper lemma. Lower panel: Equivalent form of the lower panel in fig. \ref{fig:11} 
  after contracting into suitable legs of a $6j$ symbol, using duality and unitarity.}
\end{figure}
Thus, we need to show the equality between the middle and lower panel in fig. \ref{fig:12}. 
This equality holds true by an orthogonality property of the zipper tensors, with 
ONBs 
{\footnotesize
\ben
\begin{split}
&t_8 \in \Hom_\cM(a_3, C a_1,), \\
&t_8' \in \Hom_\cM(a_3',C a_1), \\
&t_3 \in \Hom_\cM(a_2,B a_1),\\
& t_5' \in \Hom_\cM(a_2,A a_1),\\ 
& t_5 \in \Hom_\cM(a_2, A a_1),\\ 
&t_0 \in \Hom_\cM(C, AB), 
\end{split} 
\een
}
and 
 $a_i, a_i' \in {}_\cM X_\cN$, $A,B,C \in {}_\cM X_\cM$:
\ben
\sum_B \sum_{t_3,t_0} \frac{d_A}{d_C}
\overline{ \overline Y_{A,B}^C
\begin{pmatrix}
& t_3 & \\
t_0; & & t_8 \\
& t_5 &
\end{pmatrix}}
\overline
Y_{A,B}^C
\begin{pmatrix}
& t_3 & \\
t_0; & & t_8' \\
& t_5' &
\end{pmatrix}
= \delta_{t_5,t_5'} \delta_{t_8,t_8'} \delta_{a_3,a_3'},
\een
which follows from repeated application of Frobenius duality. This implies the statement. 
\end{proof}

Instead of inserting $P^L$ into each cell, we could also think of inserting $Q_{\mu,\nu}^L$, which also is a projection (cor. \ref{cor:triangle}). 
Thus, in a sense, we would get a different ``sector'' for every choice of $Q_{\mu,\nu}^L$, and furthermore, 
the standard choice $P^L$ proposed by \cite{bultinck2017anyons} and investigated further by 
\cite{kawahigashi2021projector,kawahigashi2020remark}
is included by putting $\mu=\nu=id$. We think that it is worthwhile investigating such constructions further. 

\bibliography{main-short}
\bibliographystyle{ieeetr}

\end{document}